\documentclass[11pt,letterpaper]{article}
\usepackage{mehdis}

\usepackage{pgfplots}
\usepackage{caption}
\usepackage{subcaption}
\usepackage{multicol}
\usepackage{neuralnetwork}

\usepackage[toc,title,page]{appendix}

\newcommand{\capacity}{\operatorname{cap}}

\usepackage{authblk}
\usepackage{tocvsec2}

\usepackage{xcolor}
\definecolor{light-gray}{gray}{0.95}
\usepackage{underscore}
\newcommand{\code}[1]{\colorbox{light-gray}{\texttt{#1}}}
\newif\ifnotes\notestrue
\ifnotes
\usepackage{color}
\definecolor{mygrey}{gray}{0.50}
\newcommand{\notename}[2]{{\textcolor{mygrey}{\footnotesi\bm{z}e{\bf (#1:} {#2}{\bf ) }}}}

\newcommand{\pnote}[1]{{\endnote{#1}}}

\else

\newcommand{\notename}[2]{{}}

\newcommand{\pnote}[1]{}

\fi

\newcommand{\secondtitle}[1]{
   \begin{center}
      \LARGE #1
   \end{center}
   \vspace{1em}  
}

\usepackage{etoolbox} 
\usepackage{listofitems} 

\tikzset{>=latex} 
\colorlet{myred}{red!80!black}
\colorlet{myblue}{blue!80!black}
\colorlet{mygreen}{green!60!black}
\colorlet{mydarkred}{myred!40!black}
\colorlet{mydarkblue}{myblue!40!black}
\colorlet{mydarkgreen}{mygreen!40!black}
\tikzstyle{node}=[very thick,circle,draw=myblue,minimum size=22,inner sep=0.5,outer sep=0.6]
\tikzstyle{connect}=[->,thick,mydarkblue,shorten >=1]
\tikzset{ 
  node 1/.style={node,mydarkgreen,draw=mygreen,fill=mygreen!25},
  node 2/.style={node,mydarkblue,draw=myblue,fill=myblue!20},
  node 3/.style={node,mydarkred,draw=myred,fill=myred!20},
}

\pgfplotsset{compat=1.18}
\begin{document}
\title{Certifying almost all quantum states with\\ few single-qubit measurements}

\author[1,2,3]{Hsin-Yuan Huang}
\author[1,4]{John Preskill}
\author[1]{Mehdi Soleimanifar}
\affil[1]{California Institute of Technology}
\affil[2]{Google Quantum AI}
\affil[3]{Massachusetts Institute of Technology}
\affil[4]{AWS Center for Quantum Computing}

\date{\today}

\maketitle

\begin{abstract}
Certifying that an $n$-qubit state~$\rho$ synthesized in the lab is close to the target state $\ket{\psi}$ is a fundamental task in quantum information science.
However, existing rigorous protocols either require deep quantum circuits or exponentially many single-qubit measurements.
In this work, we prove that almost all $n$-qubit target states $\ket{\psi}$, including those with exponential circuit complexity, can be certified from only $\mathcal{O}(n^2)$ single-qubit measurements.
This result is established by a new technique that relates certification to the mixing time of a random walk.
Our protocol has applications for benchmarking quantum systems, for optimizing quantum circuits to generate a desired target state, and for learning and verifying neural networks, tensor networks, and various other representations of quantum states using only single-qubit measurements.
We show that such verified representations can be used to efficiently predict highly non-local properties of~$\rho$ that would otherwise require an exponential number of measurements on~$\rho$.
We demonstrate these applications in numerical experiments with up to $120$ qubits, and observe advantage over existing methods such as cross-entropy benchmarking (XEB).
\end{abstract}

\thispagestyle{empty}
\clearpage
\tableofcontents
\thispagestyle{empty}
\clearpage
\pagenumbering{arabic} 
\setcounter{page}{1}

\section{Introduction}\label{sec:intro}

Our empirical knowledge of a quantum system often relies on statistical comparisons with a target model of its state.
This, for instance, can be achieved by certifying that an $n$-qubit quantum state $\rho$, which is synthesized in the lab and can be measured experimentally, is close to a target quantum state $\ket{\psi}$.
For many relevant applications, the description of this target state $\ket{\psi}$ is provided to us through a query model $\Psi$.
When queried with $x \in \{0,1\}^n$, the model $\Psi$ returns the complex amplitude $\braket{x}{\psi}$, up to an unimportant overall normalization factor. Here $\{\ket{x}: x \in \{0,1\}^n\}$ denotes a chosen orthonormal basis for the $n$ qubits.

For systems of small size $n$, models with query access to the amplitudes may simply be obtained by storing the entire description of the state $\ket{\psi}$ on a classical memory. 
For larger systems, powerful models, such as neural network quantum states or tensor networks, can be employed instead.  
These models have been the focus of many past and more recent works \cite{carleo2017solving, VerstraeteTruncationMPS, Shi2006TreeTN, pfau2020, Hibat2020RNN, carrasquilla2019reconstructing, melko2019restricted, torlai2018neural, iouchtchenko2023neural, zhao2023empirical, Torlai2020Precise} that explore their promising practical performance. 
There are also many quantum states, including some highly entangled ones, which have simple classical descriptions and are therefore natively equipped with a query access model; examples include phase states, coherent Gibbs states, GHZ states, and W states.

We can query the model $\Psi$ to obtain information about the target state $\ket{\psi}$.
We can also perform measurements on independent copies to gather information about the lab state $\rho$.
\emph{Quantum state certification} is a 
task that uses the results of the measurements to certify that the~fidelity~$\bra{\psi}\rho\ket{\psi}$ is sufficiently close to one.
To be most effective, this certification procedure should demand minimal experimental and computational resources. 
Therefore we aim for the measurements on the state $\rho$ to consist of simple single-qubit measurements.
These measurements are compatible with a wide range of experimental platforms and typically yield data that is easier to analyze computationally.
The primary question that we seek to address is: \emph{how many copies of the state $\rho$ do we need to measure in order to accurately determine if the fidelity $\bra{\psi}\rho\ket{\psi}$ is close to one or not?}

Despite the considerable previous research, it has remained open whether a certification procedure exists that (a) relies solely on a few single-qubit measurements on separate copies of a general $n$-qubit state $\rho$ and (b) can validate the overlap with a generic highly-entangled target state $\ket{\psi}$.
On the face of it, demanding both features (a) and (b) may seem contradictory.
After all, the relevant information in an entangled state is distributed non-locally among its constituent qubits, and it might appear that single-qubit measurements on $\rho$ lack the capacity to probe global properties such as the fidelity $\bra{\psi}\rho\ket{\psi}$ with a highly-entangled target state $\ket{\psi}$.
Indeed prior results either require measurements involving deep quantum circuits \cite{huang2020predicting, odonnell2016tomography, odonnell2017tomography2, haah2017sample}, need exponentially many single-qubit measurements \cite{Flammia2011DirectFidelity, Poulin2011Characterization, aolita2015reliable}, are limited to special families of target states \cite{gluza2018fidelity, takeuchi2018verification}, or lack rigorous guarantee for certification \cite{arute2019quantumsupremacy, choi2023preparing, cotler2023emergent}.
In the next section, we present our main result --- a certification procedure, featuring a surrogate for the fidelity which we call the \emph{shadow overlap}, that can certify almost all quantum states with few single-qubit measurements, achieving both (a) and (b).

\section{Our main results}\label{sec:mainResult}
In this work, we devise a simple procedure (shown in \fig{verificationprotocol} and introduced in detail in \secref{verificationProcedure}) for certifying the overlap $\bra{\psi}\rho\ket{\psi}$ between a (possibly mixed) state $\rho$ and a target pure state $\ket{\psi}$ over $n$ qubits.
This certification procedure proceeds by performing single-qubit Pauli measurements on each qubit of the state $\rho$ and outputs an estimate $\hat{\bm{\omega}}$. 
The expectation $\E[\hat{\bm{\omega}}]$, which we refer to as the \emph{shadow overlap}, satisfies the following relation with the fidelity $\bra{\psi} \rho \ket{\psi}$:
\begin{align}
 \E[\hat{\bm{\omega}}] \geq 1 - \epsilon & \,\,\,\text{implies}\,\,\, \bra{\psi}\rho\ket{\psi} \geq 1 - \tau\epsilon,\label{eq:localToGlobalFidelityMixing}\\
 \bra{\psi}\rho\ket{\psi} \geq 1 - \epsilon & \,\,\,\text{implies}\,\,\, \E[\hat{\bm{\omega}}] \geq 1 - \epsilon.\label{eq:GlobalToLocalFidelity}
\end{align}
Here, the parameter $\tau$ corresponds to the relaxation time of a Markov chain for sampling from the measurement distribution $\pi(x) : =|\braket{x}{\psi}|^2$ induced by the state $\ket{\psi}$.
This Markov chain is introduced more formally in \secref{verificationProcedure}.
When we neglect $\log(n)$ factors, the relaxation time of a Markov chain is bounded by its mixing time and relates to the number of random bit flips needed to sample from the stationary distribution.
When the relaxation time $\tau$ is bounded, the shadow overlap $\E[\hat{\bm{\omega}}]$ offers a good surrogate for the fidelity $\bra{\psi} \rho \ket{\psi}$.
Stated equivalently in the following theorem, we show that one can efficiently certify the fidelity $\bra{\psi}\rho\ket{\psi}$ for any target state $\ket{\psi}$ with a polynomial relaxation (or mixing) time.
The theorem is proved in \appref{Performanceguarantees}.

\begin{theorem}[Certification of quantum states, informal]\label{thm:verificationFastMixing}
    Given an $n$-qubit target pure state $\ket{\psi}$ with a relaxation time $\tau \geq 1$.
    There is a certification procedure that performs single-qubit Pauli measurements on $T = \mathcal{O}(\tau^2 / \epsilon^2)$ samples of an unknown $n$-qubit state $\rho$ and, with high probability, outputs \textsc{Failed} if the fidelity is low $\bra{\psi}\rho\ket{\psi} < 1 - \epsilon$ and outputs \textsc{Certified} if the fidelity is high $\bra{\psi}\rho\ket{\psi} \geq 1 - \frac{\epsilon}{2 \tau}$.

    When one allows more general single-qubit measurements on the unknown $n$-qubit state $\rho$, the sample complexity can be improved to $T = \mathcal{O}(\tau / \epsilon)$.
\end{theorem}
The certification procedure in \thmref{verificationFastMixing} uses $\mathcal{O}(T)$ queries to the model $\Psi$ of the target state~$\ket{\psi}$.
As we will explain soon, when $\tau \leq \poly(n)$ and the model $\Psi$ can be queried efficiently, this certification procedure is also \emph{computationally efficient}, an important feature made available by the single-qubit nature of our measurements.

Although not all quantum states in a given basis exhibit a polynomial relaxation time, our analysis based on an intricate study of random walks on the Boolean hypercube $\{0,1\}^n$ proves that the relaxation time $\tau$ is bounded by $\tau \leq \tau^* = \mathcal{O}(n^2)$ for almost all $n$-qubit pure states.
This result is established in \appref{HaarSpectralGap}.
Together with \thmref{verificationFastMixing}, we prove that one can certify almost all quantum states, including highly entangled states with exponential circuit complexity, from a few single-qubit measurements.

\begin{theorem}[Certification of almost all quantum states, informal]\label{thm:HaarSampleComplexity}
    For all except an exponentially small $2^{- \Omega(n)}$ fraction of $n$-qubit target pure states $\ket{\psi}$,
    there is a certification procedure that performs single-qubit measurements on $\mathcal{O}(n^2 / \epsilon)$ samples of an unknown state $\rho$ and, with high probability, outputs \textsc{Failed} if the fidelity is low $\bra{\psi}\rho\ket{\psi} < 1 - \epsilon$ and outputs \textsc{Certified} if the fidelity is high $\bra{\psi}\rho\ket{\psi} \geq 1 - \frac{\epsilon}{2 \tau^*}$, where~$\tau^* = \mathcal{O}(n^2)$.
\end{theorem}
While \thmref{HaarSampleComplexity} establishes an efficient certification procedure for generic quantum states, we can also prove a polynomial bound $\tau  \leq \poly(n)$ for a variety of structured quantum states, rendering our certification scheme efficient for such quantum states. 
Quantum phase states and the GHZ-like states both have a relaxation time $\tau = \mathcal{O}(n)$ as shown in \appref{Quantum phase states} and \appref{GHZ}. 
We also prove in \appref{gappedGroundStates} that $\tau = \mathcal{O}(n^{\kappa})$ for a family of quantum states such that the probability distribution $|\braket{x}{\psi}|^2$ matches the distribution for the ground state of a gapped sign-free $\kappa$-local Hamiltonian but the phases~$\tfrac{\braket{x}{\psi}}{|\braket{x}{\psi}|}$ can be arbitrary.

Even though exhibiting fast mixing can be considered an assumption concerning the target state, in \appref{enforcingMixing}, we also present a scheme that could be used to enforce this assumption by modifying a given query model $\Psi$ of a target state $\ket{\psi}$.
This scheme runs efficiently and has the following feature:
If the model $\Psi$ satisfies a sufficient condition for fast mixing, known as the \emph{local escape property} defined formally in \defref{localEscape}, then it remains unchanged by the enforcement procedure. 
If not, a new model $\Psi'$ is induced which satisfies the local escape property and exhibits fast mixing.
However, this comes with the potential cost of $\Psi'$ being very different from the original model $\Psi$.
We show that local escape property holds for almost all quantum states, which is the essential ingredient in our proof of \thmref{HaarSampleComplexity}.

\begin{figure}[t!]
    \centering
\includegraphics[width=1\textwidth]{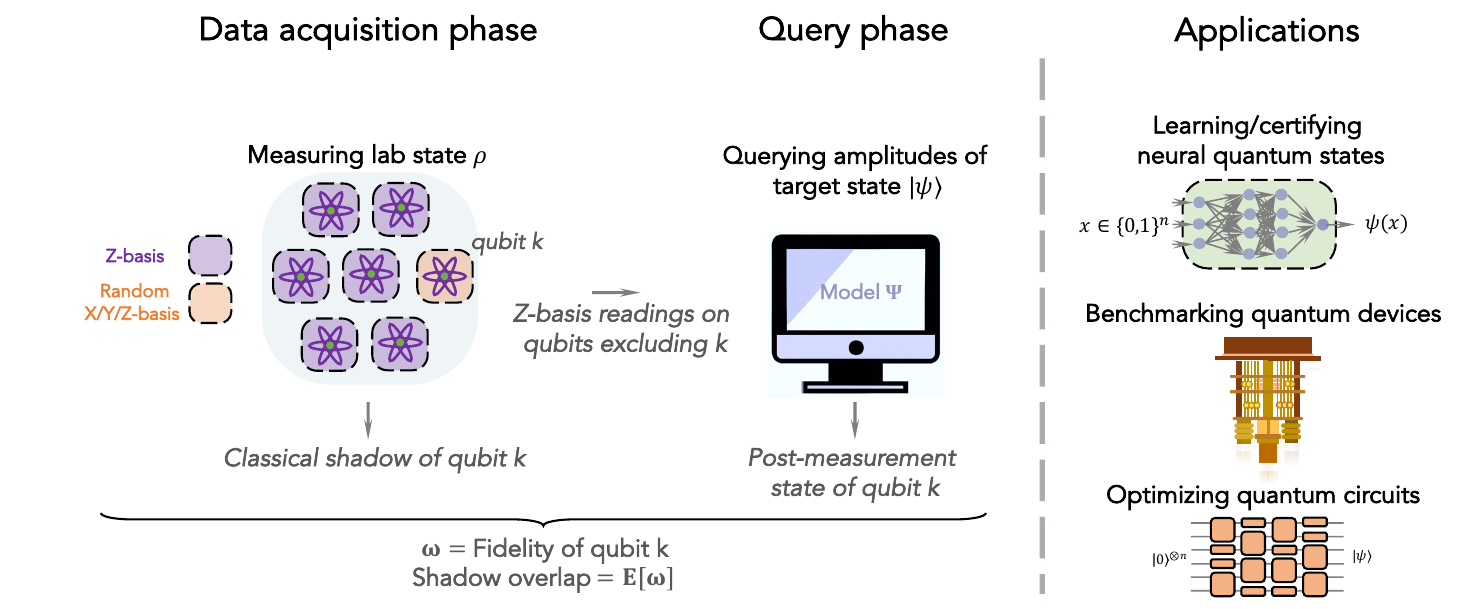}
    \caption{\textbf{Estimating the shadow overlap}. \textit{Data collection phase:} For each copy of the lab state $\rho$, a random qubit $\bm{k}$ is selected. 
    All qubits except $\bm{k}$ are measured in the $Z$ basis.
    Qubit $\bm{k}$ is measured in a random $X$, $Y$, or $Z$ basis to obtain its classical shadow.
    \textit{Query phase:}
   By querying the amplitudes of the target state $\ket{\psi}$ twice, the ideal  post-measurement state $\ket{\psi_{\bm{k}, \bm{z}}}$ of qubit $\bm{k}$ is found. 
   Using the classical shadow of qubit $\bm{k}$ from the lab state, its overlap $\bm{\omega}$ with $\ket{\psi_{\bm{k}, \bm{z}}}$ is evaluated. 
   Finally, the shadow overlap  $\E[\bm{\omega}]$ is estimated by averaging $\bm{\omega}$ across all copies. }
    \label{fig:verificationprotocol}
\end{figure}

This certification procedure has much broader applications that we also investigate in this work.
Here we briefly summarize these results, postponing a more detailed discussion until \secref{applications}.

\vspace{0.5em}

\noindent \textbf{ML tomography of quantum states:}
In \secref{LearningMLModels}, we demonstrate that the shadow overlap provides a theoretically backed yet practically feasible procedure for learning a machine learning model of a quantum state or certifying the fidelity of an already trained model. 

\vspace{0.5em}

\noindent \textbf{Near-term benchmarking of quantum devices:}
We show in \secref{BenchmarkingQuantumDevices} that our certification procedure offers a flexible method for benchmarking noisy quantum devices with limited measurements and gate controls, producing results closely mirroring the fidelity.

\vspace{0.5em}

\noindent \textbf{Optimizing quantum circuits for state preparation:} In \secref{Optimizingcircuits}, we show that the shadow overlap exhibits favorable properties for training quantum circuits.
Unlike fidelity, which faces the barren plateau phenomenon \cite{mcclean2018barren, cerezo2021cost}, shadow overlap acts similarly to the Hamming distance when the target state has no global correlations and its amplitudes are spread across the hypercube $\{0,1\}^n$. That is, the shadow overlap improves steadily as we increase the number of state preparation steps; in contrast, the fidelity remains stuck near zero until the preparation circuit is sufficiently large, and then increases abruptly.

\vspace{0.5em}

Before proceeding with a detailed explanation of each of these results, we introduce our certification scheme, the overall idea of why it works, and its high-level analysis.

\section{Certification procedure}\label{sec:verificationProcedure}

In its simplest form, depicted in \fig{verificationprotocol}, our certification procedure consists of two steps.
In the first step, we acquire a single copy of the state $\rho$, 
and randomly choose one of its $n$ qubits, denoted by $\bm{k}\in \{1,\dots,n\}$. 
We then measure all the qubits of $\rho$ except for qubit $\bm{k}$ in the Pauli $Z$-basis. 
We denote the measurement outcomes collectively by $\bm{z}\in\{0,1\}^{n-1}$.
Following these measurements, we select a Pauli $X$, $Y$, or $Z$-basis measurement uniformly at random and measure the $\bm{k}$th qubit of $\rho$ in that basis. 
The post-measurement state of the $\bm{k}$th qubit is denoted by $\ket{\bm{s}}$.

Having obtained measurement results from the state $\rho$, the certification protocol then moves to the second phase where we query a model $\Psi$ that represents the state $\ket{\psi} = \sum_{x \in \{0,1\}^n} \psi(x) \ket{x}$ by giving us access to the amplitudes $\psi(x)$ of the state in the computational basis.
That is, upon receiving an input string $x \in \{0,1\}^n$, the model $\Psi$ returns a complex number $\Psi(x)$. The numbers $\{\Psi(x)\}$ are assumed to be proportional to the amplitudes $\{\psi(x)\}$, but we do not require the outputs to be normalized.
This extends the applicability of our results because computing the normalization constant $\sum_{x\in \{0,1\}^n} |\Psi(x)|^2$ is in general intractable.

In the query phase of the certification procedure, the model $\Psi$ is queried twice with binary strings $\bm{z}^{(0)}$, $\bm{z}^{(1)}$.
Here the string input $\bm{z}^{(a)}$ equals $a\in\{0,1\}$ on its $\bm{k}$th bit and matches the measurement outcome $\bm{z}\in \{0,1\}^{n-1}$ from the first step of the protocol on the remaining $n{-}1$ bits.
We use these queries to compute the single-qubit state $\ket{\Psi_{\bm{k,z}}}$ defined by 
\begin{align}
\ket{\Psi_{\bm{k,z}}}:=\frac{\Psi(\bm{z}^{(0)})\cdot \ket{0}+\Psi(\bm{z}^{(1)})\cdot \ket{1}}{\sqrt{|\Psi(\bm{z}^{(0)})|^2+|\Psi(\bm{z}^{(1)})|^2}}. \label{eq:querystate-inline}
\end{align}
Finally, the measurement and the query data are used to compute the \emph{local overlap}
\begin{align}        \bm{\omega}:=\bra{\Psi_{\bm{k,z}}}\cdot \left(3\ketbra{\bm{s}}{\bm{s}}-\iden \right)\cdot \ket{\Psi_{\bm{k,z}}}.\label{eq:overlapinprotocol-inline}
    \end{align}
If both queries $\Psi(\bm{z}^{(0)}), \Psi(\bm{z}^{(1)}) = 0$, we report $\bm{\omega} = 0$.
We then repeat this two-step protocol for a total of $T$ times on independent copies of the state $\rho$ to obtain overlaps $\bm{\omega_1},\dots,\bm{\omega_T}$. 
We report the empirical average $\hat{\bm{\omega}} = \frac{1}{T}\sum_{t=1}^T \bm{\omega_t}$ as our estimated shadow overlap between $\rho$ and the state represented by model $\Psi$.
The complete certification protocol is summarized in~\protoref{certification}. 

\begin{table}[ht!]
\begin{protocol}{Certifying that states $\rho$ and $\ketbra{\psi}{\psi}$ are close to each other using shadow overlaps}{proto:certification}
\textbf{Input:} $T$ samples of an unknown state $\rho$, a model $\Psi$ that gives query access to the amplitudes of $\ket{\psi}$, an error $0\leq \epsilon < 1$, and the relaxation time $\tau$ associated with the Markov chain from $|\braket{x}{\psi}|^2$.

\sbline

\textbf{Goal:} Certify that the overlap $\bra{\psi}\rho\ket{\psi} \geq 1-\epsilon$.\\
\quad\quad\quad  If the fidelity is low $\bra{\psi}\rho\ket{\psi} < 1 - \epsilon$, output \textsc{Failed} with high probability.\\
\quad\quad\quad  If the fidelity is high $\bra{\psi}\rho\ket{\psi} \geq 1 - \frac{\epsilon}{2 \tau}$, output \textsc{Certified} with high probability.

\sbline

\textbf{Procedure:}
 \begin{enumerate}
   \item Select $\bm{k}\in \{1,\dots,n\}$ uniformly at random.
    \item Perform single qubit $Z$-basis measurements on all but the $\bm{k}$'th qubit of $\rho$. Denote the measurement outcomes collectively by $\bm{z}\in\{0,1\}^{n-1}$.
    \item Choose an $X$, $Y$, or $Z$-basis measurement uniformly at random and measure the $\bm{k}$'th qubit of $\rho$ in that basis. Denote the post-measurement state of the $\bm{k}$'th qubit by $\ket{\bm{s}}$.
    \item Query the model $\Psi$ twice to obtain the normalized state 
    \begin{align}
        \ket{\Psi_{\bm{k,z}}}:=\frac{\Psi(\bm{z}^{(0)})\cdot \ket{0}+\Psi(\bm{z}^{(1)})\cdot \ket{1}}{\sqrt{|\Psi(\bm{z}^{(0)})|^2+|\Psi(\bm{z}^{(1)})|^2}} \label{eq:querystate}
    \end{align}
    where $\bm{z}^{(a)}$ is a binary string that equals $a\in\{0,1\}$ on its $\bm{k}$'th bit and equals to $\bm{z}\in \{0,1\}^{n-1}$ on the remaining $n-1$ bits. 
    \item Compute the overlap
    \begin{align}        \bm{\omega}:=\bra{\Psi_{\bm{k,z}}}\cdot \left(3\ketbra{\bm{s}}{\bm{s}}-\iden \right)\cdot \ket{\Psi_{\bm{k,z}}}\label{eq:overlapinprotocol}
    \end{align}
    \item Repeat steps 1. to 5. for $T$ times to obtain overlaps $\bm{\omega_1},\dots,\bm{\omega_T}$. Report the estimated shadow overlap $\hat{\bm{\omega}} := \frac{1}{T}\sum_{t=1}^T \bm{\omega_t}$.
     \item If the estimated shadow overlap $\hat{\bm{\omega}} \geq 1 - \frac{3\epsilon}{4\tau}$, output \textsc{Certified}. Otherwise, output \textsc{Failed}.
    \end{enumerate}
\end{protocol}
\end{table}

\subsection{Overview of the analysis}\label{sec:whyItWorks}
We discuss the performance of this certification procedure in depth in \appref{PerformanceguaranteesGeneral}.
Here we explain the high-level idea behind why our protocol works. 
Consider a modified version of this test where instead of randomized Pauli measurements, the qubit $\bm{k}$ is measured in the orthogonal basis $\{ \ketbra{\Psi_{\bm{k,z}}}{\Psi_{\bm{k,z}}}, \iden -\ketbra{\Psi_{\bm{k,z}}}{\Psi_{\bm{k,z}}}\}$ where the single-qubit state $\ket{\Psi_{\bm{k,z}}}$ is the state specified in Equation \eqref{eq:querystate-inline}.
Upon measuring the outcome $\ketbra{\Psi_{\bm{k,z}}}{\Psi_{\bm{k,z}}}$, we output value $\bm{\omega} = 1$, and otherwise the measurement returns $\bm{\omega} = 0$. 
First, suppose $\rho = \ketbra{\psi}{\psi}$ and assume that after measuring $n{-}1$ qubits in the $Z$-basis, we obtain the outcome $\bm{z}$.
In this case, the post-measurement state of qubit $\bm{k}$ equals $\ketbra{\Psi_{\bm{k,z}}}{\Psi_{\bm{k,z}}}$. 
Hence, all the local overlaps measured in this manner are $\bm{\omega} = 1$. 
This means that the estimated shadow overlap is
\begin{equation}
\hat{\bm{\omega}} = \frac{1}{T}\sum_{t=1}^T \bm{\omega_t} = \bra{\psi}\rho\ket{\psi} = 1.
\end{equation}

In \protoref{certification}, however, we do not measure qubit $\bm{k}$ in this orthogonal basis; instead we measure in a randomly chosen Pauli-operator basis, and compute the local overlap $\bm{\omega}$ using expression \eqref{eq:overlapinprotocol-inline}. Fortunately, as briefly reviewed in \appref{randomizedPauliMeasurements}, the expression $\left(3\ketbra{\bm{s}}{\bm{s}}-\iden \right)$ matches the measured state $\ketbra{\Psi_{\bm{k,z}}}{\Psi_{\bm{k,z}}}$ of qubit $\bm{k}$ when averaged over the choice of Pauli basis and the measurement outcome. This feature is the key observation underlying the ``classical shadow'' protocol for learning properties of a quantum state \cite{huang2020predicting}; it ensures that if the lab state is $\rho= |\psi\rangle\langle \psi|$, then the expectation value of $\omega$ is 1. Moreover, the empirical average $\hat \omega$ will be close to 1 if $\rho$ is close to $|\psi\rangle$ and the number of samples $T$ is sufficiently large. 
In summary, this means that if the lab state is close to the target state, our protocol will successfully certify that the fidelity $\langle\psi|\rho|\psi\rangle$ is close to 1. 

But can the protocol be fooled into certifying a lab state that does not have high fidelity with the target state? To address this question, we observe that the expectation of $\omega$ is  $\Tr[L \rho]$ for an observable $L$ that can be constructed by querying the model $\Psi$ of the state $\ket{\psi}$. 
Ideally, we would want this observable to be the projector $\ketbra{\psi}{\psi}$ onto the state $\ket{\psi}$.
What we show instead in \appref{PerformanceguaranteesGeneral} is that this observable satisfies $L \ket{\psi} = \ket{\psi}$ and $\bra{\psi^{\perp}}L\ket{\psi^{\perp}} \leq 1- \frac{1}{\tau}$ for any state $\ket{\psi^{\perp}}$ orthogonal to $\ket{\psi}$,
where $\tau\geq 1$ is a parameter that depends on $\ket{\psi}$. 
In this sense, the observable $L$ forms an approximate projector onto the target state $\ket{\psi}$. As a result, we can distinguish a lab state $\rho$ that has high fidelity with $|\psi\rangle$ from a lab state that has low fidelity with $|\psi\rangle$ by measuring $T=\mathcal{O}(\tau^2)$ samples of $\rho$.

An empowering fact, established in \appref{PerformanceguaranteesGeneral}, is that the observable $L$ has the same eigenvalues as the (normalized) transition matrix $P$ of a suitably-defined random walk (or Markov chain) on the $n$-dimensional hypercube $\{0,1\}^n$. Thus $1/\tau$ is the eigenvalue gap of $P$, the difference between its largest and second largest eigenvalue; correspondingly, $\tau$ is the relaxation time of the Markov chain. This is very useful, because we can draw on the extensive literature concerning relaxation times of Markov chains to infer upper bounds on $\tau$ and hence on the sample complexity of our certification protocol.
We emphasize that this random walk defined by $P$ is not itself part of the protocol. Rather, it is merely used in the analysis of the performance of the protocol.

The transition matrix $P$ is determined by the measurement distribution $\pi(x)= |\langle x|\psi\rangle|^2$ which is sampled when the state $|\psi\rangle$ is measured in the computational basis. Specifically, for a state $\ket{\psi}=\sum_{x\in\{0,1\}^n} \sqrt{\pi(x)} e^{i\phi(x)} \ket{x}$, the corresponding walk transitions from vertex $x\in\{0,1\}^n$ to vertex $y \in \{0,1\}^n$ with probability 
\begin{align}
    P(x,y)=\begin{cases} 
       \frac{1}{n}\cdot \frac{\pi(y)}{\pi(x)+\pi(y)} & x\sim y, \\
       \frac{1}{n}\cdot \sum_{x':x'\sim x}\frac{\pi(x)}{\pi(x)+\pi(x')} & x=y,\\
       0 & \text{otherwise},
       \end{cases}\label{eq:TransitionMatrixIntro}
\end{align}
where two vertices $x$ and $y$ are connected (denoted $x{\sim}y$) when they differ in exactly $1$ bit.
The walk is designed such that in its unique stationary distribution vertex $x$ is occupied with probability $\pi(x)$. When $\pi(x) = \frac{1}{2^n}$ for all $x\in \{0, 1\}^n$, this transition matrix defines a lazy random walk on the Boolean hypercube $\{0, 1\}^n$, which remains at vertex $x$ with probability $1/2$ and moves to one of its $n$ neighboring vertices, chosen equiprobably, with probability $1/2$.

We can leverage results concerning the relaxation times of random walks to analyze the performance of our certification test. 
We consider various families of quantum states for which the relaxation time is $\tau\leq \poly(n)$, and therefore, our certification protocol is efficient as well. 
This includes generic quantum states drawn from the Haar measure in \appref{HaarSpectralGap}, as well as various structured entangled states such as quantum phase states in \appref{Quantum phase states}, ground states in \appref{gappedGroundStates}, and GHZ-like states in \appref{GHZ}.

Our analysis of the relaxation time for Haar random $n$-qubit states in \appref{HaarSpectralGap} draws on the concept of multi-commodity flows. 
In this framework, we distribute a unit flow from each vertex $x$ to another vertex $y$, dividing it across multiple paths such that no edge is congested. 
Stated more formally, given a set of simple directed paths connecting $x$ to $y$ denoted by $\mathcal{P}_{xy}$, a multi-commodity flow is a function $f: \cup_{x\neq y} \mathcal{P}_{xy} \mapsto \mathbb{R}$ such that $\sum_{p \in \mathcal{P}_{xy}}f(p)=1$ for all two distinct vertices $x\neq y$.
The resistance $R(f)$ of a flow $f$ is defined by
\begin{align}
    R(f):= \max_{e} \frac{1}{Q(e)} \sum_{x,y} \sum_{p\in \mathcal{P}_{xy}: p \ni e} \pi(x)\pi(y) f(p) |p|,\label{eq:flowDef}
\end{align}
where the weight $Q(e)=\pi(e^{+})P(e^{+},e^{-})$ of an edge $e=(e^+,e^-)$ is the probability of the transition $(e^+,e^-)$ occurring in the random walk. 
It is well-known that the relaxation time of a Markov chain is bounded by $\tau \leq R(f)$ for any flow $f$ \cite{sinclair1992improved}.

The measurement distribution of random quantum states exhibits probabilities $\pi(x)$ that can vary significantly, being either excessively small or large. 
As outlined in the expression \eqref{eq:flowDef}, this can lead to a large resistance $R(f)$.
To get around this and find a tighter upper bound on the relaxation time, the flow $f$ needs to avoid such congested vertices. 
We achieve this using the concept of \emph{local escape property} introduced in \lemref{localEscapeProperty} and \defref{localEscape}.
We show that this property holds with high probability in random states and allows us to spread the flow from any vertex $x$ to its neighbors using edges within some constant Hamming distance of $x$ while avoiding congested vertices.
A similar approach has been followed before in \cite{McDiarmidHypercube}, where instead of congested vertices, some of the edges are removed. 

\section{Applications}\label{sec:applications}

We stated in \thmref{HaarSampleComplexity} that almost all quantum states can be certified using shadow overlaps, which can be reliably estimated with few single-qubit measurements. 
In what follows, we give an overview of various interesting applications of the shadow overlap formalism.

\subsection{Neural network quantum state tomography}\label{sec:LearningMLModels}

To facilitate learning and simulation of quantum systems, we desire classical models that are expressive enough to capture essential features of intricate quantum states, which are also well-suited for predicting various properties of the systems.
A rich class of such models grant us direct access to the amplitudes of quantum states. 
More precisely, given an $n$-qubit quantum state  $\ket{\psi}=\sum_{x\in\{0,1\}^n} \psi(x)\ket{x}$ in a fixed basis, such models can provide us with \emph{query access}: the ability to compute the amplitudes $\psi(x) \in \mathbb{C}$ up to an overall normalization constant.

A family of models that provide query access are machine learning (ML) models of quantum states based on neural networks or tensor networks such as those considered in many prior works  \cite{carleo2017solving, pfau2020, Hibat2020RNN, carrasquilla2019reconstructing, melko2019restricted, torlai2018neural, iouchtchenko2023neural, zhao2023empirical, Torlai2020Precise, NNrepresentationTNSharir, Carleo2023Tensor}.
Neural networks with $\poly(n)$-bounded depth and width can compute the amplitudes $\psi(x)$ of the represented $n$-qubit state efficiently. 
Tensor networks that admit an efficient contraction method, such as matrix product states \cite{VerstraeteTruncationMPS} and tree tensor networks \cite{Shi2006TreeTN}, also yield efficient query access to the amplitudes. 

\vspace{0.5em}

\noindent \textbf{Learning ML models via hypothesis selection:} 
Our certification scheme yields an algorithm for learning ML models of quantum states with rigorous sample complexity guarantees. 
This is achieved using \emph{learning by hypothesis selection}, which can generally be applied to a set of models $\{\Psi_1, \dots, \Psi_M\}$ each describing an $n$-qubit state $\ket{\psi_i}$, for $i \in [M]$.
Our objective is to use the measurement data obtained from identical copies of a state $\rho$ and learn a model $\Psi_i$ among $i\in [M]$ which achieves the highest overlap $\bra{\psi_i}\rho\ket{\psi_i}$.
This approach to learning is relevant in applications where we either naturally have a set of $M$ hypotheses (e.g., from different theories describing the physics of a quantum system) or where we can obtain such a discrete set by casting a covering net (or carrying out some form of coarse-graining) over a larger and more expressive family of~models.

We show in \appref{learningHypothesisselection} that assuming the fast mixing condition for the set of models $\{\Psi_1,\cdots, \Psi_M\}$, we can use the shadow overlap to learn a model that achieves a high fidelity with the lab state  $\rho$ using $\mathcal{O}(\log M)$ copies of $\rho$. 
In \appref{learningHypothesisselection}, we give a concrete application of this scheme for learning a feedforward neural network representation of a quantum state.
We show that the sample complexity of this problem scales as $\widetilde{\mathcal{O}}\left(n L^3 W^3 s^{2L}\right)$ for a network of depth $L$, width $W$, and spectral norm~$s$ that takes $n$-bit strings as input.
In \appref{LearningNoisyGS}, we also discuss another application of this learning algorithm in the context of gapped ground states.

\begin{figure}[t!]
    \centering    \includegraphics[width=0.95\textwidth]{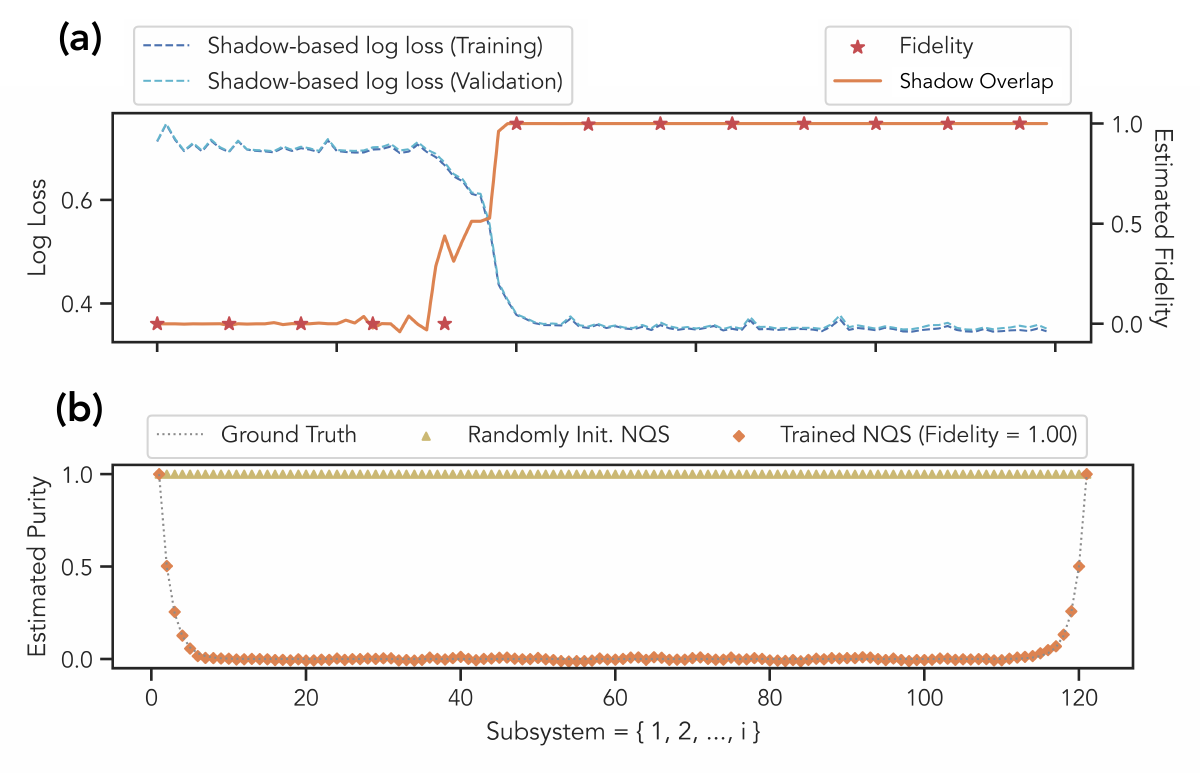}
    \captionsetup{justification=raggedright,singlelinecheck=false}
    \caption{\textbf{Neural network quantum state tomography: training and certifying a neural quantum state with the shadow overlap.} \textbf{(a)} A dual-input neural network is trained to learn a quantum phase state  \eqref{eq:randomPhaseStateML} with random phases $\phi(x)$ on $n=120$ qubits using single qubit measurements.
    A shadow-based loss function trains the model on $50,000$ measurement data acquired as outlined in \protoref{certification}.
    The model is then certified using fidelity and shadow overlap on a separate data set of size $10,000$.
    \textbf{(b)} The trained neural quantum state is used to estimate the subsystem purity of the random phase state, exhibiting a high degree of entanglement compared to a randomly initialized neural quantum state.}
    \label{fig:NQSTraining}
\end{figure}

\vspace{0.5em}

\noindent \textbf{Training neural quantum states with shadow overlap:} Although hypothesis selection provides a learning scheme with a rigorous sample complexity, the run time of this algorithm scales linearly with the number of models $M$, rendering it inefficient for many applications where $M$ grows exponentially with the number of qubits $n$.
In practice, though, as shown in \fig{NQSTraining} and detailed in \appref{NumericalTraining}, we can use the shadow overlap along with the stochastic gradient descent (SGD) to efficiently train and certify an ML model of a quantum state.

To this end, we consider training a neural network representation of an $n$-qubit state $\ket{\psi}$.
Conventionally, such neural quantum states take as input an $n$-bit string $x$ and directly output a complex value proportional to the amplitude $\braket{x}{\psi}$. 
To achieve an improved performance, we instead train a \emph{dual-input} neural network that admits two inputs $x_0, x_1 \in \{0,1\}^n$ which differ only in one bit.
This neural network computes $\frac{\braket{x_0}{\psi}}{\braket{x_1}{\psi}}$ as its output.
We will see in \appref{NumericalTraining} that $n$ applications of this neural network architecture allow us to compute the amplitude $\braket{x}{\psi}$ for a given~$x$.

The dual-input neural quantum states can be trained using a shadow-based log loss, leveraging data acquired by single-qubit measurements as prescribed in \protoref{certification}.
The log loss is minimized via stochastic gradient descent.
\fig{NQSTraining} shows an application of this scheme to learning highly entangled phase states
\begin{align}
    \ket{\psi} = \frac{1}{\sqrt{2^n}} \sum_{x\in\{0,1\}^n} e^{i \phi(x)} \ket{x}\label{eq:randomPhaseStateML}
\end{align}
with random binary phases $\phi(x)$ on $n = 120$ qubits, and training data consist of tuples $(x_0, x_1, |\phi(x_0) - \phi(x_1)|)$ which represent the phase difference between two adjacent strings $x_0$ and $x_1$.
Such quantum states have exponentially large circuit complexity \cite{Brandao2021ComplexityGrowth}, are indistinguishable from Haar-random states with polynomially many copies with high probability \cite{ji2018pseudorandom, brakerski2019pseudo}, and exhibit volume-law scaling of entanglement \cite{Aaronson2024Pseudoentanglement}. 
The findings reported in \fig{NQSTraining} indicate that beyond a certain training threshold, the model attains a fidelity of $1.00$ with the target state. 
As explained next, this performance can also be certified using the shadow overlap, as an efficient alternative to the fidelity.

\vspace{0.5em}

\noindent \textbf{Certifying ML models:} One drawback of machine learning models for quantum states is that their training usually relies on heuristic algorithms.
The absence of performance guarantees highlights the need for certification procedures capable of efficiently verifying the accuracy of the trained models.
The result of \thmref{verificationFastMixing} can be restated in terms of certifying the overlap between an $n$-qubit state $\ket{\psi}$ and its trained ML model with a relaxation time~$\tau$.
This is achieved using single-qubit Pauli measurements performed independently on $\mathcal{O}(\tau^2/\epsilon^2)$ copies of $\ket{\psi}$ along with two queries to the trained ML model per each copy of $\ket{\psi}$.

\fig{NQSTraining} shows a numerical implementation of this certification procedure for a dual-input neural network representation of a 120-qubit random phase state introduced before.
After training the neural net with $50,000$ measurements using shadow-overlap-based stochastic gradient decent, we estimate and compare the shadow overlap of the resulting model with its fidelity.
We observe that the predicted shadow overlap closely mirrors the fidelity, serving as an effective proxy. 

\vspace{0.5em}

\noindent \textbf{Estimating sparse observables:} The certified ML models of quantum states can be employed to statistically estimate many properties of interest \cite{iouchtchenko2023neural, Torlai2020Precise} if in addition to query access, we assume the models are also equipped with \emph{sampling access:} the ability to sample from the measurement distribution corresponding to~$|\psi(x)|^2:= |\braket{x}{\psi}|^2$.
The sampling access can be obtained in various ways.
Once the lab state $\rho$ has been certified to have large overlap with the target state $|\psi\rangle$, we may obtain sampling access to $|\psi\rangle$ by measuring $\rho$ in the computational basis. 
Alternatively, we can use Markov chain sampling, running the random walk defined in \eqref{eq:TransitionMatrixIntro} for a number of steps given by mixing time and then sampling from the walk. 
This procedure is closely related to the Metropolis-Hastings algorithm conventionally used in ML applications. Another option is using autoregressive methods to obtain direct sampling access~\cite{sharir2020deep}. 
We show in \appref{Estimatingsparseobservables} how to apply a verified ML model of a quantum state with query and sampling access to estimate the expectation value of any sparse observable $G$, such as the energy of a local Hamiltonian, or highly non-local properties such as R\'enyi entanglement entropies, up to an error $\epsilon$ with a number of samples that scales as $T = \mathcal{O}\left( \bra{\psi}G^2\ket{\psi}/\epsilon^2\right)$.
When no certified ML model is available, estimating certain non-linear observables such as the subsystem purity $\Tr(\rho_A^2$) requires a number of samples exponential in the size of the subsystem $A$; e.g., see \cite{chen2022exponential} for an exponential lower bound that applies to any single-copy measurements, and \cite{huang2020predicting} for an upper bound via the classical shadow formalism.
However, as shown in \appref{Estimatingsparseobservables}, the same task can be conducted using a verified ML model with a sample complexity $\mathcal{O}\left(1/\epsilon^2\right)$, independent of the system size.

In \fig{NQSTraining}, we demonstrate this feature with a numerical experiment on the trained neural network representation of the random phase state in equation \eqref{eq:randomPhaseStateML}.
The purity $\Tr(\rho_A^2$) of the phase state is estimated for subsystems of size $|A| \in \{1,\dots, 120\}$, confirming that the state of the subsystem is close to maximally mixed for sufficiently large subsets $A$.

\begin{figure}
    \centering
    \includegraphics[width=\textwidth]{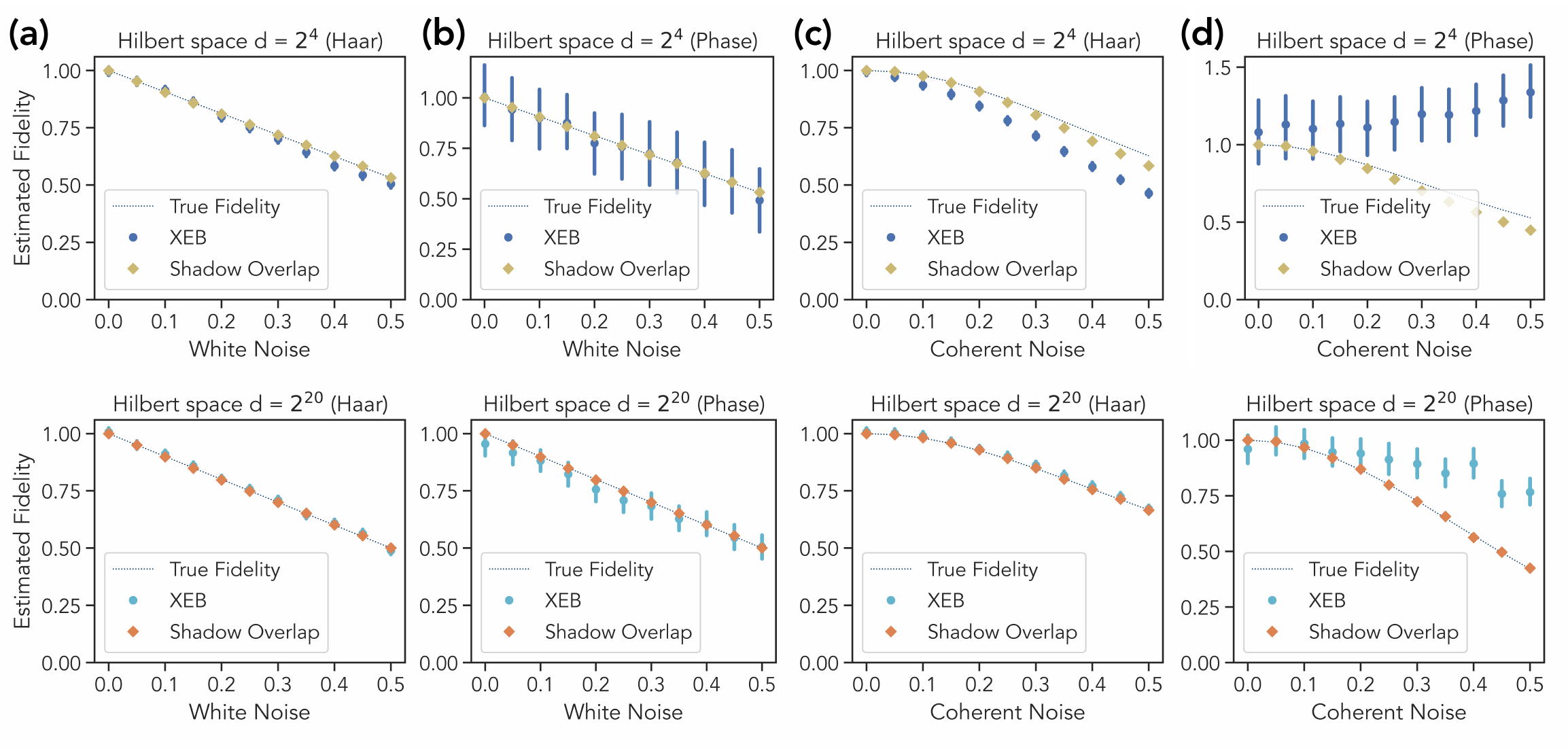}
\captionsetup{justification=raggedright,singlelinecheck=false}
    \caption{\textbf{Benchmarking with the shadow overlap.} The performance of the (normalized) shadow overlap, fidelity, and cross entropy benchmark (XEB) are compared in benchmarking noisy quantum states on 4 and 20 qubits for \textbf{(a)} a Haar random state subject to white noise, \textbf{(b)} a structured state, specifically a random phase state (generated from random product states) subjected to white noise, \textbf{(c)} a Haar random state with coherent noise, and \textbf{(d)} a random phase state with coherent noise.
    The error bars indicate statistical measurement errors, with shadow overlap displaying notably lower variance than XEB.}
    \label{fig:plotsBenchmarking}
\end{figure}

\subsection{Benchmarking quantum devices}\label{sec:BenchmarkingQuantumDevices}
Certifying the fidelity between a state $\rho$ prepared using a quantum device and a known quantum state $\ket{\psi}$ offers a rigorous approach for benchmarking quantum machines. 
However, the exponential resources and the high level of control needed for estimating fidelity limit the applicability of this approach in practice.
To address these challenges, a number of studies have proposed and deployed other statistical quantities that act as a form of \emph{proxy} for the fidelity in the situations often encountered practically \cite{boixo2018characterizing, arute2019quantumsupremacy, mark2022benchmarking, choi2023preparing}.
For a proxy of fidelity to be the most informative, one may ask for features such as (1) minimal hardware requirements and easy statistical and computational evaluation, (2) close tracking of fidelity, and (3) being equipped with rigorous bounds.
The shadow overlap meets these requirements by providing a provable lower bound on the fidelity $\bra{\psi}\rho\ket{\psi}$ as stated in \thmref{verificationFastMixing}.
As discussed earlier, \protoref{certification} is also computationally efficient and experimentally feasible, requiring only minimal hardware control. 

In \fig{plotsBenchmarking}, we compare the performance of the shadow overlap with that of fidelity and the cross entropy benchmark (XEB), a prominent metric employed in the evaluation of quantum supremacy experiments with local random quantum circuits \cite{boixo2018characterizing, arute2019quantumsupremacy}.
In this numerical experiment, the shadow overlap is normalized, as explained in \appref{benchmarkingComplexity}, such that the target state attains value $1$ and the maximally mixed state attains value $1/2^n$.
We benchmark two families of states with $4$ and $20$ qubits: (1) Haar random states and (2) structured states which are phase states of the form $U_{\operatorname{phase}}\cdot \otimes_{i=1}^n \ket{\psi_i}$.
Here, each $\ket{\psi_i}$ is a single qubit state with random real amplitudes, and $U_{\operatorname{phase}}$ is diagonal with random complex phases.
We explore the effect of white noise (i.e. global depolarizing noise) as well as coherent noise realized as small Gaussian errors in both the magnitude and the phase of the probability amplitude, as discussed in \appref{NumericalBenchmarking}.

We observe that XEB performs well in Haar-random states but tends to overestimate fidelity for phase states, and has decreased effectiveness for smaller system size, likely due to reduced concentration effects.
In contrast, the shadow overlap closely matches the fidelity across different noise regimes and system sizes. 

These findings suggest that, much like XEB or similar benchmarks \cite{mark2022benchmarking, choi2023preparing}, an estimated shadow overlap may be accepted at face value, yielding a statistical figure of merit for the quality of the prepared states.
Otherwise, in the high-fidelity regime where the error $\epsilon \ll 1/\tau$, one can also apply \thmref{verificationFastMixing} to uncover a provable lower bound on the actual fidelity of the prepared state.
As a concrete example, in \appref{benchmarkingComplexity}, we show that, in a quantum processor that prepares a family of quantum states with tunable circuit complexity $C$, one can benchmark the fidelity of the device using a number of single-qubit measurements that scales polylogarithmically~with~$C$.

\subsection{Optimizing quantum circuits for state preparation}\label{sec:Optimizingcircuits}

Many variational quantum algorithms use the fidelity between two quantum states as their cost function. 
Such cost functions are known to suffer from exponentially vanishing gradients, known as barren plateaus \cite{mcclean2018barren, cerezo2021cost, caro2023out, jerbi2023power}, and require a high sample complexity to be statistically estimated.
One may use the shadow overlap $\E[\bm{\omega}]$ in place of the fidelity $\bra{\psi}\rho\ket{\psi}$ in some of these algorithms.
Besides demanding a substantially lower sample complexity, shadow overlaps may offer an improved optimization landscape with non-vanishing gradients. 
In particular, the shadow overlap displays behavior similar to the Hamming distance in cases where the target state shows no global correlations and its probability amplitudes are well-distributed across the Boolean hypercube $\{0,1\}^n$.
Indeed, for the special case of bit strings in the $X$-basis, the Hamming distance and the shadow overlap precisely coincide. 
One can see this for the simple case of $\ket{\psi} = \ket{+}^{\otimes n}$ where, as discussed in \appref{PerformanceguaranteesGeneral}, our protocol effectively measures the expectation $\Tr(L \rho)$ for the observable $L=\frac{1}{n}\sum_{i=1}^n \ketbra{+}{+}_i \otimes \iden_{\setminus i}$.
This observable is local and has favorable features when used as the cost function compared to the non-local observable $\ketbra{+}{+}^{\otimes n}$ used in the fidelity estimation; see the discussion in \cite{cerezo2021cost, caro2023out, jerbi2023power}.
\begin{figure}
    \centering
    \includegraphics[width=0.9\textwidth]{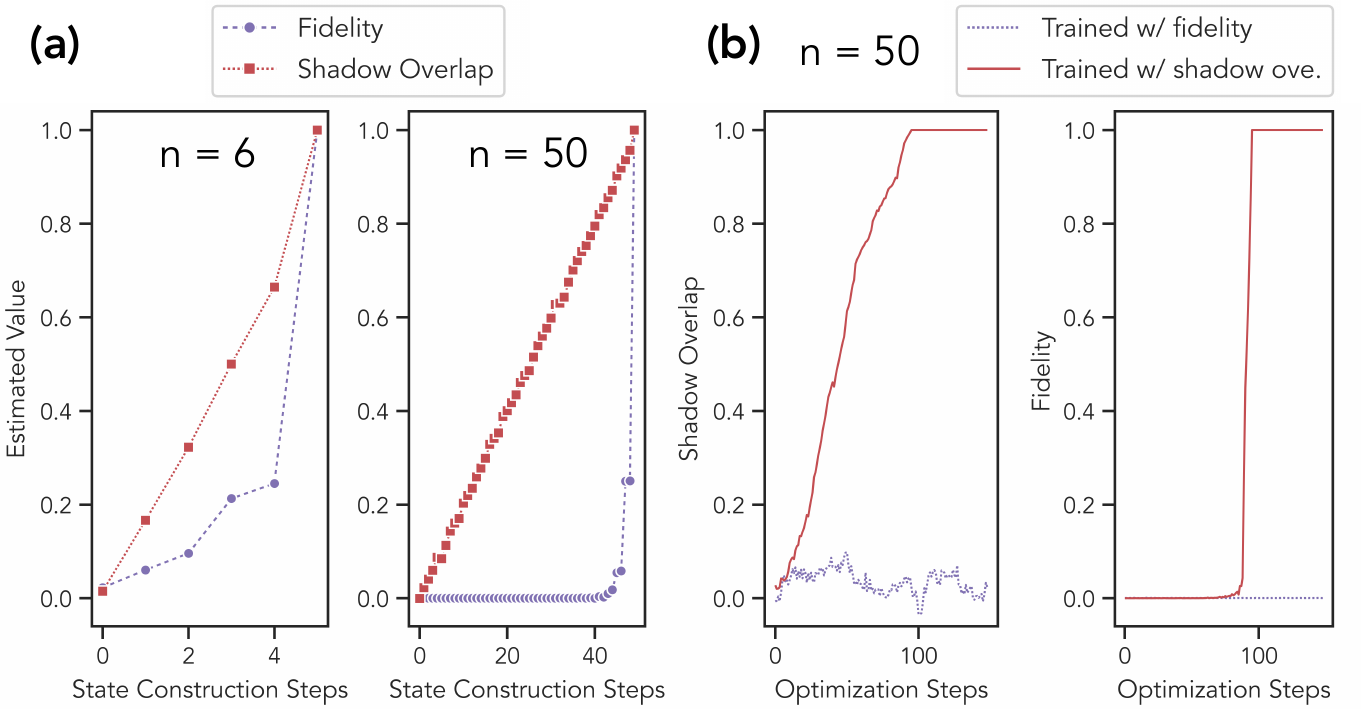}
    \captionsetup{justification=raggedright,singlelinecheck=false}
    \caption{\textbf{Optimizing quantum circuits for state preparation.} Training a low-depth quantum circuit consisting of Hadamard, controlled-$Z$, and $T$ gates to prepare a target state $\ket{\psi}$ given as a matrix product state (MPS). \textbf{(a)} As we approach the target state by building an appropriate circuit, the shadow overlap increases steadily with the number of circuit steps; in contrast, the fidelity sticks close to zero for many steps before growing abruptly.
    \textbf{(b)} Because the optimization landscape of fidelity has a barren plateau, training with fidelity fails to find a high-fidelity state-preparation circuit. In contrast, training with shadow overlap successfully finds a high-fidelity circuit. }
    \label{fig:OptStatePrep}
\end{figure}

In a numerical experiment presented in \fig{OptStatePrep} and discussed in \secref{NumericalStatePrep}, we investigate this feature of the shadow overlap in the context of training quantum circuits to optimally prepare a target state.
We have access to the matrix product state (MPS) representation of the target state, which corresponds to the output of a one-dimensional IQP circuit \cite{bremner2011classical} infused with random $T$ gates.
Through a variational optimization, we train a quantum circuit employing Hadamard, controlled-$Z$, and $T$ gates to generate the target state, optimizing for maximum shadow overlap. 
We then assess this method's performance against fidelity-based training.
Changes in both fidelity and shadow overlap are monitored across optimization steps.
When employing $n=50$ qubits, fidelity-based training encounters barren plateaus, whereas shadow overlap-based training successfully prepares the target state with a fidelity very close to $1$.
We also note that, akin to the linear decrease in Hamming distance between two binary strings as suitable bits are flipped, under shadow-overlap-based training the deviation of the shadow overlap from 1 decreases linearly as suitable gates are added to the circuit. This contrasts with the fidelity, which fails to exhibit a steady, gradual increase as the number of state construction steps increases.

\section{Outlook}
Further extending the reach of our certification protocol based on the shadow overlap raises many interesting open questions.

\vspace{0.5em}
    
\noindent \textbf{Quantum states with fast relaxation times}: What families of quantum states provably admit a $\poly(n)$ relaxation time with respect to the Markov chain \eqref{eq:TransitionMatrixIntro} introduced in our analysis?
We show that Haar random quantum states exhibits a relaxation time bounded by $\tau \leq O(n^2)$.
Can our arguments for Haar random states be extended to ``state $t$-designs'' whose first $t$ moments match that of the Haar measure?
Such quantum states can be efficiently prepared with random quantum circuits of size $\poly(n, t)$  \cite{Brandao16Design, Haferkamp2022randomquantum}. 
More generally, can we show that states prepared with (random) quantum circuits of \emph{arbitrary} depth satisfy a relaxation time $\tau \leq\poly(n)$?

\vspace{0.5em}

\noindent \textbf{States that cannot be certified with few single-qubit measurements:} Are there concrete examples of target quantum states that cannot be certified using any protocol that only relies on a $\poly(n)$ number of single-qubit measurements?

\vspace{0.5em}

\noindent \textbf{Mixed states}: Can a similar protocol be developed when the target state belongs to a certain family of \emph{mixed} quantum states?
If we allow arbitrary mixed states, then known lower bounds from certifying maximally mixed state \cite{Buadescu2019Certification} rule out a protocol with $\poly(n)$ sample complexity even with entangled measurements.
Going beyond worst cases, the instance-optimal sample complexity of certifying mixed states has been studied in \cite{Chen2022InstanceOptimal}, where approximately low-rank mixed states can be certified efficiently using highly-entangled measurements.
Can almost all approximately low-rank mixed states be certified with few single-qubit measurements?

\vspace{0.5em}
\subsection*{Code and Data Availability:}

The code and data for conducting the numerical experiments and for generating the figures in this work are openly available on Google Drive at \url{https://bit.ly/3U93gvl}.

\vspace{0.5em}
\subsection*{Acknowledgments:}

{ The authors thank Anurag Anshu, Ryan Babbush, Michael Broughton, David Gosset, Robin Kothari, and Jarrod R. McClean for valuable input and inspiring discussions.
HH is supported by a Google PhD fellowship and a MediaTek Research Young Scholarship.
HH acknowledges the visiting associate position at the Massachusetts Institute of Technology.
JP acknowledges support from the U.S. Department of Energy Office of Science, Office of Advanced Scientific Computing Research (DE-NA0003525, DE-SC0020290), the U.S. Department of Energy, Office of Science, National Quantum Information Science Research Centers, Quantum Systems Accelerator, and the National Science Foundation (PHY-1733907). 
MS was supported by AWS Quantum Postdoctoral Scholarship and funding from the National Science Foundation.
Institute for Quantum Information and Matter is an NSF Physics Frontiers Center. }

\newpage
\addtocontents{toc}{\protect\setcounter{tocdepth}{-5}}
\settocdepth{part}
\newpage

\appendix

\secondtitle{\textbf{Appendices}\vspace{-0.6em}}
\addappheadtotoc

\settocdepth{subsection}

\setcounter{section}{0}
\renewcommand*{\thesection}{\Alph{section}}

\section{Related work}

Many past works have studied the problem of certifying the fidelity $\bra{\psi}\rho\ket{\psi}$ between a quantum state $\ket{\psi}$ and a general state $\rho$. 
The full-blown tomography of a $2^n$-dimensional quantum state $\rho$ is known to require $\Theta(4^n/\epsilon^2)$ many copies to achieve an $\epsilon$ error in trace distance---or $\Theta(4^n/\epsilon)$ copies for an $\epsilon$ error measured by infidelity \cite{odonnell2016tomography, odonnell2017tomography2, haah2017sample}.
While this is, in principle, sufficient for estimating the fidelity $\bra{\psi}\rho\ket{\psi}$, it is also known that the same certification task can be achieved with a dramatically lower copy complexity of $\Theta(1/\epsilon^2)$, independent of the system dimensions.  
Achieving this scaling, however, entails performing certain quantum operations that are often contrary to the objective of the certification.
This includes starting from the description of a (potentially highly-entangled) $n$-qubit state $\ket{\psi}$ and performing the two-outcome measurement $\{ \ketbra{\psi}{\psi}, \iden - \ketbra{\psi}{\psi}\}$) or preparing independent copies of the state $\ket{\psi}$ (e.g., to perform the swap test) \cite{montanaro2013survey}.
The recent framework of classical shadows \cite{huang2020predicting} improves this by performing Clifford measurements on independent copies of the state $\rho$.
Implementing randomized Clifford measurements, however, requires deep circuits that may be practically infeasible. 
The efficiency of this protocol further depends on computing the overlap between the state $\ket{\psi}$ and random stabilizer states, which for instance, may not be achievable for states with high stabilizer~rank.

Provided that single-qubit Pauli measurements can be performed on the copies of the state $\rho$, a method known as the \emph{direct fidelity estimation} introduced and analyzed in \cite{Flammia2011DirectFidelity} (also see \cite{Poulin2011Characterization, aolita2015reliable}) can estimate the overlap $\bra{\psi}\rho\ket{\psi}$ up to an additive error $\epsilon$ using $\mathcal{O}(2^n/\epsilon^4)$ number of copies.
This scaling is an $\mathcal{O}(2^n)$ improvement on a naive application of the full tomography scheme but still grows exponentially $\Omega(2^n)$ with the number of qubits. 
A variant of classical shadows can also be applied by employing randomized Pauli measurements, where the number of samples required scales exponentially with the weight of the measured observable. This approach enables us to effectively investigate only small subsystems \cite{huang2020predicting} but cannot estimate fidelity with a highly-entangled state efficiently.

When the target state comes from certain special classes of states, such as stabilizer states or states generated by shallow quantum circuits, direct fidelity estimation \cite{Flammia2011DirectFidelity} and classical shadow based on randomized Pauli measurements \cite{huang2020predicting, huang2024learning} can efficiently certify the target states.
Other examples include hypergraph states, output states of IQP circuits \cite{takeuchi2018verification}, bosonic Gaussian states \cite{aolita2015reliable}, and fermionic Gaussian states \cite{gluza2018fidelity}.
For generic random target states, \emph{cross entropy benchmarking (XEB)} provides a good estimate for fidelity under certain noise models as studied in \cite{arute2019quantumsupremacy, choi2023preparing, cotler2023emergent, dalzell2024random}.

XEB and shadow overlap both use only single-qubit measurements and both require having access to $\braket{x}{\psi}$ for $x \in \{0, 1\}^n$, which can be time consuming during classical postprocessing.
However, XEB does not solve the quantum state certification task rigorously because XEB only uses Z-basis measurement.
As a result, XEB can outputs a fidelity score of one even when the state $\rho$ is a classical probability distribution with no quantum entanglement and is far from the target state $\ket{\psi}$.
In contrast, shadow overlap provably solves the certification task for almost all target states $\ket{\psi}$.
On a high level, shadow overlap can be seen as an enhancement of XEB using classical shadow based on randomized Pauli measurements on a randomly chosen qubit instead of all $Z$-basis measurements.

\section{Review of randomized Pauli measurements}\label{sec:randomizedPauliMeasurements}
As part of our certification protocol, we use randomized Pauli measurements on pure $m$-qubit states. 
These measurements provide sufficient statistical information for obtaining an unbiased estimator of the state. 
At the same time, they are practically appealing and require minimal experimental capabilities to be performed in practice.

Such measurements are part of a broader framework for the statistical study of quantum systems, called classical shadows. 
In this framework, measurements give us classical access to snapshots of the quantum state. 
With sufficiently many randomized snapshots, one can efficiently and accurately estimate various properties of subsystems of a quantum state.

Suppose we perform randomized $X$, $Y$, or $Z$ Pauli measurements on all $m$ qubits of a state $\ket{\varphi}$. 
As a result of these measurements, we obtain $m$ single-qubit states denoted by
\begin{align}
    \ket{\bm{s_1}} \otimes \cdots \otimes \ket{\bm{s_m}} \quad \text{where} \quad \ket{\bm{s_1}}, \dots, \ket{\bm{s_m}} \in \{\ket{0}, \ket{1}, \ket{+}, \ket{-}, \ket{i+}, \ket{i-}\}.
\end{align}
Here, $\{\ket{+}, \ket{-}\}$, $\{\ket{i+}, \ket{i-}\}$, and $\{\ket{0}, \ket{1}\}$ are eignestates of Pauli $X$, $Y$, and $Z$, respectively.
We use the collected data to compute the operator 
\begin{align}
\bm{\sigma} := \left(3 \ket{\bm{s_1}}\bra{\bm{s_1}} - \iden\right) \otimes \cdots \otimes \left(3 \ket{\bm{s_m}}\bra{\bm{s_m}} - \iden\right)\label{eq:measuredShadows}
\end{align}
that when averaged sufficiently many times, gives an accurate approximation of the original state $\ketbra{\varphi}{\varphi}$.
This can be better seen by expressing the operator \eqref{eq:measuredShadows} in terms of Pauli operators. 
Let $\bm{W_1},\dots, \bm{W_m} \in \{X, Y, Z\}$ denote the randomly chosen Pauli operators that we measured on the state $\ket{\varphi}$,
and let $\bm{o_1},\dots, \bm{o_m} \in \{+1, -1\}$ be the observed outcomes. 
The measurements return 
\begin{align}
    \frac{1}{2}\Tr\left( \left( \iden + 3 \bm{o_1} \bm{W_1} \right) \otimes \cdots \otimes \left( \iden + 3 \bm{o_n} \bm{W_}n \right) \rho \right)
\end{align}
with a probability that can be computed as 
\begin{align}\frac{1}{3^n}\Tr\left( \left( \iden + \bm{o_1} \bm{W_1} \right) \otimes \cdots \otimes \left( \iden + \bm{o_m} \bm{W_m} \right) \rho \right)
\end{align}
A direct calculation on this random ensemble (see \cite{huang2020predicting}) shows that in expectation, we have $\E[\bm{\sigma}] = \ketbra{\varphi}{\varphi}$.
Moreover, we have $\norm{\bm{\sigma}}_{\infty} = 2^m$. 
We use these facts to establish concentration bounds for the certification protocol in \appref{Performanceguarantees}. 

\section{Performance guarantees of the certification protocol}\label{sec:PerformanceguaranteesGeneral}

In this section, we provide a detailed analysis of the certification protocol based on shadow overlaps. 
We begin with a broad overview before discussing the technical details.

\addtocontents{toc}{\protect\setcounter{tocdepth}{1}}
\subsection{Technical overview}

Our certification protocol is based on the notion of property testing.
Similar quantum state certification tasks were studied in a series of previous works with the a focus on certifying mixed states \cite{Buadescu2019Certification, Chen2022Certification, chen2022toward}. 
In our framework, we are given states $\rho$ and $\ket{\psi}$ and promised that either $\bra{\psi}\rho\ket{\psi} \geq 1 - \frac{\epsilon}{2 \tau}$ or $\bra{\psi}\rho\ket{\psi} < 1 - \epsilon$.
The certification procedure outputs \textsc{Certified} in the first case and \textsc{Failed} in the second case.
The goal of our analysis, as stated in \thmref{verificationFastMixing-formal} in \appref{Performanceguarantees}, is to show that by using \protoref{certification} and setting the number of samples to $T = \mathcal{O}\left(\frac{\tau^2}{\epsilon^2}\cdot \log(\frac{1}{\delta})\right)$, we can correctly decide between the  \textsc{Certified} and \textsc{Failed} instances with probability at least $1 - \delta$.

The soundness of our protocol crucially relies on the ability to measure the state in both the Pauli-$X$ and Pauli-$Z$ bases.
This quantum mechanical feature is in sharp contrast to the classical testing of probability distributions, where, in effect, only Pauli-$Z$ measurements are possible.
For instance, given a distribution $p(x)$, consider measuring the state $\ket{\psi} = \sum_{x \in \{0, 1\}^n}\sqrt{p(x)}\ket{x}$ only in the $Z$ basis. 
This gives us samples $\bm{x}_1, \dots, \bm{x_T} \sim p(x)$.
It is well-known \cite{Paninski2591136Uniformity} that testing whether the distribution $p(x)$ equals the uniform distribution or is far from the uniform distribution requires a number of samples that scales exponentially, $T = \Omega(2^{n/2})$. 
We get around such grim scalings encountered classically by  also allowing some complementary Pauli measurements (i.e. in $X$ or $Y$-basis). 
Indeed in the simple case of $\ket{\psi} = \sum_{x \in \{0, 1\}^n}\frac{1}{\sqrt{2^n}}\ket{x}$, measuring all $n$ qubits in the $X$ basis enables directly measuring $\{\ket{\psi}\bra{\psi}, \iden - \ket{\psi}\bra{\psi} \}$. 
In contrast, our protocol shows that measuring just one of the $n$ qubits in a randomized Pauli basis is sufficient for certifying generic quantum states $\ket{\psi}$.

We note that the notion of relaxation time used in our framework is basis-dependent.
In fact, a quantum state can mix slowly in one basis and mix rapidly in a locally rotated basis.
The GHZ state provides an example of such a state.
In the standard $Z$-basis, we have $\ket{\mathrm{GHZ}} = \frac{1}{\sqrt2}(\ket{0}^{\otimes n}  + \ket{1}^{\otimes n})$.
Hence the distribution of this state is supported on two vertices $0^n$ and $1^n$ and does not mix under the walk defined by the transition matrix $P(x,y)$.
In the Hadamard basis, this state can be expanded as $H^{\otimes n} \ket{\mathrm{GHZ}} = \frac{1}{\sqrt{2^{n-1}}
}\sum_{x \in \mathrm{even} } \ket{x}$ and enjoys fast mixing, though for this purpose we need to change the protocol so that the observable $L$ corresponds to a walk that jumps from a vertex to its next-to-nearest neighbor.
This is discussed in more depth in \appref{GHZ}. 

A slightly modified version of the \protoref{certification} offers an improved sample complexity~$T = \mathcal{O}\left(\frac{\tau}{\epsilon}\cdot \log(\frac{1}{\delta})\right)$ with a quadratically better dependency on the error $\epsilon$ and mixing time $\tau$.     
This can be achieved if we replace the randomized Pauli measurement on qubit $\bm{k}$ with a direct measurement in the orthogonal basis $\{ \ketbra{\Psi_{\bm{k,z}}}{\Psi_{\bm{k,z}}}, \iden -\ketbra{\Psi_{\bm{k,z}}}{\Psi_{\bm{k,z}}}\}$ and output $\bm{\omega} = 0$ or $\bm \omega = 1$ depending on the measurement outcome. 
The sample complexity of this procedure is analyzed in \thmref{sampleComplexityImprovedInformal} proved in \appref{Performanceguarantees}. However, the modified protocol has the disadvantage that the query model $\Psi$ must be consulted in each measurement round to determine the measurement basis used. In contrast, using the original protocol we can measure all the copies of $\rho$ first, and then query $\Psi$ later on when we wish to estimate the fidelity. In fact, in the original protocol we can use the same data set repeatedly to estimate the fidelity with a variety of target states, each with its own query model. 
This feature is further explored in \secref{Learning by hypothesis selection}.

One enhancement to the certification \protoref{certification} is obtained by allowing more than one qubit to be measured in a random Pauli bases. 
This gives us a ``hierarchy'' of protocols (introduced more formally as \protoref{verificationGeneral} in \appref{Performanceguarantees}) where level $1$ corresponds to \protoref{certification} and level $m$, for a constant $m$, is constructed as follows:
(1) Choose $n-m$ qubits uniformly at random and measure them in the Pauli-$Z$ basis. 
(2) The remaining $m$ qubits are each measured in a random Pauli $X$, $Y$, or $Z$ basis. 
(3) As before, this allows us to construct the classical shadow $\bm{\sigma}$ of the post-measurement state on $m$ qubits.
(4) Based on the measurement outcomes on the $n-m$ qubits, we query the model $\Psi$ to compute an observable $L_{\bm{z_k}}$ that we specify soon.
(5) 
Instead of an overlap of the form in Equation \eqref{eq:overlapinprotocol-inline}, we compute $\bm{\omega} = \Tr(L_{\bm{z_k}} \bm{\sigma})$. (6) Following $T$ repetitions of the previous steps, we report the empirically estimated shadow overlap $\frac{1}{T}\sum_{t=1}^T \bm{\omega_t}$.

In this procedure, the observable $L_{\bm{z_k}}$ is the transition matrix of a weighted random walk among $n$-bit strings that differ in at most $m$ bits. 
When $m = 1$, this observable is simply the projector $\ketbra{\Psi_{\bm{k,z}}}{\Psi_{\bm{k,z}}}$.
More generally, for $m > 1$, the observable $L_{\bm{z_k}}$ is a sum of projectors onto states defined similarly to $\ket{\Psi_{\bm{k,z}}}$ in Equation \eqref{eq:querystate-inline} and can be easily computed by querying the model $\Psi$.

In later sections, we show examples of states that can be verified with this generalized certification procedure. 
This includes entangled states such as the GHZ state, as well as the ground states of gapped sign-problem-free Hamiltonians.

\begin{table}[!ht]
\begin{protocol}{Level-$m$ of certification procedure for states $\rho$ and $\ketbra{\psi}{\psi}$ using shadow overlaps}{proto:verificationGeneral}
 \small{
\textbf{Input:} $T$ samples of an unknown state $\rho$, a model $\Psi$ that gives query access to the amplitudes of $\ket{\psi}$, a level $m\in [n]$, an error $0\leq \epsilon < 1$, and the relaxation time $\tau$ of the Markov chain sampling from $|\braket{x}{\psi}|^2$.
}
\sbline

\small{ \textbf{Goal:} Certify that the overlap $\bra{\psi}\rho\ket{\psi} \geq 1-\epsilon$. }\\
\small{\quad\quad\quad  If the fidelity is low $\bra{\psi}\rho\ket{\psi} < 1 - \epsilon$, output \textsc{Failed} with high probability.} \\
\small{\quad\quad\quad  If the fidelity is high $\bra{\psi}\rho\ket{\psi} \geq 1 - \frac{\epsilon}{2 \tau}$, output \textsc{Certified} with high probability.}

\sbline

\textbf{Procedure:}
 \begin{enumerate}
  \small{ \item Among the total $n$ qubits of $\rho$, choose a uniformly random subset of size at most $m$ qubits. Denote these qubits by $$\bm{k} = \{\bm{k_1},\dots, \bm{k_r}\}$$ where $\bm{r}\leq m$ is the size of the subset.}
   \small{ \item Perform single-qubit $Z$-basis measurements on all but qubits $\bm{k_1}, \dots, \bm{k_r}$ of $\rho$. Denote the measurement outcomes collectively by $\bm{z_k}\in\{0,1\}^{n-\bm{r}}$. }
   \small{ \item For each qubit $\bm{k_1},\dots,\bm{k_r}$, choose an $X$, $Y$, or $Z$-basis measurement uniformly at random and measure that qubit of $\rho$. Denote the post-measurement state of the qubits $\bm{k_1},\dots,\bm{k_r}$ by $\ket{\bm{s_1}},\dots,\ket{\bm{s_r}}$ respectively.
    Compute the classical shadow 
    \begin{align}
        \bm{\sigma} = \left(3\ketbra{\bm{s_1}}{\bm{s_1}}-\iden \right)\otimes \left(3\ketbra{\bm{s_2}}{\bm{s_2}}-\iden\right)\otimes\cdots \otimes \left(3\ketbra{\bm{s_r}}{\bm{s_r}}-\iden\right).
    \end{align}}
    \small{\item Query the model $\Psi$ for all choices of $\bm{r}$-bit strings $\ell_1$ and $\ell_2$ that differ exactly in $\bm{r}$ bits (i.e. $\ell_1, \ell_2 \in \{0,1\}^{\bm{r}}$ and $\dist(\ell_1, \ell_2) = \bm{r}$) to obtain the normalized states
    \begin{align}
        \ket{\Psi^{\ell_1, \ell_2}_{\bm{z_k}}}:=\frac{\Psi(\bm{z_k}^{(\ell_1)})\cdot \ket{\ell_1}+\Psi(\bm{z_k}^{(\ell_2)})\cdot \ket{\ell_2}}{\sqrt{|\Psi(\bm{z_k}^{(\ell_1)})|^2+|\Psi(\bm{z_k}^{(\ell_2)})|^2}}.\label{eq:querystateGeneral}
    \end{align}
    Here the $n$-bit string $\bm{z_k}^{(\ell)}$ matches $\ell\in\{0,1\}^{\bm{r}}$ on bits $\bm{k_1},\dots,\bm{k_r}$ and equals $\bm{z_k}\in \{0,1\}^{n-\bm{r}}$ on the remaining $n-\bm{r}$ bits.
     } 
    \small{\item Compute the overlap 
    \begin{align}
     \bm{\omega}:= \Tr(L_{\bm{z_k}} \bm{\sigma}) \text{\quad with \quad} L_{\bm{z_k}} := \sum_{\substack{\ell_1, \ell_2 \in \{0, 1\}^{\bm{r}}\\ \dist(\ell_1, \ell_2) = \bm{r}}}\ketbra{\Psi^{\ell_1, \ell_2}_{\bm{z_k}}}{\Psi^{\ell_1, \ell_2}_{\bm{z_k}}}.  \label{eq:querystateProtocol}
    \end{align}}
   \small{ \item Repeat steps 1. to 5. for $T$ times to obtain overlaps $\bm{\omega_1},\dots,\bm{\omega_T}$. Report the estimated shadow overlap $$\hat{\bm{\omega}} := \frac{1}{T}\sum_{t=1}^T \bm{\omega_t}.$$ }
    \small{ \item If the estimated shadow overlap $\hat{\bm{\omega}} \geq 1 - \frac{3\epsilon}{4\tau}$, output $\textsc{Certified}$. Otherwise, output $\textsc{Failed}$.}
    \end{enumerate}
\end{protocol}
\end{table}

\subsection{Detailed analysis}\label{sec:Performanceguarantees}

Suppose we are provided with a representation of a many-body quantum state  
\begin{align}
    \ket{\psi}=\sum_{x\in\{0,1\}^n} \sqrt{\pi(x)} e^{i\phi(x)} \ket{x}
\end{align}
via query access to a model $\Psi: \{0,1\}^n\mapsto \mathbb{C}$. This model upon querying any $x\in\{0,1\}^n$ and $y\in\{0,1\}^n$ returns the possibly un-normalized amplitudes $\Psi(x)$ and $\Psi(y)$ such that~$\frac{\Psi(x)}{\Psi(y)}=\frac{\braket{x}{\psi}}{\braket{y}{\psi}}$.

We are also given multiple identical copies of a state $\rho$. Our goal is to use queries to $\Psi$ along with performing unentangled simple local measurements on copies of $\rho$ to certify that the state $\ketbra{\psi}{\psi}$ is close or far from the state $\rho$.

In this section, we analyze the performance of a generalized version of \protoref{certification} introduced in \secref{intro}.
In this version, which is stated in detail in \protoref{verificationGeneral}, we choose $m$ qubits uniformly at random and measure each on a randomized basis. 
This section includes the proof of \thmref{verificationFastMixing} as well as the equivalent statement for the level-$m$ protocol.

Fix a level $m$ for the certification protocol. The measurement distribution $\pi(x) = |\braket{x}{\psi}|^2$ is a distribution on a graph $G=(V,E)$ where the vertices $V=\{0,1\}^n$ are $n$-bit strings, 
and an edge $e = (x, y) $ exists between vertices $x$ and $y$ when they differ in $k \in \{1, \dots, m\}$ bits. 
Let~$\mathcal{S}:=\{x: \pi(x)>0\}$ denote the support of $\pi(x)$.
Let $N = \sum_{k=1}^m \binom{n}{k}$ be number of neighbors of each vertex. 
Consider a weighted version of this graph where every edge $(x,y)$ that connects vertices $x,y \in \mathcal{S}$ is assigned a weight $W(x,y)$ according to 
\begin{align}
W(x,y)=\begin{cases} 
        \frac{1}{N}\cdot \frac{\pi(x)\pi(y)}{\pi(x)+\pi(y)} & (x,y)\in E, \\
       \frac{1}{N}\cdot \sum_{x':(x', x)\in E}\frac{\pi(x)^2}{\pi(x)+\pi(x')} & x=y,\\
       0 & \text{otherwise}.
      \end{cases}\label{eq:TransitionMatrix-W}
\end{align}

Given this weight matrix $W = \sum_{x,y} W(x,y) \ketbra{x}{y}$, there is a canonical way to define a random walk or a Markov chain on this graph.
The transition matrix $P = \sum_{x\in\{0,1\}^n}P(x,y)\ketbra{x}{y}$ of this walk is defined by $P = S^{-1} W$ where the scaling matrix is set to be $$S = \sum_{x\in\{0,1\}^n} \pi(x) \ketbra{x}{x}.$$
Therefore, the transition probability $P(x,y)$ from a vertex $x\in \mathcal{S}$ to $y\in \{0,1\}^n$ is given by
 \begin{align}
P(x,y)=\begin{cases} 
       \frac{1}{N}\cdot \frac{\pi(y)}{\pi(x)+\pi(y)} & (x,y)\in E, \\
       \frac{1}{N}\cdot \sum_{x':(x', x)\in E}\frac{\pi(x)}{\pi(x)+\pi(x')} & x=y,\\
       0 & \text{otherwise}.
      \end{cases}\label{eq:TransitionMatrix}
\end{align}

In our application, it is more convenient to consider a normalized version of the transition matrix $P$ given by $S^{\frac{1}{2}} P S^{-\frac{1}{2}} = S^{-\frac{1}{2}}W S^{-\frac{1}{2}}$ where $S^{-\frac{1}{2}} = \sum_{x\in \mathcal{S}} \frac{1}{\sqrt{\pi(x)}} \ketbra{x}{x}$. 
We now claim that the observable whose expectation is measured in \protoref{verificationGeneral} directly relates to this normalized transition matrix. 

 \begin{proposition}\label{prop:performanceGuranteeLevelM}
     Suppose the certification \protoref{verificationGeneral} is performed on copies of the state $\rho$ and a model $\Psi$ of the quantum state $\ket{\psi} = \sum_{x\in\{0,1\}^n} \sqrt{\pi(x)} e^{i\phi(x)} \ket{x}$.
     Define the `phase matrix' by $F = \sum_{x \in \{0,1\}^n} e^{i\phi(x)}\ketbra{x}{x}$ and let $L$ be the Hermitian operator given by
    \begin{align}
        L = F\cdot S^{\frac{1}{2}} P S^{-\frac{1}{2}}\cdot F^{\dagger}. \label{eq:operatorL}
    \end{align}
    We have $L\ket{\psi} = \ket{\psi}$ and $\Tr(L \rho) = \E [\mathbf{\omega}]$, where $\E[\mathbf{\omega}]$ denotes the expected output of the certification protocol. 
\end{proposition}

\begin{proof}
The entries of the observable $L$ for any $x \in \mathcal{S}$ are given by
\begin{align}
\bra{x}L\ket{y}=\begin{cases} 
       \frac{1}{N}\cdot \frac{\sqrt{\pi(x)\pi(y)}}{\pi(x)+\pi(y)}\cdot e^{i(\phi(x)-\phi(y))} & (x,y)\in E, \\
       \frac{1}{N}\cdot \sum_{x':(x', x)\in E}\frac{\pi(x)}{\pi(x)+\pi(x')} & x=y,\\
       0 & \text{otherwise}.
      \end{cases}\label{eq:TransitionMatrix-L}
\end{align}
For $x\in\mathcal{S}$, we have
\begin{align}
    \bra{x}L\ket{\psi} &= \bra{x}L\ket{x}\cdot \braket{x}{\psi} + \sum_{y\neq x} \bra{x}L\ket{y}\cdot \braket{y}{\psi}\nn\\
    & =  \frac{1}{N} \sum_{y:(y, x)\in E}\frac{\pi(x) }{\pi(x)+\pi(y)}\cdot \sqrt{\p(x)}e^{i\phi(x)}
    + \frac{1}{N} \sum_{y:(y, x)\in E} \frac{\sqrt{\pi(x)\pi(y)}}{\pi(x)+\pi(y)}\cdot e^{i(\phi(x)-\phi(y))} \cdot \sqrt{\pi(y)}e^{i\phi(y)}\nn\\
    & = \frac{1}{N} \sum_{y:(y, x)\in E}\left(\frac{\pi(x) }{\pi(x)+\pi(y)} + \frac{\pi(y) }{\pi(x)+\pi(y)} \right)\sqrt{\p(x)}e^{i\phi(x)}\nn\\
    & = \sqrt{\p(x)}e^{i\phi(x)}.\label{eq:entriesOfL}
\end{align}
This shows that $L\ket{\psi} = \ket{\psi}$.
Next we prove that $\Tr(L \rho) = \E[\bm{\omega}]$. 
Consider subsets of qubits with size $r\leq m$.
There are $N = \sum_{ k=1}^m \binom{n}{k}$ choices for the location of these qubits. 
For any $r$, we enumerate the chosen qubits by $k_1,\dots, k_r$ and collectively denote them by $k = \{k_1, \dots, k_r\}$.
For a fixed $k$, the set $\{z_k \in \{0, 1\}^{n-r}$\} denotes all the possible bit strings on the remaining $n - r$ bits. 
Direct inspection reveals that the observable $L$ corresponding to the model $\Psi$ can be expressed as
\begin{align}
    L = \frac{1}{N} \sum_{r\in [m]} \sum_{k = \{k_1, \dots, k_r\}}\sum_{z_k\in\{0,1\}^{n-r}} \ketbra{z_k}{z_k} \otimes  L_{z_k},
\end{align}
where $ L_{z_k}$ is an operator acting on $r$ qubits $\{k_1,\dots, k_r\}$, and is given by
\begin{align}
      L_{z_k} := \sum_{\substack{\ell_1, \ell_2 \in \{0, 1\}^{r}\\ \dist(\ell_1, \ell_2) = r}}\ketbra{\Psi^{\ell_1, \ell_2}_{z_k}}{\Psi^{\ell_1, \ell_2}_{z_k}} \text{\quad with \quad} \ket{\Psi^{\ell_1, \ell_2}_{z_k}}:=\frac{\Psi(z_k^{(\ell_1)})\cdot \ket{\ell_1}+\Psi(z_k^{(\ell_2)})\cdot \ket{\ell_2}}{\sqrt{|\Psi(z_k^{(\ell_1)})|^2+|\Psi(z_k^{(\ell_2)})|^2}} .\label{eq:querystate-inline2}
\end{align}
In this expression, the binary string $z_k^{(\ell)}$ equals $\ell\in\{0,1\}^r$ on bits $k_1,\dots,k_r$ and equals~$z_k\in \{0,1\}^{n - r}$ on the remaining $n - r$ bits. 

Let $\bm{\sigma}$ denote the classical shadow obtained after performing randomized Pauli measurements on the post-measurement state of qubits $r$. 
That is, if the $r$ single Pauli measurements return states $\ket{\bm{s_1}},\dots, \ket{\bm{s_r}}$, 
we set $\bm{\sigma} = \left(3\ketbra{\bm{s_1}}{\bm{s_1}}-\iden \right)\otimes \left(3\ketbra{\bm{s_2}}{\bm{s_2}}-\iden\right)\otimes\cdots \otimes \left(3\ketbra{\bm{s_r}}{\bm{s_r}}-\iden\right)$. 
It follows from the discussion in \appref{randomizedPauliMeasurements} that $\E_{\mathrm{shadows}}[\bm{\sigma}] = \frac{\bra{\bm{z_k}}\rho\ket{\bm{z_k}}}{\Tr(\bra{\bm{z_k}}\rho\ket{\bm{z_k}})}$, 
where the expectation is over Pauli measurements on qubits $\bm{k}$. 

Using this, we can expand $\Tr[L \rho]$ as follows:
\begin{align}
    \Tr[L \rho] & = \frac{1}{N} \sum_{\substack{r\in [m] \\ k = \{k_1, \dots, k_r\}}}\sum_{z_k\in\{0,1\}^{n-r}}\Tr(\bra{z_k}\rho\ket{z_k})\cdot \Tr\left(L_{z_k} \cdot \frac{\bra{z_k}\rho\ket{z_k}}{\Tr(\bra{z_k}\rho\ket{z_k})}\right)\nn\\
    & = \E_{\bm{k}, \bm{z_k}} \Tr\left(L_{\bm{z_k}} \cdot \frac{\bra{\bm{z_k}}\rho\ket{\bm{z_k}}}{\Tr(\bra{\bm{z_k}}\rho\ket{\bm{z_k}})}\right)\nn\\
    & =\E_{\bm{k}, \bm{z_k}} \Tr\left(L_{\bm{z_k}} \cdot \E_{\mathrm{shadows}}[\bm{\sigma}]\right) \nn\\
    & =\E_{\bm{k}, \bm{z_k}} \E_{\mathrm{shadows}}[\Tr\left(L_{\bm{z_k}} \bm{\sigma}\right)] \nn\\
    & = \E[\bm{\omega}].\nn
\end{align}
In the last expression, the expectation is with respect to the location of the Pauli $Z$ measurements, their outcomes, as well as the randomized measurements on the remaining qubits. 
\end{proof}

When we first average over the classical shadows, the shadow overlap $\E[\omega]$ is equal to the average overlap between the postselected state on $\rho$ and the postselected state on the target state $\ket{\psi}$. Hence $0 \leq \E[\omega] = \Tr(L \rho) \leq 1$ for any state $\rho$. This implies that $0 \preceq L \preceq I$.

\begin{theorem}\label{thm:expectationEqualsOperatoL}
    Let $\lambda_1 = 1 - \frac{1}{\tau}$ be the second largest eigenvalue of the transition matrix $P$ defined with respect to the measurement distribution $\pi(x)$ of the state $\ket{\psi}$.
    The shadow overlap satisfies 
    \begin{align} \label{eq:gapFidelity}
        &\text{if\quad} \E[\bm{\omega}] \geq 1- \epsilon \text{\quad then we have \quad} \bra{\psi}\rho \ket{\psi}\geq 1- \tau \epsilon;\\
        &\text{if\quad} \bra{\psi}\rho \ket{\psi}\geq 1- \epsilon \text{\quad then we have \quad} \E[\bm{\omega}] \geq 1- \epsilon. \label{eq:nogapFidelity}
    \end{align}
\end{theorem}

\begin{proof}
    We first study the spectrum of the observable $L$.
    From the previous theorem and the fact that $0 \preceq L \preceq I$, the eigenvalues of $L$ are given by $1 = \lambda_0 \geq \lambda_1 \geq \lambda_2 \geq \cdots \geq 0$. 
    The two operators $P$ and $L$ are related by a similarity transformation. Hence, they have the same set of eigenvalues. 
    Let $\ket{\lambda_i}$ denote the eigenstate of observable $L$ corresponding to the eigenvalue $\lambda_i$.
    We claim that the top eigenstate $\ket{\lambda_0}$ of the operator $L$ is the quantum state $\ket{\psi}$.
    This can be seen by the direct calculation in \propref{performanceGuranteeLevelM} or by noting that the measurement distribution $\pi(x)$ is the unique stationary distribution of $P$. Hence, we have $P \ket{\pi} = \ket{\pi}$, 
    where $\ket{\pi} = \sum_{x\in\{0,1\}^n} \sqrt{\pi(x)} \ket{x}$.
    From this and the fact that $\ket{\pi} = F^{\dagger} \ket{\psi}$, we have $L\ket{\psi} = \ket{\psi}$, as claimed.
    
    Now we prove the implication stated in \eqref{eq:gapFidelity}. 
    From \propref{performanceGuranteeLevelM}, we know that $\E[\mathbf{\omega}] = \Tr[L\rho]$. Assuming $\E[\bm{\omega}] \ge  1-\epsilon$, we have
    \begin{align}
        1-\epsilon\leq \E[\mathbf{\bm{\omega}}] &= \Tr[L\rho]\nn\\
        & = \bra{\psi}\rho\ket{\psi} + \sum_{i\geq 1}\lambda_i\cdot \bra{\lambda_i}\rho\ket{\lambda_i} && \text{using\ } \ket{\lambda_0} = \ket{\psi}\nn\\
        &  \leq \bra{\psi}\rho\ket{\psi} +  \lambda_1\cdot \sum_{i\geq 1} \bra{\lambda_i}\rho \ket{\lambda_i} && \text{definition of\ } \lambda_1 \nn\\    
        &  \leq \bra{\psi}\rho\ket{\psi} +  \lambda_1\cdot(1- \bra{\psi}\rho\ket{\psi}) && \text{since\ } \Tr(\rho) = 1. \nn
    \end{align}
    By rearranging the two sides of the inequality, we arrive at the bound $\bra{\psi}\rho\ket{\psi} \geq 1- \frac{\epsilon}{1-\lambda_1} = 1 - \tau\epsilon$.
    We next prove the implication stated in \eqref{eq:nogapFidelity}.
    \begin{align}
        \E[\mathbf{\bm{\omega}}] &= \Tr[L\rho]\nn\\
        & = \bra{\psi}\rho\ket{\psi} + \sum_{i\geq 1}\lambda_i\cdot \bra{\lambda_i}\rho\ket{\lambda_i} && \text{using\ } \ket{\lambda_0} = \ket{\psi}\nn\\
        &\geq \bra{\psi}\rho\ket{\psi} \geq 1 - \epsilon.
    \end{align}
    This concludes the proof of this theorem.
\end{proof}

It is worth mentioning that the fast mixing assumption used in our results in this section, in particular \thmref{expectationEqualsOperatoL}, assumes a non-zero gap $\lambda_0 - \lambda_1\geq \frac{1}{\tau}$ between the first and the second eigenvalues of the transition matrix $P$, and therefore a \emph{unique} eigenvector with eigenvalue $\lambda_0 = 1$.
There are, however, quantum states $\ket{\psi}$ whose transition matrix possesses a degenerate eigenspace $\Pi_0$ with $\lambda_0 = 1$. 
Such states can occur when the support of the measurement distribution $\pi(x)$ is a union of disjoint subsets of the hypercube $\{0,1\}^n$. 
As shown in \appref{HaarSpectralGap}, these states are non-generic since Haar random states exhibit a non-zero gap $\tau \geq \Omega\left(\frac{1}{n^2}\right)$ with high probability. 
A simple modification of the proof of \thmref{expectationEqualsOperatoL} shows that in these cases, our protocol certifies the overlap $\Tr[\Pi_0 \rho]$ between the lab state $\rho$ and the degenerate subspace $\Pi_{0}$ that includes the target state $\ket{\psi}$.
That is if the shadow overlap $\E[\bm{\omega}] \geq 1- \epsilon$ then we have $\Tr[\rho \Pi_0]\geq 1- \tau \epsilon$. 
We discuss this in more depth in \appref{GHZ} in the context of the GHZ state.

\begin{theorem}[Sample Complexity of level-$m$ of \protoref{verificationGeneral}]\label{thm:sampleComplexity}
Using $T = 2^{2m}\cdot \frac{1}{\epsilon^2}\cdot \log\left(\frac{2}{\delta}\right)$ samples of the state $\rho$, the empirical average $\frac{1}{T}\sum_{t=1}^T \bm{\omega_t}$ computed by the protocol is within $\epsilon$ additive distance from the shadow overlap $\E[\bm{\omega}]$ with probability at least $1 - \delta$.
\end{theorem} 
\begin{proof}
  
    This claim is a consequence of conventional tail bounds.
    Note $\E\left(\frac{1}{T}\sum_{t=1}^T \bm{\omega_t}\right)=\E(\bm{\omega})$. 
    Since subset sizes $r$ satisfy $r\leq m$, it also holds that 
    \begin{align}
        \norm{L_{\bm{z_k}}}_1\leq  \max_{r\in[m]} \sum_{\substack{\ell_1, \ell_2 \in \{0, 1\}^{r}\\ \dist(\ell_1, \ell_2)=r}}\norm{\ketbra{\Psi^{\ell_1, \ell_2}_{\bm{z_k}}}{\Psi^{\ell_1, \ell_2}_{\bm{z_k}}}}_1  \leq 2^{m-1} 
    \end{align}
    We have $\norm{\bm{\sigma}}_{\infty} = 2^m$ from \appref{randomizedPauliMeasurements}. 
      From these bounds, we get $|\bm{\omega}| = |\Tr(\bm{L_{z_k}} \bm{\sigma})| \leq 2^{2m-1}$.
     Hence, an application of Hoeffding's inequality shows that
    \begin{align}
        \mathbf{Pr}\left[ \left| \frac{1}{T}\sum_{t=1}^T \bm{\omega_t} -\E[\bm{\omega}] \right| > \epsilon\right]\leq 2 e^{- \frac{T\epsilon^2}{2^{2m-1}}}.\label{eq:concentrationEomega}
    \end{align}
    Hence, when the number of samples $T$ satisfies
    $$T \geq 2^{2m}\cdot \frac{1}{\epsilon^2}\cdot \log\left(\frac{2}{\delta}\right),$$
    the empirical average $\frac{1}{T}\sum_{t=1}^T \bm{\omega_t}$ is within $\epsilon$ additive distance from the expected value $\E[\bm{\omega}]$ with probability $\geq 1-\delta$.
\end{proof}

\begin{theorem}[Efficient certification using level-$m$ of \protoref{verificationGeneral}]\label{thm:verificationFastMixing-formal}
    Given any $n$-qubit target state $\ket{\psi}$ with a relaxation time $\tau \geq 1$, error $\epsilon > 0$, failure probability $\delta > 0$, and
    \begin{equation}
        T = 2^{2m + 4}\cdot \frac{\tau^2}{\epsilon^2}\cdot \log\left(\frac{2}{\delta}\right)
    \end{equation}
    samples of an unknown $n$-qubit state $\rho$.
    With probability at least $1 - \delta$, \protoref{verificationGeneral} will correctly output \textsc{Failed} if the fidelity is low $\bra{\psi}\rho\ket{\psi} < 1 - \epsilon$ and will correctly output \textsc{Certified} if the fidelity is high $\bra{\psi}\rho\ket{\psi} \geq 1 - \frac{\epsilon}{2 \tau}$.
\end{theorem}
\begin{proof}
    From \thmref{sampleComplexity}, with probability at least $1 - \delta$, we have
    \begin{equation}
        \left| \frac{1}{T}\sum_{t=1}^T \bm{\omega_t} -\E[\bm{\omega}] \right| \leq \frac{\epsilon}{4 \tau}.
    \end{equation}
    We condition on the above event.
    In the case of $\bra{\psi}\rho\ket{\psi} < 1 - \epsilon$, we have $\E[\bm{\omega}] < 1 - \frac{\epsilon}{\tau}$ from \thmref{expectationEqualsOperatoL}.
    Hence, $\frac{1}{T}\sum_{t=1}^T \bm{\omega_t} < 1 - \frac{3\epsilon}{4 \tau}$ and the protocol will output \textsc{Failed}.
    In the case of $\bra{\psi}\rho\ket{\psi} \geq 1 - \frac{\epsilon}{2 \tau}$, we have $\E[\bm{\omega}] \geq 1 - \frac{\epsilon}{2 \tau}$ from \thmref{expectationEqualsOperatoL}.
    Hence, $\frac{1}{T}\sum_{t=1}^T \bm{\omega_t} \geq 1 - \frac{3\epsilon}{4 \tau}$ and the protocol will output \textsc{Certified}.
    This concludes the proof.
\end{proof}

We can improve the dependency of the sample complexity on $\tau / \epsilon$ for the level $m = 1$ test given in \protoref{certification} by replacing randomized Pauli measurements with measurement in the basis $\{ \ketbra{\Psi_{\bm{k,z}}}{\Psi_{\bm{k,z}}}, \iden -\ketbra{\Psi_{\bm{k,z}}}{\Psi_{\bm{k,z}}}\}$.
In this setting, we can show the following theorem.

\begin{theorem}[Efficient certification using a modified version of \protoref{certification}]\label{thm:sampleComplexityImprovedInformal}
    Let the number of samples $T$ in the modified version of \protoref{certification} with single-qubit measurements performed in the basis $\{ \ketbra{\Psi_{\bm{k,z}}}{\Psi_{\bm{k,z}}}, \iden -\ketbra{\Psi_{\bm{k,z}}}{\Psi_{\bm{k,z}}}\}$ be
    \begin{equation}
        T = 32 \cdot \frac{\tau}{\epsilon}\cdot \log\left(\frac{1}{\delta}\right).
    \end{equation}
    With probability at least $1 - \delta$, the protocol correctly outputs \textsc{Failed} if the fidelity is low $\bra{\psi}\rho\ket{\psi} < 1 - \epsilon$ and correctly outputs \textsc{Certified} if the fidelity is high $\bra{\psi}\rho\ket{\psi} \geq 1 - \frac{\epsilon}{2 \tau}$.
\end{theorem}
\begin{proof}[Proof of \thmref{sampleComplexityImprovedInformal}]
    In the modified test, the measured overlaps $\bm{\omega_t}$ satisfy $\bm{\omega _t} \in \{0, 1\}$ and $\E[\bm{\omega}] = \Tr[L \rho]$ with $L$ being the level-$1$ observable in Equation \eqref{eq:entriesOfL}. 
    It follows from the multiplicative form of the Chernoff bound that for any error $\eta$, we have
    \begin{align} \label{eq:chernoff-multi}
        \Pr\left[ \frac{1}{T}\sum_{t=1}^T \bm{\omega_t} \leq \E[\bm{\omega}] - \eta \right]\leq e^{-\frac{T \eta^2}{2(1 - \E[\bm{\omega}] + \eta)}} \text{\quad and \quad } \Pr\left[ \frac{1}{T}\sum_{t=1}^T \bm{\omega_t} \geq \E[\bm{\omega}] + \eta \right]\leq e^{-\frac{T \eta^2}{2(1 - \E[\bm{\omega}])}}. 
    \end{align}
    We now consider the two cases.

    In the case of $\bra{\psi}\rho\ket{\psi} < 1 - \epsilon$, we have $\E[\bm{\omega}] < 1 - \frac{\epsilon}{\tau}$ from \thmref{expectationEqualsOperatoL}.
    Let $\eta = 1 - \frac{3 \epsilon}{4\tau} - \E[\bm{\omega}] > \frac{1}{4}(1 - \E[\bm \omega])$.
    From the inequality in the right of Eq.~\eqref{eq:chernoff-multi}, we have
    \begin{equation}
        \Pr\left[ \frac{1}{T}\sum_{t=1}^T \bm{\omega_t} \geq 1- \frac{3\epsilon}{4\tau} \right]\leq e^{-\frac{T \eta^2}{2(1 - \E[\bm{\omega}])}} < e^{-\frac{T (1 - \E[\bm{\omega}])^2}{32 (1 - \E[\bm{\omega}])}} < e^{-\frac{T (\epsilon / \tau)}{32}} \leq \delta
    \end{equation}
    for $T = 32 \cdot \frac{\tau}{\epsilon}\cdot \log(\frac{1}{\delta})$.
    Hence, with probability $\geq 1 - \delta$, the protocol will output \textsc{Failed}.

In the case of $\bra{\psi}\rho\ket{\psi} \geq 1 - \frac{\epsilon}{2 \tau}$, we have $\E[\bm{\omega}] \geq 1 - \frac{\epsilon}{2 \tau}$ from \thmref{expectationEqualsOperatoL}.
    Let $\eta = \frac{3\epsilon}{4\tau} - (1 - \E[\bm \omega]) \geq \frac{\epsilon}{4 \tau}$.
    From the inequality in the left of Eq.~\eqref{eq:chernoff-multi}, we have
    \begin{equation}
        \Pr\left[ \frac{1}{T}\sum_{t=1}^T \bm{\omega_t} < 1- \frac{3\epsilon}{4\tau} \right]\leq e^{-\frac{T \eta^2}{2(1 - \E[\bm{\omega}] + \eta)}} \leq e^{-\frac{T (\epsilon / (4 \tau))^2}{2((3 \epsilon) / (4 \tau))}} \leq \delta
    \end{equation}
    for $T = 32 \cdot \frac{\tau}{\epsilon}\cdot \log(\frac{1}{\delta})$.
    So, with probability $\geq 1 - \delta$, the protocol will output \textsc{Certified}.
\end{proof}

Despite the quadratically larger $\tau^2 / \epsilon^2$ scaling, the original \protoref{certification} with randomized Pauli measurements offers a more versatile framework for the following two reasons.
\begin{enumerate}
    \item The experiment can be done without needing to interact with a particular query model of the state $\ket{\psi}$. This means an experimentalist can collect data from the state $\rho$. After the completion of the data acquisition phase, they can conduct the certification procedure to benchmark their device or certify the performance of their machine learning algorithms. 
    \item The classical shadow framework allows us to generalize the original protocol \protoref{certification} with $m=1$ to level-$m$ protocols in \protoref{verificationGeneral} with $m > 1$, which can be used to certify a broader family of quantum states using single-qubit measurements.
\end{enumerate}

\addtocontents{toc}{\protect\setcounter{tocdepth}{2}}

\section{Haar random states}\label{sec:HaarSpectralGap}
Consider a probability distribution $\pi(x)$ defined on the $n$-dimensional hypercube $G=(V, E)$ where $V=\{0,1\}^n$ and two vertices $x$ and $y$ are connected when they differ only in one bit. This distribution is the unique stationary distribution of a Markov chain whose transition matrix includes additional self-loops and has the following entries:
 \begin{align}
P(x,y)=\begin{cases} 
       \frac{1}{n}\cdot \sum_{x':(x', x)\in E}\frac{\pi(x)}{\pi(x)+\pi(x')} & x=y,\\
       \frac{1}{n}\cdot \frac{\pi(y)}{\pi(x)+\pi(y)} & (x,y)\in E, \\
       0 & \text{otherwise}.
      \end{cases}\label{eq:TransitionMatrix-Haar}
\end{align}
Our goal is to analyze the spectral gap of the transition matrix $P$ when $\pi(x)$ is randomly chosen according to the Porter-Thomas distribution. That is, for any $x\in \{0,1\}^n$, we independently draw a sample $\bm{z}(x)$ from the exponential distribution $\Pr(\bm{z})=e^{-\bm{z}}\cdot \iden[\bm{z}\geq 0]$. We then set $\pi(x)=\frac{\bm{z}(x)}{2^n}$. As we will soon show, this distribution is with high probability normalized.

\subsection{Preliminaries}

Before giving the proof of this theorem, we need to gather some preliminary facts. 
\begin{definition}
A random variable $\bm{z}$ with mean $\mu$ is `sub-exponential' if there are non-negative parameters $(a,b)$ such that
\begin{align}
    \E\left(e^{\l(\bm{z}-\m)}\right)\leq e^{\frac{a^2 \lambda^2}{2}} & \quad \text{for all\ } |\l|\leq \frac{1}{b}
\end{align}
\end{definition}
 \begin{lemma}
     A variable $\bm{z}$ sampled from the exponential distribution $\Pr(\bm{z})=e^{-\bm{z}}\cdot \iden[\bm{z}\geq 0]$ is sub-exponential with parameters $(a,b)=(2,2)$.
 \end{lemma}

This lemma, along with known tail bounds for sub-exponential random variables, result in the following proposition:

 \begin{proposition}[Tail bounds]\label{prop:concentration}
The following two tail bounds hold:
 \begin{enumerate}
     \item Let $\bm{x_1}, \dots, \bm{x_m}$ be $m$ i.i.d. Bernoulli random variables with $\Pr[\bm{x_i} = 1] = p$ for $i\in [m]$. There exist constants $\alpha, \beta > 0$ such that 
     \begin{align}
         \Pr\left(\sum_{i=1}^m \bm{x_i} \leq \alpha \cdot n\right)\leq 2^{-(1+\beta)n}.\label{eq:concentrationBin}
     \end{align}
     \item Suppose $\bm{z}_1,\dots,\bm{z}_m$ are drawn independently from the exponential distribution $\Pr(\bm{z})=e^{-\bm{z}}\cdot \iden[\bm{z}\geq 0]$. Then, the following two-sided concentration bound holds:
     \begin{align}
         \Pr\left(\left|\frac{1}{m}\sum_{k=1}^m \bm{z}_k - 1\right|\geq t \right)\leq
         \begin{cases}
             2e^{-\frac{m t^2}{8}} & 0\leq t\leq 2,\\
             2e^{-\frac{m t}{4}} & t\geq 2.
         \end{cases}\label{eq:concentrationSubExp}
     \end{align}
Moreover, for any $k$, we have
\begin{align}
    \Pr\left(\bm{z}_k\leq t\right)=1-e^{-t}\leq t \label{eq:CDFexp}
\end{align}
 \end{enumerate}
     
 \end{proposition}

\begin{proposition}\label{prop:nVertexConnectionHypercube}
For any two vertices $x,y \in \{0,1\}^n$ in the $n$-dimensional hypercube, there are $n$ simple pairwise internally disjoint paths connecting them such that no two paths share a vertex except for the starting and end vertices $x$ and $y$. Moreover, the length $|\g_{xy}|$ of all such paths is bounded by $n+1$. 
\end{proposition}
\begin{proof}
    We can explicitly design such paths by repeating the following procedure for any $i\in [n]$: Choose the $i$'th bit of the starting vertex $x \in \{0,1\}^n$ and flip it. Then moving rightward from bit $i+1$ to $n$ and back to $i$, flip any bit in which the current vertex differs from the end vertex $y \in \{0,1\}^n$. 
\end{proof}

\subsection{Warm-up: a loose analysis}\label{sec:WarmupLooseAnalysis}

A set of canonical paths $\G$ in a graph $G$ is a collection of simple paths $\gamma_{xy}$, each connecting a unique pair of distinct vertices $(x,y)\in E$. Given a set of canonical paths, the path congestion parameter is defined by
\begin{align}
    \rho(\G):= \max_{e \in E} \frac{1}{Q(e)} \sum_{\g_{xy}\ni e} \pi(x)\pi(y) |\gamma_{xy}|
\end{align}
where for an edge $e=(e^+,e^-)$, its weight given by $Q(e)=\pi(e^{+})P(e^{+},e^{-})$ is the probability of the transition $(e^+,e^-)$ occurring in the random walk. 
\begin{proposition}[Path congestion vs. spectral gap, cf. \cite{sinclair1992improved}]
Let $\lambda_1$ be the second largest eigenvalue of the transition matrix $P$ for a reversible Markov chain. For any choice of canonical paths $\G$, it holds that 
\begin{align}
        1-\lambda_1\geq \frac{1}{\rho(\G)}.
\end{align}
\end{proposition}

Hence, to obtain a lower bound of $\frac{1}{\poly(n)}$ on the spectral gap $1-\l_1$, it suffices to prove an upper bound $\rho(\G)\leq \poly(n)$ for the congestion parameter for some set of canonical paths. In this section, we show:
\begin{theorem}\label{thm:upperBoundCongestion}
Consider the random walk in \eqref{eq:TransitionMatrix-Haar} with $\pi(x)$ drawn independently from the Porter-Thomas distribution. Then, with probability $\geq 1-\frac{1}{2^n}$, there exists a set of canonical paths $\G$ for which the congestion parameter $\r(\G)$ is upper bounded by 
\begin{align}
    \r(\G)\leq \mathcal{O}(n^4 \cdot \log n).
\end{align}
\end{theorem}

Prior to proving this theorem, we gather some useful tools.  
Define the normalization factor $\bm{S}=\sum_{x\in \{0,1\}^n} \bm{z}(x)$. Assuming the set of canonical paths does not contain a self-loop, we can rewrite the path congestion as:
\begin{align}
    \bm{\rho}(\G):= \frac{n}{\bm{S}}\max_{e=(e^+,e^-)} \frac{\bm{z}_{e^+}+\bm{z}_{e^-}}{\bm{z}_{e^+}\bm{z}_{e^-}} \sum_{\g_{xy}\ni e} \bm{z}(x)\bm{z}(y) |\gamma_{xy}|.\label{eq:congestionRewrite}
\end{align}

In a random realization of the vertices $x\in \{0,1\}^n$, the value $\bm{z}(x)$ of some vertices may be $o(1/\poly(n))$. This means the `capacity' $Q(e)$ of the edges connected to such vertices could also be $o(1/\poly(n))$. To obtain the desired upper bound $\rho(\G)\leq \poly(n)$, we need to carefully design a set of canonical paths $\G=\{\g_{xy}\}$ that when possible, avoid traversing through such vertices $\bm{z}(x)$ that have negligible weights. 
To this end, we partition the vertices into two sets of `bad' and `good' vertices $\bm{V_{\mathrm{bad}}}:=\{x\in \{0,1\}^n: \bm{z}(x)\leq 1/(64n)\}$ and $\bm{V_{\mathrm{good}}}:=V\setminus \bm{V_{\mathrm{bad}}}$.

\begin{lemma}\label{lem:propertiesofCanonicalPaths}
    Except for probability $\leq 2^{-3n}$, there exists a set of canonical paths $\G=\{\g_{xy}\}$ with the following properties:
    \begin{enumerate}
    \item All paths are simple and without self-loops.
    \item The length of each path is bounded by $|\g_{xy}|\leq n+3$. 
    \item A bad vertex in $\bm{V_{\mathrm{bad}}}$ that is part of a canonical path can only be the starting or the end point of the path. 
\end{enumerate}
    \end{lemma}
\begin{proof}
We start by designing the paths $\g_{xy}$ whose starting and end points $x$ and $y$ are both good vertices in $\bm{V_{\mathrm{good}}}$. According to \propref{nVertexConnectionHypercube}, there are $n$ pairwise internally disjoints paths of length $\leq n+1$ connecting $x$ and $y$. For each such path, the probability that at least one of its vertices is bad is $\leq n\cdot \frac{1}{64n}$. Hence, the probability that at least on of the $n$ disjoints paths contains only good vertices is $\geq 1-2^{-6n}$. We choose one such good path as the canonical path $\g_{xy}$ between $x$ and $y$. Via a union bound, we see that except with probability $\leq 2^{2n}\cdot 2^{-6n}=2^{-4n}$, this procedure yields a set of canonical paths $\{\g_{xy}\}$ between any pair of good vertices such that $|\g_{xy}|\leq n+1$.

Next, we consider paths connecting a bad vertex $x$ in $\bm{V_{\mathrm{bad}}}$ to another vertex $y$. For any  vertex $x \in \{0,1\}^n$, the probability that all of its neighbors are bad vertices in $V_0$ is $\leq (\frac{1}{64n})^n$. This means with probability $\geq 1-2^n\cdot 2^{-n\log(64n)}$, any bad vertex in $\bm{V_{\mathrm{bad}}}$ has a good neighbor in $\bm{V_{\mathrm{good}}}$. The canonical path $\g_{xy}$ between a bad vertex $x$ and any other vertex $y$ is defined to be the path that first connects $x$ to a designated good vertex $x'$ among its neighbors, and then connects $x'$ to $y$ (or if $y\in\bm{V_{\mathrm{bad}}}$, its designated good vertex $y'$) via the canonical path as constructed in the previous case. It is not hard to see that these canonical paths satisfy the properties 1. to 3. mentioned earlier. 
\end{proof}

\begin{proof}[Proof of \thmref{upperBoundCongestion}]
We upper bound the path congestion parameter $\bm{\r}(\G)$ for any edge $e=(e^+,e^-)$ that participates in a path in separate cases. From property 3. of \lemref{propertiesofCanonicalPaths}, we have that if both $e^+$ and $e^-$ are bad vertices in $\bm{V_{\mathrm{bad}}}$, then the edge $(e^+,e^-)$ is not part of any canonical path. 
Next, suppose $e^+ \in \bm{V_{\mathrm{bad}}}$ while $e^-\in \bm{V_{\mathrm{good}}}$ (the case of $e^-$ being bad instead is similar). This implies that $\bm{z}_{e^+}\leq \bm{z}_{e^-}$. Following Equation \eqref{eq:congestionRewrite}, we bound the quantity
\begin{align}
  &\frac{n}{\bm{S}}\cdot \frac{\bm{z}_{e^+}+\bm{z}_{e^-}}{\bm{z}_{e^+}\bm{z}_{e^-}} \sum_{\g_{xy}\ni e} \bm{z}(x)\bm{z}(y) |\gamma_{xy}|\\
  &=  \frac{n}{\bm{S}}\cdot\frac{\bm{z}_{e^+}+\bm{z}_{e^-}}{\bm{z}_{e^+}\bm{z}_{e^-}} \sum_{y \in \{0,1\}^n: y\neq e^+} \bm{z}(e^+)\bm{z}(y) |\gamma_{e^+y}|\nn\\
  &\leq \frac{2n}{\bm{S}}  \sum_{y \in \{0,1\}^n: y\neq e^+} \bm{z}(y) |\gamma_{e^+y}| && \text{$e^-\in \bm{V_{\mathrm{good}}}$ and $\bm{z}_{e^+}\leq \bm{z}_{e^-}$}\nn\\
  &\leq \frac{2n(n+3)}{\bm{S}}\cdot  \sum_{y \in \{0,1\}^n: y\neq e^+} \bm{z}(y)  && \text{Property 2. \lemref{propertiesofCanonicalPaths}}\nn\\
  &\leq \frac{2n(n+3)}{\bm{S}}\cdot  \sum_{y \in \{0,1\}^n} \bm{z}(y) \nn\\
   &\leq 2n(n+3).\label{eq:congestionboundcase1}
\end{align}
Now we consider the last case, when $e^+$ and $e^-$ are both good vertices in $\bm{V_{\mathrm{good}}}$.
We have
\begin{align}
  \frac{n}{\bm{S}}\cdot \frac{\bm{z}_{e^+}+\bm{z}_{e^-}}{\bm{z}_{e^+}\bm{z}_{e^-}} \sum_{\g_{xy}\ni e} \bm{z}(x)\bm{z}(y) |\gamma_{xy}|& \leq \frac{128n^2}{\bm{S}}\cdot\sum_{\g_{xy}\ni e} \bm{z}(x)\bm{z}(y) |\gamma_{xy}| && \bm{z}_{e^+} ,\bm{z}_{e^-} \in \bm{V_{\mathrm{good}}}\nn\\
  & \leq \frac{128n^2(n+3)}{\bm{S}}\cdot\sum_{\g_{xy}\ni e} \bm{z}(x)\bm{z}(y) && \text{Property 2. \lemref{propertiesofCanonicalPaths}}\label{eq:congestionboundcase2}
\end{align}
To upper bound the term $\sum_{\g_{xy}\ni e} \bm{z}(x)\bm{z}(y)$, recall that by the proof of \propref{nVertexConnectionHypercube}, for any canonical path $\g_{xy}$ between two good vertices $x,y \in \bm{V_{\mathrm{good}}}$, there exists an $i\in[n]$ such that $x$ and $y$ are connected first by flipping the $i$'th bit in $x$. 
Then, the bits of $x$ and $y$ are matched starting from the $i+1$'th bit and moving rightward and back to the $i$'th bit. 
For a fixed $i \in [n]$ and edge $e=(e^+, e^-)$, this procedure produces a set of paths (connecting different vertices)  that contain the edge $e=(e+,e-)$. 
We collectively denote such paths by $\G_i$.
We define $L_i$ (and $R_i \subseteq \{0,1\}^n$) to be the set of all starting (respectively end) vertices of the paths in $\G_i$. We have $|L_i|\cdot |R_i|\leq 2^n$ for all $i\in[n]$. 

Although for each pair of vertices, at most one of the $n$ paths generated by different choices of $i$ is used, we charitably upper bound $\sum_{\g_{xy}\ni e} \bm{z}(x)\bm{z}(y)$ by including all $n$ paths for any $x,y \in \{0,1\}^n$ including the bad vertices. Moreover, the bad vertices use a designated good neighbor to connect to other vertices. Again, we loosely upper bound this by multiplying the contribution of all vertices by a factor of $2$. Overall, we have
\begin{align}
  \frac{1}{\bm{S}} \cdot \sum_{\g_{xy}\ni e} \bm{z}(x)\bm{z}(y)\leq \frac{2^{n+2}}{\bm{S}} \cdot \sum_{i\in[n]} \left( \frac{1}{|L_i|}\sum_{x\in L_i} \bm{z}(x)\cdot \frac{1}{|R_i|}\sum_{y\in R_i}\bm{z}(y) \right)
\end{align}
Next, we apply the tail bound \eqref{eq:concentrationSubExp} in \propref{concentration} to $\bm{S}=\sum_{x\in \{0,1\}^n} \bm{z}(x)$ and also $\frac{1}{|L_i|}\sum_{x \in L_i} \bm{z}(x)$ and $\frac{1}{|R_i|} \sum_{x\in R_i} \bm{z}(x)$ for all $i\in [n]$.
We see that for a constant $c>1$, except with $\leq 2^{-3n}$ (or in fact up to a tighter bound $\leq 2^{-\Omega(n^c)}$), we have 
\begin{align}
     \text{when $i \in [c\cdot \log n, n- c\cdot\log n]$,\quad} \frac{1}{|L_i|}\sum_{x \in L_i} \bm{z}(x)\cdot \frac{1}{|R_i|} \sum_{x\in R_i} \bm{z}(x) \leq \mathcal{O}(1).
\end{align}
Similarly, with probability except with probability $2^{-3n}$, we have
\begin{align}
     \text{when $i \leq c\cdot \log n$ or $i \geq n - c\cdot \log n$,\quad} \frac{1}{|L_i|}\sum_{x \in L_i} \bm{z}(x)\cdot \frac{1}{|R_i|} \sum_{x\in R_i} \bm{z}(x) \leq \mathcal{O}(n).
\end{align}
Together, these bounds along with the concentration of $\bm{S} = \sum_{x \in \{0,1\}^n} \bm{z}(x)$ around $2^n$ imply that with probability $\geq 1 - 2^{-3n}$:
\begin{align}
   \frac{2^{n+2}}{\bm{S}}\cdot \sum_{i\in[n]} \left( \frac{1}{|L_i|}\sum_{x\in L_i} \bm{z}(x)\cdot \frac{1}{|R_i|}\sum_{y\in R_i}\bm{z}(y) \right) \leq \mathcal{O}(n\cdot \log n).
\end{align}
When combined with \eqref{eq:congestionboundcase2}, the following bound holds with probability $\geq 1-2^{-3n}$ for a fixed $e=(e^+,e^-)$:
\begin{align}
   \frac{n}{\bm{S}}\cdot \frac{\bm{z}_{e^+}+\bm{z}_{e^-}}{\bm{z}_{e^+}\bm{z}_{e^-}} \sum_{\g_{xy}\ni e} \bm{z}(x)\bm{z}(y) |\gamma_{xy}|&\leq \mathcal{O}(n^4 \cdot \log n). \nn
\end{align}
This bound along with \eqref{eq:congestionboundcase1} shows that for all edges $e=(e^+,e^-)$ (whose number we upper bound by $2^{2n}$) the congestion parameter is upper bound by $\r(\G)\leq \mathcal{O}(n^4 \cdot \log n)$ with probability $\geq 1-2^{-n}$.
\end{proof}

\subsection{A tighter analysis}\label{sec:tighterAnalysis}
To obtain a tight analysis of the spectral gap, we consider a notion known as `resistance' that improves and generalizes the path congestion method discussed in the warm-up analysis before. 

The idea of this technique is based on the multi-commodity flow in which we route a unit flow from any vertex $x$ to $y$ by splitting it among several paths that connect $x$ and $y$ such that no edge is congested. This is in contrast to the path congestion method where the flow only follows one such path $\g_{xy}$. 
More formally, let $\mathcal{P}_{xy}$ be set of simple directed paths connecting $x$ to $y$. Then a (multi-commodity) flow is a function $f: \cup_{x\neq y} \mathcal{P}_{xy} \mapsto \mathbb{R}$ such that $\sum_{p \in \mathcal{P}_{xy}}f(p)=1$ for all two distinct vertices $x\neq y$.
Given a flow $f$, we define resistance $R(f)$ by
\begin{align}
    R(f):= \max_{e \in E} \frac{1}{Q(e)} \sum_{x,y} \sum_{p\in \mathcal{P}_{xy}: p \ni e} \pi(x)\pi(y) f(p) |p|
\end{align}
where as before the weight $Q(e)=\pi(e^{+})P(e^{+},e^{-})$ of an edge $e=(e^+,e^-)$ is the probability of the transition $(e^+,e^-)$ occurring in the random walk. 
\begin{proposition}[Multi-commodity flows vs spectral gap, cf. \cite{sinclair1992improved}]
Let $\lambda_1$ be the second largest eigenvalue of the transition matrix $P$ for a reversible Markov chain. For any flow $f$, it holds that 
\begin{align}
        1-\lambda_1\geq \frac{1}{R(f)}.
\end{align}
\end{proposition}

We apply this proposition to the random walk given in \eqref{eq:TransitionMatrix-Haar}. Following \eqref{eq:congestionRewrite}, we can rewrite the resistance as

\begin{align}
    \bm{R}(f):= \frac{n}{\bm{S}}\max_{e=(e^+,e^-)} \frac{\bm{z}_{e^+}+\bm{z}_{e^-}}{\bm{z}_{e^+}\bm{z}_{e^-}}   \sum_{x,y} \sum_{p\in \mathcal{P}_{xy}: p \ni e} \bm{z}(x)\bm{z}(y) f(p) |p|.\label{eq:congestionRewriteResistance}
\end{align}

We are going to divide the vertices into two sets of good $\bm{V_{\mathrm{good}}}$ and bad $\bm{V}_{B}$ vertices. The bad vertices $\bm{V_{\mathrm{bad}}}$ are those whose weight $\bm{z}(x)$ is either too small or too large. More precisely, for two constant $c_{\ell}$ and $c_{u}$ that are fixed soon, we define 
\begin{align}
    \bm{V_{\mathrm{good}}}:=\{x\in V: c_{\ell}\leq \bm{z}(x)\leq c_{u}\},\quad \quad \bm{V_{\mathrm{bad}}}:=V\setminus \bm{V_{\mathrm{bad}}} \label{eq:goodVertices}
\end{align}
The value $\bm{z}(x)$ of vertex $x$ has an exponential distribution $\Pr(\bm{z})=e^{-\bm{z}}\cdot \iden[\bm{z}\geq 0]$. It follows form \eqref{eq:CDFexp} that for $c_\ell = 1/11, c_u = 5$, 
\begin{align}
    p:= \pr[x\in \bm{V_{\mathrm{good}}}]=e^{-c_\ell}-e^{-c_u}\geq \frac{9}{10}.\label{eq:probIntervalExp} 
\end{align}
Hence, for any edge that only involves two good vertices  $e^+,e^-\in \bm{V_{\mathrm{good}}}$, we have 
\begin{align}
    \bm{R}(f):= \frac{n}{\bm{S}}\frac{2c_u^2}{c_{\ell}}  \max_{e=(e^+,e^-)} \sum_{x,y} \sum_{p\in \mathcal{P}_{xy}: p \ni e}f(p) |p|.\label{eq:ResistanceSimplified}
\end{align}

In the next few sections, we prove the existence of a multi-commodity flow with $R(f) \leq \mathcal{O}(n^2)$, as stated in the following theorem. 

 \begin{theorem}[multi-commodity flow construction]\label{thm:multiCommodityFlowBound}
     There exists a multi-commodity flow $f$ defined for the random walk \eqref{eq:TransitionMatrix-Haar} such that the resistance $\bm{R}(f)$ is with probability $1-2^{-c n}$ bounded by $\bm{R}(f) \leq c' n^2$ for some constants $c, c' \geq 0$.
\end{theorem}
The proof of this theorem is given in \appref{proofFlow}. Before proceeding to the proof, we need to lay out some facts about the Boolean hypercube. We then start our construction of multi-commodity flows between vertices. 

\subsubsection{Geometric facts about Boolean hypercube}\label{sec:GeometricFacts}

To bound the resistance $R(f)$ in \eqref{eq:congestionRewriteResistance}, we consider an $n$-bit Boolean hypercube, where each vertex is deleted with probability $p$ ($p$ is a small value between $0$ and $1$ given by \eqref{eq:probIntervalExp}).
The removed vertices form the set $\bm{V_{\mathrm{bad}}}$ while the remaining ones form $\bm{V_{\mathrm{good}}}$.
For now, we ignore the non-uniform weights $\bm{z}$ in our initial setup and and assume all edges in the hypercube have weight one.
We would like to create a multi-commodity flow between every pair of vertices that have not been deleted.
The goal is to find a multi-commodity flow that does not congest at any particular edge.

We define a \emph{local escape property} that exists with high probability and allows us to spread the flow from any vertex $x'$ to its neighbors by using edges within some constant Hamming distance of $x'$. This is done in a similar way to \cite{McDiarmidHypercube} where instead of vertices, some of the edges are removed. In order to route a unit flow from $x$ to $y$, we apply this local escape property repeatedly to carry the flow from the set of vertices with Hamming distance $k$ from the starting vertex $x$ to those with Hamming distance $k+1$ until we reach $y$.   
Even though the bad vertices are removed from the set of available vertices, we find it helpful to consider a \emph{fictitious} vertex in place of each bad vertex. A flow to (or from) each fictitious vertex means splitting (or collecting) the flow equally between all of its \emph{good} neighbors which belong to $\bm{V_{\mathrm{good}}}$. 

\begin{lemma}[Local escape property]\label{lem:localEscapeProperty}
    There are constants $\alpha, \beta > 0$ such that with probability at least $1 - 2^{-\beta n}$: (1) each vertex $v \in \{0, 1\}^n$ has at least $\alpha n$ good neighbors, and (2) there are at least $\alpha n$ simple pairwise internally disjoint paths of length $\leq 5$ connecting any two good vertices within Hamming distance $3$ of each other. 
\end{lemma}
\begin{proof}
    The proof of statement (1) follows immediately from the tail bound \eqref{eq:concentrationBin} in \propref{concentration}. 
    The proof of statement (2) uses the same construction of internally disjoint paths as in the proof of \propref{nVertexConnectionHypercube}. That is, we repeat the following procedure for any $i\in [n]$: 
    Choose the $i$'th bit of the starting vertex $x \in \{0,1\}^n$ and flip it.
    Then moving rightward from bit $i+1$ to $n$ and back to $i$, flip any bits in which the current vertex differs from the end vertex $y \in \{0,1\}^n$.
    It is not hard to see that because the starting and end vertices are within Hamming distance $3$ of each other, the length of these paths is bounded by $5$. 
   
    Fix two good vertices $s$ and $t$ within Hamming distance $3$ of each other. 
    The bound \eqref{eq:probIntervalExp} implies that (conditioned on the end points being good vertices) the probability of any such paths consists only of good vertices is $p^4$ for $p \geq  9/10$. 
    From the tail bound \eqref{eq:concentrationBin} in \propref{concentration}, we have that since $p^4 > 1/2$, there exist constants $\alpha, \beta > 0$ so that except with probability at most $2^{-(1 + 2\beta)n}$, there are $\geq \alpha \cdot n$ such paths between $s$ and $t$. 
    There are at most $\mathcal{O}(n^3\cdot 2^n)$ pair of good vertices with Hamming distance $3$. We arrive at statement (2) by a union bound over  all these choices.
\end{proof}

Given an even $n$, consider $k = 0,\dots, n / 2 + 1$.
In the following, when we refer to distance $d(s, t)$, we are referring to the Hamming distance between two vertices $s$ and $t$.
We will often view the hypercube $\{0,1\}^n$ from the perspective of an \emph{origin} vertex $u$. We let a vertex layer $V_k(u)$ denote the set of vertices with Hamming distance $k$ to $u$. The set of edges between layers $V_k(u)$ and $V_{k+1}(u)$ is denoted by $E_{k+1}(u)$. It follows that 
$$|E_{k+1}| = (n-k)|V_k| = (k + 1) |V_{k+1}|.$$
We define the level $L(s)$ of a vertex $s$ as the distance $d(s,u)$ of $s$ to the origin~$u$.

For each edge $e$ on the hypercube, we distinguish between an up arrow $e^\uparrow$ or a down arrow $e^\downarrow$ based on a fixed origin $u$ and the hierarchy $\{V_k(u)\}_{k=0}^{n/2+1}$ defined by the Hamming distance to the origin $u$.
If an up arrow $e^{\uparrow}$ (or down arrow $e^{\downarrow}$) connects vertices $s$ and $t$, then we write $e^{\uparrow} = s \uparrow t$ (or $e^{\downarrow} = s \downarrow t$).
We define the level $L(e)$ of an edge $e = (s, t)$ (or an arrow $e^{\uparrow}, e^{\downarrow}$) to be the minimum of the levels of the left node $s$ and the right node $t$ of $e$.

Given an up arrow $s \uparrow t$, we consider $s'$ to be $s$ if $s \in \bm{V_{\mathrm{good}}}$; otherwise, $s'$ is distance one away from $s$.
We also define the length $\ell_{s}$ as follows
\begin{equation}
\ell_s = \begin{cases}
    0, & s \in \bm{V_{\mathrm{good}}} \\
    1, & s \in \bm{V_{\mathrm{bad}}}
\end{cases}.\nn
\end{equation}
We let $t'$ to be $t$ if $t \in \bm{V_{\mathrm{good}}}$; otherwise, $t'$ is distance one away from $t$.
Similarly, we define the length $\ell_{t}$ as follows
\begin{equation}
\ell_t = \begin{cases}
    0, & t \in \bm{V_{\mathrm{good}}} \\
    1, & t \in \bm{V_{\mathrm{bad}}}
\end{cases}.\nn
\end{equation}
We also consider $s^*$ (resp. $t^*$) to be distance one away from $s'$ (resp. $t'$). The relevance of vertices $s'$ or $t'$ is that when $s$ or $t$ are bad, we instead route the flow from $s'$ to $t'$. There are, of course, many choices of vertices $s'$ and $t'$, so we divide the flow equally between them. 

The local escape property in \lemref{localEscapeProperty} yields a set of paths connecting the two vertices $s'$ and $t'$. 
Some examples of such paths are shown in \fig{cycles} in red, green, and yellow.
Each of these paths together with the path $(t', t, s, s')$ (shown as dotted lines in \fig{cycles}) forms a cycle.
We use $C^*$ when referring to one such cycle. The number of edges in $C^*$ is denoted by $|C^*|$.
To avoid inconsistency, in our later analysis we separately treat the case when the path connecting $s$ and $t$ is simply the edge $(s,t)$. 
This happens only when $s$ and $t$ are both $\in \bm{V_{\mathrm{good}}}$. 

Our future analysis of the total flow in an edge $e$ depends on a delicate balance between the number of cycles $C^*$ in which edge $e$ participates and the level of this edge $L(e)$.
In the remainder of this section, we first look more closely at the number of cycles involving an edge $e$ in \lemref{cycleCount}. Then, we prove an upper bound on the level $L(e)$ of edge $e$ in a cycle in \lemref{BoundingLevel}.

\begin{definition}[Cycle type]\label{def:CycleType}
    We can enumerate the vertices of a cycle $C^*$ starting from $s$ and ending at $t$ by assigning an index $i = 1, \dots, |C^*|$ to each vertex. The type of a cycle is determined by identifying (1) the level of each of the cycle's vertices $L(v_i)$ relative to the level $L(s)$, (2) the index $i$ assigned to vertices $t'$, $t$, $s$, and $s'$. 
\end{definition}
There are in fact various \emph{types} of cycles $C^*$ involving any given pair $(s,t)$. 
We have included some instances of these cycles in \fig{cycles}. 
There are, however, only a constant number of types of cycles $C^*$ since $|C^*|=\mathcal{O}(1)$ as corroborated more in \lemref{cycleCount}.

\begin{figure}[h!]
     \centering
     \includegraphics[width=.85\textwidth]{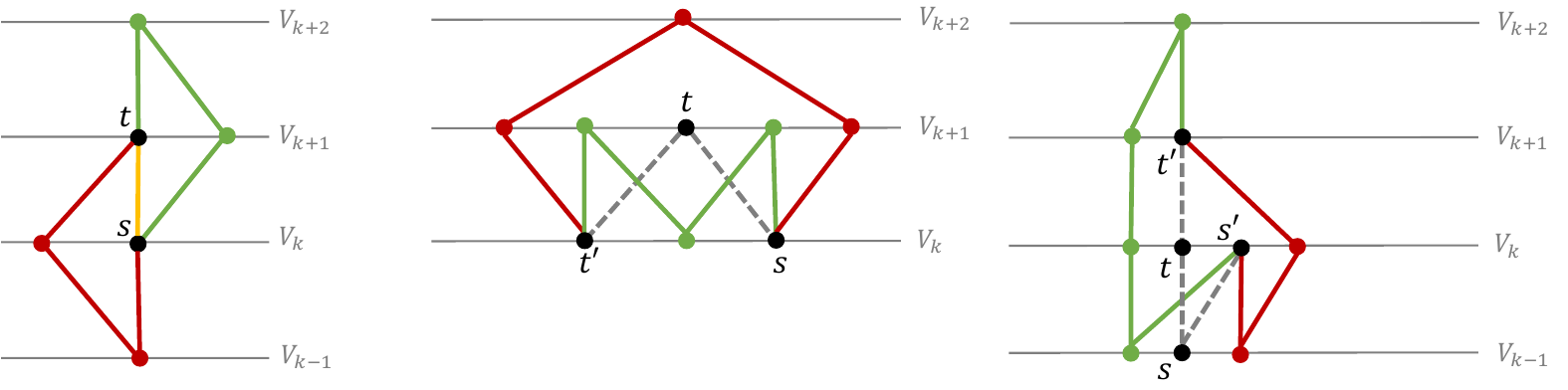}
     \caption{Various types of cycles $C^*$ for a given $(t',t,s,s')$. From left to right: when $s$ and $t$ are both good, $s$ is good but $t$ is bad, and $s$ and $t$ are both bad. The cycle formed by the green path and $(s,t)$ on the left has the following type: (1) the relative levels $L(v_i)-L(s)$ are $(0,1,2,1)$ in order and $|C^*|=4$, (2) $v_1=s, v_4=t$.}
    \label{fig:cycles}
 \end{figure}

\begin{lemma}\label{lem:cycleCount}
Fix an edge $e$. Let $C_e^*$ denote the set of all cycles $C^*$ involving this edge $e$ which can be constructed by applying the local escape property to some arbitrary choice of vertices $(s,t)$. It holds that:
\begin{enumerate}
    \item The length $|C^*|$ of any cycle in $C^*_e$ is either $4$, $6$, or $8$.
    \item Consider $\ell \in \{4,6,8\}$. It is possible to choose a path of length $\ell/2$ with the following properties: This path includes edge $e$. The number of length-$\ell$ cycles in $C_{e}^*$ with a fixed type that include this path is a constant $\leq 24$. 
\end{enumerate}
\end{lemma}

\begin{proof}
The paths connecting $s'$ and $t'$ are constructed using \lemref{localEscapeProperty} and have length either equal to $d(s',t')+2$ or $d(s',t')$. The choices of $d(s',t')$ are as follows:
When $\ell_s + \ell_t = 0$ (resp. $\ell_s + \ell_t = 1$), we have $d(s',t')=1$ (resp. $d(s',t')=2)$. Similarly, when $\ell_s + \ell_t = 2$, the distance $d(s',t') \in \{1,3\}$. 
Enumerating all these cases, we see that $|C^*| \in \{4, 6, 8\}$ (or that edge $e$ corresponds to the single edge $(s,t)$). 

We next prove the second statement. 
When enumerating the cycles of length $\ell$ in $C^*_e$, we may first choose a path of length $\ell/2$ that contains $e$. This can be done in a way that the remaining edges in the cycle form some predefined paths which have $\mathcal{O}(1)$ count. These paths are constructed explicitly in the rest of the proof.

We start with when $\ell_s + \ell_t = 2$ and $d(s',t') = 3$.
A direct examination shows that we either have $d(s^*,t') = d(s',t^*) = 4$ or $d(s^*,t') = d(s',t^*) = 2$.
We first assume $d(s^*,t') = d(s',t^*) = 4$. 
Consider a path of length $4$ that includes edge $e \neq (s^*,s')$ and whose end points are the two vertices $s^*$ and $t'$ (if $e = (s^*, s')$ then the end points are instead $s'$ and $t^*$). 
We can complete this path to a length-$8$ cycle $C^*$ by adding the complementary path $(s^*, s', s, t, t')$ (or $(s', s, t, t',t^*)$ if $e = (s^*, t')$). The number of choices of the complementary path $(s^*, s', s, t, t')$ (or $(s', s, t, t',t^*)$) is $4!=24$.
This is because these are length-$4$ paths that connect two vertices of Hamming distance $4$.

Next, we assume $d(s^*,t') = d(s',t^*) = 2$. Then given any edge $e$, we consider a path involving edge $e$ with end points corresponding to $s'$ and $t'$.
To turn this into a cycle $C^*$, we consider the complementary $(s',s,t,t')$. As before since $d(s', t') = 3$, the number of such complementary paths is $3!=6$.

A very similar argument also applies to when $\ell_s + \ell_t = 1$ and $d(s',t')=2$ or when $\ell_s + \ell_t = 0$ and $d(s',t')=1$. 
Hence, we move on to the remaining cases starting with when $\ell_s + \ell_t = 2$ but $d(s',t') = 1$.
Here, we may have a cycle $C^*$ of length $4$ with the edge sequence $(s', t') - (t', t) - (t,s) - (s, s')$. 
The two vertices $t$ and $s'$ have distance $d(s',t)=2$. 
Therefore, there are only $2$ paths $(s', s, t)$ of length $2$ connecting them. 
If instead, we have a cycle $C^*$ of length $6$, we can fix vertex pairs $(s',t')$ where $d(s',t') = 3$.
This implies that the number of paths $(s', s, t,t')$ is again just a constant equal to $3!=6$.

\end{proof}

We call the path $(s^*, s', s, t, t')$ the initial path $P^*$ between $s^*$ and $t'$.
The length of $P^*$ is $\ell^* = 2 + \ell_{s} + \ell_{t}$.
By definition, there exists a path of length $\ell^*$ connecting $s^*$ and $t'$.
The $(\ell^* - 1)$-th edge on the initial path $P^*$ is the up arrow $s \uparrow t$.

Consider an arbitrary length-$\ell^*$ path $P$ from $s^*$ to $t'$ and an index $i = 1, \ldots, \ell^*$ for an arrow on the path $P$.
We denote $P_i$ to be the $i$-th arrow on $P$. We define the six numbers $\mathds{1}_\uparrow(P, i), n_{\uparrow, a}(P, i), n_{\downarrow, a}(P, i), n_{\uparrow, b}(P, i), n_{\downarrow, b}(P, i), n_{\uparrow}(P)$:
\begin{align}
n_{\uparrow}(P) &= \mbox{number of $\uparrow$ in } P,\nn \\
\mathds{1}_\uparrow(P, i) &= \begin{cases}
    1, & P_i \mbox{ is an up arrow} \\
    0, & P_i: \mbox{ is a down arrow}
\end{cases},\nn \\
n_{\uparrow/\downarrow, a/b}(P, i) &= \mbox{number of $\uparrow/\downarrow$ in } P \mbox{ after/before } P_i.\nn
\end{align}
We can bound the level $L(P_i)$ of $P_i$ using the level of $s \uparrow t$ and the above numbers.
\begin{lemma}[Bounding the level]\label{lem:BoundingLevel}
Given an arbitrary length-$\ell^*$ path $P$ from $s^*$ to $t'$ and $i=1,\dots,\ell^*$, we have
\begin{align}
    L(P_i) &\leq L(s \uparrow t) + n_{\uparrow, b}(P, i) + n_{\downarrow, a}(P, i).\nn
\end{align}
\end{lemma}
\begin{proof}
    The proof traces the level from $s$ to $t$, then to $t'$ (which is the right node of $P_\ell$), followed by the left node of $P_\ell$, the left node of $P_{\ell-1}$ $\ldots$, to the left node of $P_i$.
    Because $P$ and $P^*$ share the same endpoint $s^*, t'$, we have 
    \begin{equation}
        L(t') - L(s^*) = n_{\uparrow}(P) - n_{\downarrow}(P) = n_{\uparrow}(P^*) - n_{\downarrow}(P^*).\nn
    \end{equation}
    From $n_{\uparrow}(P) + n_{\downarrow}(P) = \ell^* = n_{\uparrow}(P^*) + n_{\downarrow}(P^*)$, we have $n_{\uparrow}(P) = n_{\uparrow}(P^*)$.
    Next, we separately consider the following two cases.
    \begin{enumerate}
        \item $n_{\uparrow}(P) = n_{\uparrow}(P^*) \leq 1$: Because $n_{\uparrow}(P^*) \geq 1$, we know that $n_{\uparrow}(P^*) = 1$ and there can only be a single up arrow $s \uparrow t$ in $P^*$.
        Hence, the level of $t'$ is upper bounded by $L(s \uparrow t) + 1$.
        Tracing through the path $P$ backward from $t'$, every time we encounter an up arrow in $P$, the level decreases, and every time we encounter a down arrow $P$, the level increases.
        Together, the level of the left node of $P_{i+1}$ is upper bounded by $L(s \uparrow t) + 1 + n_{\downarrow, a}(P, i) - n_{\uparrow, a}(P, i)$.
        \item $n_{\uparrow}(P) = n_{\uparrow}(P^*) \geq 2$: $P^*_{\ell^*}$ can be an up arrow or a down arrow. The level of the left node of $P_{i+1}$ is upper bounded by $L(s \uparrow t) + 2 + n_{\downarrow, a}(P, i) - n_{\uparrow, a}(P, i)$.
    \end{enumerate}
    In both cases, we can upper bound the level of the left node of $P_{i+1}$ by
    $L(s \uparrow t) + 1 + \mathds{1}[ n_{\uparrow}(P) \geq 2] + n_{\downarrow, a}(P, i) - n_{\uparrow, a}(P, i)$.

    If $P_i$ is an up arrow, then the level of $P_i$ is the level of the left node of $P_{i+1}$ minus one; if $P_i$ is a down arrow, then the level of $P_i$ is the level of the left node of $P_{i+1}$. Together, we have
    \begin{equation}
        L(P_i) \leq L(s \uparrow t) + 1 + \mathds{1}[ n_{\uparrow}(P) \geq 2] + n_{\downarrow, a}(P, i) - n_{\uparrow, a}(P, i) - \mathds{1}_\uparrow(P, i).\nn
    \end{equation}
    Finally, we recall that $n_{\uparrow}(P) = n_{\uparrow, a}(P, i) + n_{\uparrow, b}(P, i) + \mathds{1}_\uparrow(P, i)$ and $n_{\uparrow}(P) = n_{\uparrow}(P^*) \geq 1$.
    If $n_{\uparrow}(P) < 2$, then $n_{\uparrow}(P) = 1$ and we have the following identity,
    \begin{equation}
        1 + \mathds{1}[ n_{\uparrow}(P) \geq 2] + n_{\downarrow, a}(P, i) - n_{\uparrow, a}(P, i) - \mathds{1}_\uparrow(P, i) = n_{\uparrow, b}(P, i) + n_{\downarrow, a}(P, i).\nn
    \end{equation}
    If $n_{\uparrow}(P) \geq 2$, we can again see that
    \begin{align}
        &1 + \mathds{1}[ n_{\uparrow}(P) \geq 2] + n_{\downarrow, a}(P, i) - n_{\uparrow, a}(P, i) - \mathds{1}_\uparrow(P, i) \nn \\
        &\leq n_{\uparrow}(P) + n_{\downarrow, a}(P, i) - n_{\uparrow, a}(P, i) - \mathds{1}_\uparrow(P, i) = n_{\uparrow, b}(P, i) + n_{\downarrow, a}(P, i).\nn
    \end{align}
    This concludes the proof.
\end{proof}

\subsubsection{Antipodal vertices}\label{sec:antipodalFlow}

The capacity $\capacity_{uv}(e)$ of an edge with respect to two vertices $u$ and $v$ is defined as the maximum flow required through that edge when we route a unit flow between $u$ and $v$. We also define the antipodal capacity $\capacity_{\mathrm{antipodal}}(e)$ of an edge to be the sum of the capacities $\capacity_{u\bar{u}}(e)$ for all choices of antipodal vertex pairs $u$ and $\bar{u}$. 

\begin{lemma}[cf. Lemma 19 of \cite{McDiarmidHypercube}]\label{lem:capacitybound}
 Consider a set of good vertices $\bm{V_{\mathrm{good}}}$ as in \eqref{eq:goodVertices} with the local escape property. 
 Consider a unit flow between any pair of antipodal vertices $u, \bar{u}$ such that for any edge $e \in |E_k|$ with a non-zero flow, we have $\capacity_{u\bar{u}}(e)\leq \frac{c}{|E_k|}$ for some constant $c > 0$. Then, it holds that $\capacity_{\mathrm{antipodal}}(e)\leq 2c$.
\end{lemma}
We next construct a flow that satisfies the requirements of \lemref{capacitybound}. 
We call a unit flow between two vertex layers $V_k(u)$ and $V_{k+1}(u)$ \emph{balanced} if the net outflow from each vertex in $V_k$ equals $\frac{1}{|V_k|}$ and the net inflow to each vertex in $V_{k+1}$ equals $\frac{1}{|V_{k+1}|}$.

\begin{proposition}\label{prop:antipodalflow}
   Consider an antipodal pair of good vertices $u$ and $\bar{u}$. Assuming the local escape property, there is a unit flow from $u$ to $\bar{u}$ such that the flow between any adjacent layers $V_k(u)$ and $V_{k+1}(u)$ is balanced and the capacity of each edge $e \in E_{k}$ for $k \leq n$ with non-zero flow satisfies $\capacity_u(e)\leq \frac{c}{|E_k|}$ for some constant $c > 0$.
\end{proposition}

 \begin{proof}
    Using the symmetry of the hypercube $\{0,1\}^n$, we can without loss of generality assume that the antipodal vertices $u, \bar{u}$ are the all-zeros $0^n$ and the all-ones $1^n$ vertices. To simplify the notation, we will drop the dependencies on $u$. 

    For most of the proof we only consider routing a flow from $0^n$ vertex to vertices $V_k$ with $k\leq n/2$. Exactly the same flow can be considered in the `reversed' direction from $V_k$ to $1^n$. This allows to carry a unit flow from $0^n$ to $1^n$.
    
   As stated before, if a starting vertex $s$ or an end vertex $t$ belongs to the set of bad vertices $\bm{V_{\mathrm{bad}}}$, we place a fictitious vertex in their place. 
   Hence, bounding the flow through an edge $e$ depends on whether the starting and end vertices $s$ and $t$ belong to $\bm{V_{\mathrm{good}}}$ or $\bm{V_{\mathrm{bad}}}$. 
   We consider each case separately. 
     We also first focus on paths where the difference between levels
     $$\Delta = L(e) - L(s \uparrow t)$$ 
     is positive $\Delta > 0$. Treating $\Delta \leq 0$ is done later in the proof via a simpler analysis.  
     Although not always mentioned explicitly, our analysis uses the fact that both the length of cycles $C^*$ and the number of cycle types, as in \defref{CycleType}, are an $\mathcal{O}(1)$ constant. 
     This means that for any edge $e$, there are only $\mathcal{O}(1)$ distinct choices for the relative position of an edge $e$ within a cycle $C^*$ and the type of the cycle $C^*$. 
     
     \paragraph{(1) $s, t\in \bm{V_{\mathrm{good}}}$ and $\Delta > 0$:} 
     Assume we have already routed a balanced unit flow to vertex layer $V_{k}$ for $k\leq n/2-1$. 
     Our goal is to transfer a flow of volume $\frac{1}{|V_k|}$ from a vertex $s \in V_{k}$ to all its neighbors in $V_{k+1}$ in a manner that guarantees a balanced flow to layer $V_{k+1}$.
     We divide the flow of volume $\frac{1}{|V_k|}$ equally between all the $n-k$ neighbors of $s$ in $V_{k+1}$, including any bad vertex.
     This flow to a neighbor $t\in V_{k+1}$ is further divided between all the internally disjoint paths connecting $s \in V_k$ to $t\in V_{k+1}$. 
     The local escape property in \lemref{localEscapeProperty} states that there are at least $\frac{2}{3}n$ such paths.
    This means that the maximum flow through a given edge $e$ involved in one of these paths that connect vertex $s \in V_k$ to $t \in V_{k+1}$, is bounded by $g_1$ given by
    \begin{align}
        g_1=\frac{1}{|V_k|}\cdot \frac{1}{\alpha n}\cdot \frac{1}{n-k}.\label{eq:gk}
        \end{align}
    This can be further bounded from above in terms of the relative distance of level $L(e)$ and level $L(s \uparrow t) = k$. It holds that 
       \begin{align}
          g_1 &= \frac{|V_{L(e)}|}{|V_k|}\cdot \frac{1}{\alpha n}\cdot \frac{1}{(n-L(e))|V_{L(e)}|}\cdot \frac{n-L(e)}{n-k}\nn\\
         &\leq \left(\frac{n-k}{k}\right)^{\Delta}\cdot \frac{1}{\alpha n}\cdot \frac{1}{|E_{L(e)+1}|} \cdot \left(1-\frac{\Delta}{n-k}\right).\label{eq:bouninggk}
    \end{align}
    Here, the first equality follows from Equation \eqref{eq:gk}. 
    The second inequality is obtained using the bound $\frac{|V_{L(e)}|}{|V_k|}\leq (\frac{n-k}{k})^{\Delta}$ and the fact that $|E_{L(e)+1}|=(n-L(e)) |V_{L(e)}|$. 
    
    Given an edge $e$, there may be more than one pair of vertices $s$ and $t$ with their flow moving through edge $e$.
    We next bound the contribution of all such choices. 
    Any edge $e$ is either $(s,t)$ that directly connects $s$ and $t$ or is part of a longer path that belongs to a cycle $C^*$ of length $4$ as in \lemref{cycleCount}. 
    The former case in which $\Delta = 0$ is analyzed at the end of the proof. 
   
    In the latter case, we use the second assertion of \lemref{cycleCount} to count the number of cycles $C^*$ that contain both edge $e$ and some vertices $s,t \in \bm{V_{\mathrm{good}}}$. 
    This lemma implies that if we form a path of length $2$ by choosing an edge connected to $e$, then there are only $\mathcal{O}(1)$ consistent length-$4$ cycles $C^*$ of a particular type (see \defref{CycleType}). 
    If edge $e = (s,s^*)$, then $\Delta \leq 0$ which is considered later.
    When $e \neq (s,s^*)$, this edge is the $i$'th edge of a length-$2$ path $P$ between $s^*$ and $t'=t$ for $i = 1,2$.
    \lemref{BoundingLevel} shows that $L(P_i) - L(s \uparrow t) \leq n_{\uparrow, b}(P, i) + n_{\downarrow, a}(P, i)$. 
    Since path $P$ has length $\ell^* = 2$, we have $n_{\uparrow, b}(P, i) + n_{\downarrow, a}(P, i) \leq 1$, and hence, $\Delta = L(P_i) - L(s \uparrow t) \leq 1$. 
    The number of length-$2$ paths $P$ is $\leq k+2$ otherwise $\Delta \leq 0$, contracting our assumption. This means that the total count of relevant cycles $C^*$ of any type is~$\mathcal{O}(k)$. 
    We see that the net flow in \eqref{eq:gk} in this case is bounded by
    \begin{align}
        g_1 = \mathcal{O}\left( \frac{n-k}{k} \cdot \frac{1}{n}\cdot \frac{1}{|E_{L(e)+1}|} \cdot k\right)  = \mathcal{O}\left(\frac{1}{|E_{L(e)+1}|}\right).
    \end{align}
 
  \paragraph{(2) $s\in \bm{V_{\mathrm{good}}}, t\in \bm{V_{\mathrm{bad}}}$ or $s\in \bm{V_{\mathrm{bad}}}, t\in \bm{V_{\mathrm{good}}}$ and $\Delta > 0$:} If $t \in \bm{V_{\mathrm{bad}}}$, then the flow bounded in \eqref{eq:gk} is additionally divided equally between all the good neighbors of $t$.
  If $s \in\bm{V_{\mathrm{bad}}}$, then the flow has been previously routed equally to all the good neighbors of $s$, and now will be moved to $t$.
  In either case, the local escape property states that there are at least $\alpha n$ good neighbors for some constant $\alpha > 0$.
  Conditioned on this and similar to inequality \eqref{eq:bouninggk}, the flow through any edge $e$ is bounded by
  \begin{align}
        g_2=\frac{1}{|V_k|}\cdot \left(\frac{1}{\alpha n}\right)^2 \cdot \frac{1}{n-k}
        \leq \left(\frac{n-k}{k}\right)^{\Delta}\cdot \left(\frac{1}{\alpha n}\right)^2 \cdot \frac{1}{|E_{L(e)+1}|} \cdot \left(1-\frac{\Delta}{n-k}\right),\label{eq:gk2}
    \end{align}
    where as before $L(s\uparrow t) = k$ and $\Delta = L(e) - k$. 
    To bound the contribution of different choices of $s$ and $t$ for a given edge $e$, we again start with the characterization of the cycles $C^*$ in \lemref{cycleCount}. 
    It is evident from the proof of this lemma that the length of these cycles is either $4$ or $6$. Direct inspection shows that when $|C^*| = 4$, then $\Delta \leq 1$. 
    The number of such cycles is $\mathcal{O}(n-k)$. 
    This is because there are $\leq n-k+1$ choices for an edge connected to edge $e$, and having fixed this length-$2$ path, there are only $\mathcal{O}(1)$ consistent cycles of any type according to \lemref{cycleCount}. 
    Overall the total contribution $g^{(1)}_2$ of this case is 
      \begin{align}
        g^{(1)}_2 \leq \mathcal{O}\left( \frac{n-k}{k} \cdot \frac{1}{n^2}\cdot \frac{1}{|E_{L(e)+1}|} \cdot (n-k)\right) 
        \leq \mathcal{O}\left( \frac{1}{k} \cdot \frac{1}{|E_{L(e)+1}|}\right).
    \end{align}
    Now suppose the length of a cycle is $|C^*| = 6$.
    If $e = (s', s^*)$, the only way that $\Delta > 0$ occurs is when $s$ is a bad vertex and $L(s\uparrow s') = L(s\uparrow t) = k$. Then it is possible to have $L(s' \uparrow s^*) = k + 1$ and $\Delta = 1$. But if this is true, then $n_{\downarrow, a}(P, i) = 1$. Thus, the net flow is bounded by
        \begin{align}
          g^{(2)}_2 \leq \mathcal{O}\left(\frac{1}{|E_{L(e)+1}|} \cdot \frac{1}{n^2}\ \frac{n-k}{k} \cdot k(n-k) \right)
          \leq \mathcal{O}\left(\frac{1}{|E_{L(e)+1}|} \right).
    \end{align}
    The last case is when $e \neq (s', s^*)$. \lemref{BoundingLevel} shows that $\Delta \leq  n_{\uparrow, b}(P, i) + n_{\downarrow, a}(P, i)$ when edge $e$ is the $i$'th edge of the length $3$ path $P$. When $n_{\uparrow, b}(P, i) + n_{\downarrow, a}(P, i) = 2$, the number of compatible cycles is $\mathcal{O}(k^2)$ and $\Delta = 2$. When $n_{\uparrow, b}(P, i) + n_{\downarrow, a}(P, i) = 1$, the number of compatible cycles is $\mathcal{O}(k\cdot (n-k))$ and $\Delta = 1$. Finally if $n_{\uparrow, b}(P, i) = n_{\downarrow, a}(P, i) = 0$, the number of compatible cycles is $\mathcal{O}((n-k)^2)$ and $\Delta = 0$. In each case, multiplying the number of cycles with the flow \eqref{eq:gk2} results in similar cancellations. The net contribution of $\Delta> 0$ instances in the flow $g^{(3)}_2$ satisfies
    \begin{align}
          g^{(3)}_2 &\leq \frac{1}{|E_{L(e)+1}|} \cdot \frac{1}{n^2}\cdot \mathcal{O}\left( \left(\frac{n-k}{k}\right)^2 \cdot k^2 + \frac{n-k}{k} \cdot k(n-k) \right) \nn\\
          &\leq \frac{1}{|E_{L(e)+1}|} \cdot \frac{1}{n^2}\cdot \mathcal{O}\left((n-k)^2\right) \leq \mathcal{O}\left(\frac{1}{|E_{L(e)+1}|} \right).
    \end{align}
    Overall $g_2^{(1)} + g_2^{(2)} + g_2^{(3)} \leq \mathcal{O}\left(\frac{1}{|E_{L(e)+1}|} \right)$.

\paragraph{(3) $s,t \in \bm{V_{\mathrm{bad}}}$ and $\Delta > 0$:}
    Here, we want to move a flow of volume $\frac{1}{|V_k|}$ already divided between the good neighbors of $s$ to the good neighbors of $t$. 
    There are at least $\alpha n$ good neighbors for both vertices $s$ and $t$ according to the local escape property. 
    The flow through any edge $e$ is therefore bounded by
 \begin{align}
        g_3=\frac{1}{|V_k|}\cdot \left(\frac{1}{\alpha n}\right)^3 \cdot \frac{1}{n-k}
        \leq \left(\frac{n-k}{k}\right)^{\Delta}\cdot \left(\frac{1}{\alpha n}\right)^3 \cdot \frac{1}{|E_{L(e)+1}|} \cdot \left(1-\frac{\Delta}{n-k}\right).\label{eq:gk3}
    \end{align}
    The length of cycles $C^*$ according to \lemref{BoundingLevel} is $|C^*| \in \{4, 6, 8\}$. There are only $\mathcal{O}(k)$ cycles of length $|C^*| = 4$ that include a fixed edge $e$ and for which $\Delta > 0$ (in fact $\Delta = 1$ in this case). 
 \begin{align}
        g^{(1)}_3 = \mathcal{O}\left(\frac{n-k}{k} \cdot \frac{1}{n^3} \cdot \frac{1}{|E_{L(e)+1}|}\cdot k \right) = \mathcal{O}\left( \frac{1}{n^2} \cdot \frac{1}{|E_{L(e)+1}|} \right)
    \end{align}
    When $|C^*| = 6$, a direct inspection shows that positive level difference may be $\Delta = 1$ or $\Delta = 2$. 
    If $\Delta = 1$, we can loosely bound the flow with $\mathcal{O}\left( \frac{1}{k} \cdot \frac{1}{|E_{L(e)+1}|} \right)$ since there are at most $\mathcal{O}((n-k)^2$ relevant cycles $C^*$. Assume $\Delta = 2$. Then necessarily, there is a down arrow after edge $e$ in the cycle $C^*$. This limits the numebr of relavant cycles to $\mathcal{O}(k(n-k))$. Plugging this in \eqref{eq:gk3} gives a net flow~of 
 \begin{align}
        g^{(2)}_3 = \mathcal{O}\left( \left(\frac{n-k}{k}\right)^2 \cdot \frac{1}{n^3} \cdot \frac{1}{|E_{L(e)+1}|}\cdot k(n-k) \right) =  \mathcal{O}\left( \frac{1}{k} \cdot \frac{1}{|E_{L(e)+1}|} \right)
    \end{align}
    Next we consider $|C^*| = 8$. When edge $e = (s', s^*)$, one can directly certify that in order to have $\Delta > 0$, it must be true that $\Delta = 1$ and there is at least one edge with a down arrow after edge $e$ in cycle $C^*$. The latter fact implies the $\mathcal{O}(k(n-k)^2)$ bound on the number of cycles $C^*$. Hence, the contribution of this case to the flow is
 \begin{align}
        g^{(3)}_3 = \mathcal{O}\left(  \frac{n-k}{k}  \cdot \frac{1}{n^3} \cdot \frac{1}{|E_{L(e)+1}|}\cdot k(n-k)^2 \right) =  \mathcal{O}\left(\frac{1}{|E_{L(e)+1}|} \right)
    \end{align}
    Finally if $e \neq (s', s^*)$ while $|C^*|= 8$, we apply \lemref{BoundingLevel} to get $\Delta = L(P_i) - L(s \uparrow t) \leq n_{\uparrow, b}(P, i) + n_{\downarrow, a}(P, i)$. 
    From this we have that the number of relevant cycles is $\mathcal{O}(k^{\Delta}\cdot (n-k)^{3 - \Delta})$ for $\Delta = 1, 2, 3$. 
    The total flow due to these configurations is bounded by 
     \begin{align}
        g^{(4)}_3 = \mathcal{O}\left( \frac{1}{n^3} \cdot \frac{1}{|E_{L(e)+1}|}\right) \cdot \sum_{\Delta =1}^3 \left(\frac{n-k}{k}\right)^{\Delta}  \cdot k^{\Delta}(n-k)^{3-\Delta} =  \mathcal{O}\left(\frac{1}{|E_{L(e)+1}|} \right).
    \end{align}
    Together, the total flow is $g_3^{(1)} + g_3^{(2)} + g_3^{(3)} + g_3^{(4)} \leq \mathcal{O}\left(\frac{1}{|E_{L(e)+1}|} \right)$.
\paragraph{(4) $\Delta \leq 0$:} Since $\Delta \leq 0$, we can simply bound the term $\left(\frac{n-k}{k}\right)^{\Delta}$ by $1$. 
Moreover, the number of cycle types, the length of cycles  $C^*$, and $\Delta$ are all $\mathcal{O}(1)$ constants. Hence, we can bound the total flow in an edge $e$ due to the cases with $\Delta \leq 0$ by
 \begin{align}
        g_4 \leq \mathcal{O}(1)\cdot \sum_{i=1}^3 \left(\frac{1}{n}\right)^i \cdot (n-k)^i \frac{1}{|E_{L(e)+1}|}\cdot\leq \mathcal{O}\left( \frac{1}{|E_{L(e)+1}|}\right),
    \end{align}
    where the $\mathcal{O}((n-k)^i)$ term is a loose bound on the number of relevant cycles $C^*$ with $i=1$ for $s,t \in \bm{V_{\mathrm{good}}}$, $i=2$ for $s \in \bm{V_{\mathrm{good}}},t \in \bm{V_{\mathrm{bad}}}$ or $s \in \bm{V_{\mathrm{bad}}},t \in \bm{V_{\mathrm{good}}}$, and $i=3$ for $s,t \in \bm{V_{\mathrm{bad}}}$.

    When the contribution of the previous cases (1) to (4) are combined, the capacity of $\capacity_u(v)$ of an edge $e \in E_k$ satisfies $\capacity_u(v) \leq \mathcal{O}\left(\frac{1}{|E_k|}\right)$ which completes the proof. 
\end{proof}

\subsubsection{Non-antipodal vertices}\label{sec:NonantipodalFlow}

\begin{proposition}\label{prop:antipodalCap}
   Fix a vertex $u \in \bm{V_{\mathrm{good}}}$ and assume that the local escape property holds. One could route a unit flow between $u$ and any other vertex $v \in \bm{V_{\mathrm{good}}}$ such that the total flow $\sum_{v \in \bm{V_{\mathrm{good}}}}\capacity_{uv}(e)$ through an edge $e \in V_{k}(u)$ is bounded by $\frac{c}{|E_k|}\cdot 2^n$ for some constant $c > 0$ and any $k \leq n-1$.
\end{proposition}
\begin{proof}
    Suppose $d(u, v) \leq n/2$ and $n$ is even. Let $k = d(u, v)$.
    We know from the construction in \propref{antipodalflow} that one could route a balanced flow of volume $1$ to vertices in layer $V_k(u)$ (including $v$) such that each vertex receives a flow of volume $\frac{1}{|V_k(u)|}$.
    By scaling up the input flow and the edge capacities, we could use the same construction to route a flow of volume $|V_k(u)| = \binom{n}{k}$ to vertices in $V_k(u)$. This allows us to allocate a \emph{unit} flow to each of the vertices in $V_k(u)$ including $v$.

    If $d(u, v) > n/2$, then instead consider the antipodal vertex $\bar{u}$ and let $k = d(\bar{u} ,v)$. 
    In this case, we first route a flow of volume $\binom{n}{k}$ from $u$ to its antipodal vertex $\bar{u}$. 
    We then move that flow from $\bar{u}$ to $v$ using the same argument as in the previous case. 
    If $\bar{u} \in \bm{V_{\mathrm{bad}}}$,  then consider a fictitious vertex in its place and moving a flow to or from this fictitious vertex is equivalent to moving a flow to or from  all the neighbors of $\bar{u}$. 
    The total (scaled-up) capacity of an edge $e$ is bounded by
    \begin{align}
        \sum_{v \in \bm{V_{\mathrm{good}}}} c_{uv}(e) \leq \sum_{k=1}^{n-1}\frac{3c}{|E_{k}(u)|}\cdot \binom{n}{k}\leq \frac{3c}{|E_{k}(u)|}\cdot 2^n,
    \end{align}
    where the first bound including constant $c > 0$ follows from \propref{antipodalflow}.
\end{proof}

\subsection{Inserting bad vertices and bounding resistance}\label{sec:proofFlow}
Our final step in constructing a flow between vertices of the hypercube $\{0, 1\}^n$ is to insert the deleted vertices in $\bm{V}_{\mathrm{bad}}$ back in place and route a flow between them and other vertices. 

We achieve this simply by splitting a flow to or from a bad vertex between its good neighbors. 
We then use the multi-commodity flow construction in the last section to move this flow to any other (good or bad) vertex. 

We are now ready to put the previous steps together and prove \thmref{multiCommodityFlowBound} which gives a multi-commodity flow $f$ with $\bm{R}(f)\leq c' n$ that holds with probability $1-2^{-c n}$ for some constant $c, c' > 0$.

Consider an edge $e=(e^+,e^-)$. As in  \eqref{eq:congestionRewriteResistance}, we define the resistance of this edge by 
\begin{align}
    \bm{R}_e= \frac{n}{\bm{S}} \frac{\bm{z}_{e^+}+\bm{z}_{e^-}}{\bm{z}_{e^+}\bm{z}_{e^-}}   \sum_{x,y} \sum_{p\in \mathcal{P}_{xy}: p \ni e} \bm{z}(x)\bm{z}(y) f(p) |p|.\label{eq:congestionfore}
\end{align}
The flow that we have developed has length $|p| = \mathcal{O}(n)$. 
Suppose one of the endpoint of edge $e$ (e.g., $e^+$) is a bad vertex.
Since the bad vertices are removed from the flow constructed in \appref{antipodalFlow} and \appref{NonantipodalFlow}, edge $e$ only contributes to moving the flow from $e^+$ to other good vertices. 
Assuming the local escape property, this flow is equally divided between $\alpha n$ good neighbors of $e^+$.
Hence, we can bound the resistance $\bm{R}_e$ by 
\begin{align}
    \bm{R}_e &\leq \mathcal{O}(n^2)\cdot \frac{1}{\bm{S}} \cdot  \left(1 + \frac{\bm{z}(e^+)}{\bm{z}(e^-)}\right)\cdot \sum_{y \in V} \bm{z}(y)\cdot \frac{1}{\alpha n}.\nn
\end{align}
The weight $z(e^+)$ of the bad vertex is either $\leq c_{\ell}$ or $\geq c_u$ where as in Equation \eqref{eq:probIntervalExp}, we assume $c_\ell = 1/11$ and $c_u = 5$.
If $z(e^+) \leq c_\ell$, then we can further bound the resistance of edge $e$~by $$\bm{R}_e \leq \mathcal{O}(n)\cdot \frac{1}{\bm{S}} \cdot \sum_{y \in V} \bm{z}(y).$$
From bound \eqref{eq:concentrationSubExp} in \propref{concentration}, we have that $\frac{1}{\bm{S}}\sum_{y \in V} \bm{z}(y) \leq \mathcal{O}(1)$ with probability $\geq 1-2^{-\Omega(n)}$. 
Conditioned on this, the resistance of edge $e$ is bounded by $\bm{R}_e \leq \mathcal{O}(n)$ in this case. 

If an end point $e^+$ of edge $e$ is instead a bad vertex with $z(e^+)\geq c_u$, then with probability $1- 2^{-\Omega(n)}$, the weight $\bm{z}(e^+)\leq \mathcal{O}(n)$. 
This means the resistance is bounded by $$\bm{R}_e \leq \mathcal{O}(n^2)\cdot \frac{1}{\bm{S}} \cdot \sum_{y \in V} \bm{z}(y) \leq \mathcal{O}(n^2)$$ 
with probability $1- 2^{-\Omega(n)}$.

We next consider edge $e$ whose end points are good vertices. 
Fix vertex $x\in \bm{V}_{\mathrm{good}}$ and assume $e \in E_k(x)$. 
For any $y \in  \bm{V}_{\mathrm{good}}$, there is a flow of volume $\bm{z}(x)\bm{z}(y)\leq \mathcal{O}(1)$ between $x$ and $y$ that may traverse through edge $e$. 
In addition to this flow, vertices $x$ and $y$ are also responsible for carrying a $\mathcal{O}(1/n^2)$ fraction of the flow between any bad neighbors of $x$ and $y$.
This is because the flow to or from any bad vertex is divided between its good neighbors and carried away from there.

To bound the total input flow, we note that according to bound \eqref{eq:concentrationSubExp} in \propref{concentration}, it holds for a good vertex $u$ that $\sum_{v \in \bm{V}_{\mathrm{bad}}} z(v) \cdot \iden[d(u,v) = 1] \leq \mathcal{O}(n)$ with probability at least $1 - 2^{-c^{\prime \prime}n}$ for some constant $c^{\prime \prime} > 2$. By union bound, this holds for all good vertices with probability $\geq 1 - 2^{-\Omega(n)}$.
Conditioned on this, we get 
\begin{align}
    \bm{z}(x)\bm{z}(y)  + \mathcal{O}\left(\frac{1}{n^2}\right)\cdot \sum_{u,v \in \bm{V}_{\mathrm{bad}}} \bm{z}(u)\bm{z}(v) \cdot \iden[d(u, x) = 1, d(v, y) = 1] \leq \mathcal{O}(1).
\end{align}

We know from \propref{antipodalCap} that if the local escape property holds, then the total flow through the edge $e$ when moving a flow of $\mathcal{O}(1)$-volume from $x$ to any good vertex is 
$$\sum_{v \in \bm{V_{\mathrm{good}}}}\capacity_{uv}(e) \leq \mathcal{O}\left(\frac{1}{|E_k|}\cdot 2^n\right).$$
The resistance of edge $e$ is, therefore, bounded by 
\begin{align}
    \bm{R}_e= \mathcal{O}(n^2)\cdot \frac{1}{\bm{S}}\cdot   \sum_{k=1}^{n-1} \sum_{x \in \bm{V}_{\mathrm{good}}}\frac{1}{|E_k|}\cdot 2^n \cdot \iden[e \in E_k(x)] \leq \mathcal{O}(n^2)\label{eq:resistancebound2}
\end{align}
where the last inequality follows from \lemref{capacitybound}.
This concludes the proof of \thmref{multiCommodityFlowBound}.

\section{Enforcing fast mixing}\label{sec:enforcingMixing}

As discussed in the \secref{whyItWorks}, the relaxation time $\tau$ of the Markov chain associated with the target state $\ket{\psi}$ provides an upper bound on  the sample efficiency of our certification procedure. 
If this relaxation time is superpolynomial in $n$, then we have no rigorous guarantee that our protocol can certify the lab state $\rho$ efficiently.
What we can do, though, is modify the Markov chain (and hence the target state), replacing it by a new Markov chain that does mix rapidly. 
Thus we can efficiently certify whether the lab state has high fidelity with this adjusted target state. 
This new target state, however, could potentially be very different from the original state.

Our starting point is the notion of \emph{local escape property} introduced formally in \defref{localEscape}.
Examining this property offers a way to enforce fast mixing for a given query model $\Psi$ of the target state $\ket{\psi}$.
To achieve this, each time the model $\Psi$ is queried, we test the value $\pi(x)$ of the queried vertex $x \in \{0,1\}^n$ as well as its local neighborhood.
If these values satisfy local constraints mentioned in \defref{localEscape}, the original value $\Psi(x)$ is returned.
Otherwise, a value $\nu > 0$ is reported instead of $\Psi(x)$.
This effectively defines a modified query model $\Psi'$ (also shown in \fig{verificationprotocol}). 

When the query model $\Psi$ is known to return normalized values (i.e. $\sum_{x} |\Psi(x)|^2 = 1$), this parameter can be set as $\nu = 2^{-n}$.
More generally, the normalization factor $\sum_{x} |\Psi(x)|^2$ serves as a hyperparameter for this scheme and needs to be set in advance.
We show in \thmref{enforcingLocalEscape} that this strategy indeed ensures that the local escape property is satisfied for all the vertices. 
Hence, the modified model $\Psi'$ induces a random walk with a mixing time $\tau \leq \mathcal{O}(n^2)$.

More formally, we start by restating the notion of local escape property, which was considered earlier in \appref{GeometricFacts} in a slightly less general form.

\begin{definition}[Local escape property, generalized]\label{def:localEscape}
Consider the random walk over $\{0,1\}^n$ with stationary distribution $\pi(x)$ and transition probabilities defined in Equation \eqref{eq:TransitionMatrixIntro}.
We say a vertex $x\in \{0,1\}^n$ is $(c_{\ell}, c_u)$-good when $c_{\ell}\cdot 2^{-n} \leq \pi(x)\leq c_u \cdot 2^{-n}$ for some $c_{\ell} < 1$ and $c_u > 1$. 
This random walk is defined to satisfy the local escape property with parameters $(\alpha, c'_u, c_u, c_{\ell})$ when
 \begin{itemize}
      \item[(0)] Each vertex $x \in \{0,1\}^n$ has a probability bounded from above by $\pi(x) \leq c'_u\cdot 2^{-n}$,
      \item[(1)] Each vertex $x \in \{0, 1\}^n$ has at least $\alpha n$ neighbors that are $(c_{\ell}, c_u)$-good, and 
      \item[(2)] There are at least $\alpha n$  pairwise internally-disjoint paths of length $\leq 5$ connecting any two $(c_{\ell}, c_u)$-good vertices within Hamming distance $3$ of each other.     
\end{itemize}
\end{definition} 
We now introduce a scheme based on this property that ensures the fast mixing of the Markov chain associated with a query model $\Psi$.
In this scheme, every time the model $\Psi$ is queried with some vertex $x \in \{0,1\}$, we perform the following checks:

\begin{itemize}
    \item[-] First, we check if $\pi(x) \leq c'_u \cdot 2^{-n}$.
    \item[-] Next, we look at all the neighbors of vertex $x$. 
    For each neighbor, we check if it has at least $\alpha n$ neighbors that are considered $(c_{\ell}, c_u)$-good.
    \item[-] Finally, starting from vertex $x$, we find all the vertex pairs within a Hamming distance $3$ of each other. These vertices must be connected by internally-disjoint paths of length $\leq 5$ such that one of these paths includes vertex $x$. For each such vertex pair, we check if there are at least $\alpha n$ paths consisting entirely of $(c_{\ell}, c_u)$-good vertices.
\end{itemize}

If any of the above checks are not passed, we return the value $\nu = 2^{-n}$ for the vertex $x$. Otherwise, we output the original value $\Psi(x)$.
This value is chosen assuming $\sum_{x} |\Psi(x)|^2 = 1$.
We suppose the value of $\sum_{x} |\Psi(x)|^2$ is given as input to this scheme. 
For other normalization factors, the parameter $\nu$, $c'_u$, $c_u$, $c_{\ell}$ can be adjusted accordingly.

Note that the last condition can be efficiently checked by starting from the vertex $x$ and creating all the internally-disjoint paths of length $\leq 5$ that include $x$.
These paths are constructed before in the proof of \lemref{localEscapeProperty} via the following procedure repeated for any $i\in [n]$: 
Choose the $i$'th bit of the starting vertex and flip it.
Then moving rightward from bit $i+1$ to $n$ and back to $i$, flip any bits in which the current vertex differs from the end vertex.
 
\begin{theorem}[Enforcing local escape property implies fast mixing]\label{thm:enforcingLocalEscape}
    Suppose the scheme discussed above is performed on a query model $\Psi$ and the new model $\Psi'$ is obtained. 
    The following statements hold:
    \begin{itemize}
        \item[a)] When $c_{\ell}, c_u, \alpha$ are some constants and $c'_u \leq \mathcal{O}(n)$, the new model $\Psi'$ has a mixing time $\tau \leq \mathcal{O}(n^2)$.
        \item[b)] If the original query model $\Psi$ satisfies the local escape property with parameters $(\alpha, c'_u, c_u, c_{\ell})$, then it remains unchanged, i.e. $\Psi' = \Psi$. 
        \item[c)] For some constant values of $c_{\ell}, c_u, \alpha$ and $c'_u \leq \mathcal{O}(n)$, all but an exponentially small $2^{-\Omega(n)}$ fraction of $n$-qubit states $\ket{\psi}$ satisfy the local escape property with these parameters.
    \end{itemize} 
\end{theorem}

\begin{proof}
We show that the new model $\Psi'$ satisfies the local escape property. 
The first statement a) then follows from the arguement in the proof of \thmref{multiCommodityFlowBound} in \appref{tighterAnalysis}.
By construction, the enforcing scheme makes each vertex satisfy condition (0) by setting the value of the violating vertices equal to $\nu = 2^{-n}$.
The condition (1) is also satisfied because when a vertex $x$ violates it, the value of all of its neighbors are replaced by $\nu = 2^{-n}$.
Since $c_{\ell} <1$ and $c_u >1$, we have $c_{\ell}\cdot 2^{-n}\leq \nu \leq c_u\cdot 2^{-n}$.
This means that the vertex $x$ now has $n$ neighbors which are $(c_{\ell}, c_u)$-good.
Finally, the condition (2) holds because if a vertex pair violate it, all of the vertices on their connecting paths, take value $\nu = 1$.
This makes them all be $(c_{\ell}, c_u)$-good as well.

Moving on to the second claim b), it follows from the definition of the enforcing scheme defined above that the queries do not change if the model satisfies the local escape property.

The proof of the last statement c) is covered in \appref{GeometricFacts}, where we showed that for some constant values of $c_{\ell}, c_u, \alpha$ and $c'_u \leq \mathcal{O}(n)$, the Markov chain induced by almost all $n$-qubit quantum states satisfy the local escape property.
\end{proof}

\section{Benchmarking quantum devices}\label{sec:benchmarkingComplexity}

As mentioned in \secref{BenchmarkingQuantumDevices}, our certification protocol for Haar random states offers a benchmarking scheme where the fidelity of an $n$-qubit device in preparing an increasing level of computationally complex quantum states can be certified. 
Here we restate and prove this result.
\begin{theorem}[Benchmarking quantum states of tunable complexity, restatement]
    Suppose we have access to a quantum device capable of preparing $n$-qubit states.
    For any level $t \in \{2^k: k \in \{1, \dots, n\}\}$, we can benchmark the fidelity of the device in preparing quantum states with circuit complexity $\mathcal{C}_{\delta}(\ket{\psi}) \geq \widetilde{\Omega} (t)$ using $T = \mathcal{O}\left((\log^2 t/\epsilon)^2\right)$ single-qubit Pauli measurements via \protoref{verificationGeneral}.
\end{theorem}
\begin{proof}
    For any level of complexity $t \in \{2^k: k \in \{1,\dots, n\}\}$, we sample a $k$-qubit Haar random state.
    We then prepare this state on an arbitrary subset of size $k$ of the $n$ qubits that the device can prepare. 
    Any $k$-qubit quantum state can be prepared, using single-qubit and CNOT gates, in depth $\mathcal{O}(\frac{2^k}{k})$ and size $\mathcal{O}(2^k)$ \cite{sun2023asymptotically}.
    Using standard counting arguments, one can also show that $\Omega \left(\frac{2^k}{k + \log(k/\delta)}\right)$ many gates are needed for generating almost all $k$-qubit Haar random states up to an error $\delta$.
    As established in \appref{HaarSpectralGap}, the relaxation time of $k$-qubit Haar random states can be bounded by $\tau \leq \mathcal{O}(k^2)$ except for an exponentially small $e^{-\mathcal{O}(k)}$ fraction of the state.
    When $k = \mathcal{O}(\log(n))$, the local escape property in \lemref{localEscapeProperty} and therefore the relaxation time of the state can be directly checked. 
\end{proof}

As mentioned in \secref{BenchmarkingQuantumDevices} and \fig{plotsBenchmarking}, the shadow overlap is normalized such that in the ideal (noiseless) case, it evaluates to $1$, and in the case of  maximally mixed state, it evaluates to $1/2^n$.
For level-$m$ \protoref{verificationGeneral}, the normalized shadow overlap can be expressed as
\begin{align}
   \mathrm{normalized \ shadow \ overlap} = \frac{2^m}{2^m-1}\cdot \frac{d-1}{d}\cdot \left(\E[\bm{\omega}]- \frac{1}{2^m}\right) + \frac{1}{d}, \label{eq:NormalizedShadowOverlap}
\end{align}
with $d=2^n$ the dimension of the Hilbert space. 
This formula is a linear map from $[1/2^m, 1]$ to $[1/2^n, 1]$, which are the fidelity ranges for $m$-qubit systems and $n$-qubit systems.
It is easy to show that the original shadow overlap $\E[\bm{\omega}]$ upper bounds the normalized version. 
Hence, the relation \eqref{eq:localToGlobalFidelityMixing} to fidelity is preserved after the normalization.
As established in the numerical experiment in \fig{plotsBenchmarking}, the normalized shadow overlap very accurately matches the fidelity. 

This close tracking of the fidelity by the shadow overlap can be justified in the weak-noise regime.
In this case, the overall effect of the local stochastic noise may be modeled as a global white-noise given by the depolarizing channel \cite{boixo2018characterizing, dalzell2021random, mark2022benchmarking}.
In other words, the noisy state $\rho$ prepared by the device relates to its ideal version by 
\begin{align}
    \rho = F \cdot \ketbra{\psi}{\psi} + (1-F) \cdot \frac{\iden}{2^n}\nn
    \end{align}
where $F$ (up to an exponentially small correction) matches the fidelity $\bra{\psi}\rho\ket{\psi}$.
Assuming such a noise model, we see that the expected overlap estimated in \protoref{certification} satisfies 
\begin{align}
    \E[\bm{\omega}] = \Tr(L \rho) = F\cdot  \bra{\psi}L\ket{\psi} + (1-F) \cdot \frac{\Tr(L)}{2^n} = \frac{1 + F}{2}
\end{align}
Applying this to level-$1$ shadow overlap in \eqref{eq:NormalizedShadowOverlap}, shows that 
\begin{align}
    \mathrm{normalized\ shadow\ overlap} = \frac{d - 1}{d}\cdot  F + \frac{1}{d}.\nn
\end{align}
In the limit of large number of qubits $n$, this shows that $\mathrm{normalized\ shadow\ overlap}$ and the fidelity $F$ are exponentially close to each other. 

 \section{Quantum phase states}\label{sec:Quantum phase states}
 Quantum phase states are defined as $\ket{\psi} = \sum_{x \in \{0,1\}^n} \frac{1}{\sqrt{2^n}} e^{i \phi(x)} \ket{x}$ for some function~$\phi(x): \{0, 1\}^n \mapsto \mathbb{R}$. 
These states include interesting families of quantum states such as graph states, which are resource states for measurement-based quantum computing, and a subset of Instantaneous Quantum Polynomial-time (IQP) circuits \cite{bremner2011classical}, which are the basis of some proposed quantum supremacy experiments. 
For suitable choices of the function $\phi(x)$, phase states can form highly-entangled states that are indistinguishable from Haar random states. 
This indistinguishability can be information theoretical or computational, as shown in \cite{ji2018pseudorandom, brakerski2019pseudo}, leading to a notion of quantum pseudorandom states with applications in quantum cryptography.
Due to their particular structure, these states may also be amenable to efficient tomography.

A recent work of \cite{arunachalam2022phase} proves among other things that when the complex phases are given by $(-1)^{ f_d(x)}$ with $f_d(x)$ being a degree-$d$ Boolean polynomial over $n$ variables, then this state can be learned using $\Theta(n^{d})$ copies and using only separable measurements.  
Here, we analyze our certification protocol for phase states with an \emph{arbitrary} choice of the function $\phi(x)$ by bounding the sample and query complexity as in \thmref{sampleComplexity}.

The measurement distribution of a phase state $|\braket{x}{\psi}|^2 = \frac{1}{2^n}$ is the uniform distribution over the $n$-dimensional Boolean hypercube $\{0, 1\}^n$.
Therefore, the transition matrix in Equation \eqref{eq:TransitionMatrix} is equivalent to the conventional (lazy) random walk on a degree-$n$ graph. 
In this walk, any vertex $x$ with probability $1/2$ does not change, or otherwise, it transitions to another vertex $y \neq x$ with Hamming distance $1$ from $x$.
The relaxation time of this walk is given by:
\begin{lemma}
 The relaxation time of a lazy random walk on the $n$-dimensional hypercube with uniform edge weights is given by $\tau = n$.
\end{lemma}
\begin{proof}
    The proof is standard and, for example, can be found in \cite[Spectral Methods]{roch2015modern}.
    Alternatively, we can show this from the following expression for the transition matrix $P$ of this walk:
    \begin{align}
        P = \frac{1}{2}\iden + \frac{1}{2n} \sum_{i=1}^n X_i
    \end{align}
    where $X_i$ is the Pauli-$X$ matrix acting on qubit $i \in [n]$.
    We see that the top two eigenvalues of $P$ are $\lambda_0 = \frac{1}{2} + \frac{1}{2n}\cdot n = 1$ and $\lambda_1 = \frac{1}{2} + \frac{1}{2n} (n-2) = 1-\frac{1}{n}$.
    This implies that $\tau = \frac{1}{\lambda_0-\lambda_1}=n$.
\end{proof}
This relaxation time leads to the proof of the result in \secref{mainResult}, which claims an $\mathcal{O}(n^2/\epsilon^2)$ sample complexity (or $\mathcal{O}(n/\epsilon)$ when more general single-qubit measurements are allowed) for certifying arbitrary quantum phase states.

\section{Gapped ground states}\label{sec:gappedGroundStates}
In this section, we analyze the performance of \protoref{verificationGeneral} for a family of gapped ground states, known as stoquastic or sign-free ground states, as well as their ``phase-shifted'' versions.

\begin{definition}[Stoquastic Hamiltonians]
A Hamiltonian $H$ is sign-free (a.k.a stoquastic) in the standard basis $\{ \ket{x}:x \in \{0, 1\}^n\}$ if all the off-diagonal terms of the Hamiltonian are non-positive. That is 
\begin{align}
    \bra{x}H\ket{y}\leq 0\quad  \text{for}\quad  x\neq y.\label{eq:stoquasticEntries}
\end{align}
\end{definition}

\begin{fact}\label{lem:signfreeAmplitudes}
Let $\ket{\psi}$ be the unique ground state of a stoquastic Hamiltonian.
This state is a sign-free quantum state of the form $\ket{\psi} = \sum_{x\in\{0, 1\}^n} \sqrt{p(x)}\ket{x}$ where the support of $p(x)$ is a connected set. 
\end{fact}

Sign-free Hamiltonians are ubiquitous in quantum many-body physics.
They are also particularly significant due to their compatibility with classical simulations using Monte Carlo techniques \cite{BravyiMarginalsSimulate}. 
The condition \eqref{eq:stoquasticEntries} on the Hamiltonian entries helps circumvent the well-known sign problem that emerges when attempting to simulate general Hamiltonians.

\begin{definition}[phase-shifted version of a sign-free state]\label{def:phase-connected states}
    We say that a state $\ket{\psi}$ is the phase-shifted version a given sign-free ground state $\sum_{x\in\{0, 1\}^n} \sqrt{p(x)}\ket{x}$ if $\ket{\psi} = \sum_{x\in\{0, 1\}^n} \sqrt{p(x)} e^{i\phi(x)}\ket{x}$ for some function $\phi(x): \{0, 1\}^n \mapsto \mathbb{R}$.
\end{definition}

Quantum phase states, which we considered in \appref{Quantum phase states}, are a particular example of states in \defref{phase-connected states}.
These states are phase shifted with respect to the state $\ket{+}^{\otimes n}$, which is the ground state of the sign-free Hamiltonian $H = \sum_{i=1}^n \ketbra{-}{-}_i \otimes \iden_{\setminus i}$. 

In what follows, we are interested in $n$-qubit stoquastic Hamiltonians that are gapped or mildly gapless with their energy gap  $\gamma$ lower bounded by  $\g \geq \frac{1}{\poly(n)}$. 
A recent result of \cite{BravyiMarginalsSimulate} shows that the unique ground state of such Hamiltonians admits a fast-mixing random walk based on the Metropolis-Hastings algorithm. 
We extend this result to show that the random walk \eqref{eq:TransitionMatrix} in our framework is also fast-mixing for such states.

Before formally stating the result, we set up some notations following \cite{BravyiMarginalsSimulate}.
Suppose the sign-free Hamiltonian $H$ is $\k$-local (which may or may not be geometrically-local) with spectral (energy) gap $\gamma$. 
We consider the level $m = \k$ of the random walk in \eqref{eq:TransitionMatrix}. 
This means that the transitions are performed between vertices of Hamming distance $\leq \k$ from each other. 
Let $$N = \sum_{i = 1}^{\k} \binom{n}{i} = \mathcal{O}(n^{\k}).$$
Also define the sensitivity parameter 
\begin{align}
    s = \max_{x\neq y}|\bra{x}H\ket{y}|\frac{|\psi(x)|}{|\psi(y)|}.\label{eq:sensitivity}
\end{align}
It is shown in \cite{BravyiMarginalsSimulate} that for sign-free Hamiltonian, the sensitivity parameter is bounded by $s\leq \max_y \bra{y}H\ket{y} - E_0$ where $E_0$ is the ground state energy of the system.

Assuming the ground state is unique, the random walk defined in \eqref{eq:TransitionMatrix} is irreducible. By design, it is also aperiodic and satisfies the detailed-balanced condition. 
The following theorem establishes a $\poly(n)$ bound on the parameter $\tau$ in \protoref{verificationGeneral}. 
\begin{theorem}[cf. \cite{BravyiMarginalsSimulate}]\label{thm:fastMixingStoquastoic}
    The relaxation time $\tau$ of the random walk in \eqref{eq:TransitionMatrix} for the unique ground state of a sign-free Hamiltonian with gap $\gamma$ and sensitivity $s$ is bounded by
    \begin{align}
        \tau \leq \frac{2Ns}{\gamma}.
    \end{align}
\end{theorem}

\begin{proof}
The statement can be obtained by minor changes to the proof of \cite{BravyiMarginalsSimulate}. 
Their proof relies on the fact that $|\psi(x)\psi(y) \bra{x}H\ket{y}| \leq 2Ns \pi(x) P(x,y)$.
This follows from $P(x, y) = \frac{1}{N} \cdot \frac{\pi(y)}{\pi(x) + \pi(y)}$, the definition of parameter $s$ in Equation \eqref{eq:sensitivity}, and the following inequality:
\begin{align}
    \frac{|\psi(x)\psi(y) \bra{x}H\ket{y}|}{\pi(x) P(x, y)} = N \cdot |\bra{x}H\ket{y}| \cdot \left(\frac{|\psi(x)|}{|\psi(y)|} + \frac{|\psi(y)|}{|\psi(x)|} \right) \leq 2Ns.
\end{align}
This concludes the proof.
\end{proof}

\section{GHZ state}\label{sec:GHZ}
The Greenberger-Horne-Zeilinger state $\ket{\mathrm{GHZ}} = \frac{1}{\sqrt2}(\ket{0}^{\otimes n}  + \ket{1}^{\otimes n})$ is a primary example of a multi-partite entangled state with application in quantum communication schemes to quantum error correction. 
Their unique structure also makes them an insightful instance for exploring the certification \protoref{verificationGeneral}.  

Following our previous framework of testing quantum states, 
suppose we are given identical copies of the state $\rho$ and intend to correctly determine with high probability whether~$\rho = \ketbra{\mathrm{GHZ}}{\mathrm{GHZ}}$ or if $\bra{\mathrm{GHZ}}\rho\ket{\mathrm{GHZ}} \leq 1-\epsilon$.

The GHZ state serves as an example of a state that a naive application of the certification \protoref{verificationGeneral} for any constant level $m$ is not capable of certifying it. 
Indeed, a simple inspection reveals that the overlaps $\bm{\omega_1}, \dots, \bm{\omega_T}$ reported by \protoref{verificationGeneral} remain unchanged if the GHZ state is replaced by a state with any relative phase $\ket{\mathrm{GHZ}} = \frac{1}{\sqrt2}(\ket{0}^{\otimes n}  + e^{i\phi} \ket{1}^{\otimes n})$. 

On a technical level, the failure of \protoref{verificationGeneral} is due to the fact that the $Z$-basis measurement distribution of the GHZ state has a disconnected support on vertices $0^n$ and $1^n$. 
This prevents the distribution of the random walk to mix properly to the measurement distribution, a prerequisite for our certification protocol to succeed.
More generally, given the level $ m \geq 1$ of \protoref{verificationGeneral}, consider a state of the form
\begin{align}
    \ket{\psi} = \sum_{k=1}^{\ell} \alpha_k \ket{\varphi_k} \text{\quad with\quad }\ket{\varphi_k} = \sum_{x \in s_k} \sqrt{p_k(x)}e^{i\phi(x)}\ket{x} \label{eq:disconnectedSectors}
\end{align}
for normalized coefficients $\alpha_k \in \mathbb{C}$ and distributions $p_k(x)$ supported on subsets $s_k \subset \{0, 1\}^n$ which have mutual Hamming distance at least $m+1$.
One can see that for any choice of complex coefficients $\alpha_k$, the observable $L$ in Equation \eqref{eq:TransitionMatrix} satisfies $L \ket{\psi} = \ket{\psi}$. 
This means level-$m$ of \protoref{verificationGeneral} measures an overlap $\bm{\omega} = 1$ for potentially orthogonal states $\ket{\psi}$ with varying choice of~$\alpha_k$.
As in the case of GHZ state, this is due to the fact that the measurement distribution $|\braket{x}{\psi}|^2$ has a disconnected support with $\ell$ sectors, breaking the \emph{irreducibility} of the random walk.
While the random walk may mix rapidly to a distribution $p_k(x)$ when starting from $x \in s_k$, this walk does not mix to the global measurement distribution $|\braket{x}{\psi}|^2$.

Assuming a bounded relaxation time $\tau$ for the random walk \eqref{eq:TransitionMatrix}, the certification \protoref{verificationGeneral} can still check each sector $\sum_{x \in s_k} \sqrt{p(x)}e^{i \varphi(x)} \ket{x}$ locally. 
This implies that we can certify that a state $\rho$ has high overlap with one of the states in the set $\{\sum_{k=1}^{\ell} \alpha_k \ket{\varphi_k}: \alpha_k \in \mathbb{C}\}$ in \eqref{eq:disconnectedSectors},  although we cannot specify which of the states. 

Going back to the example of the GHZ state, we now show that a simple change of basis allows us to certify this state with \protoref{verificationGeneral}.
Consider the rotated state $H^{\otimes n} \ket{\mathrm{GHZ}}$ where $H$ is the Hadamard operator. 
This state can be expressed in the standard basis by the uniform superposition over the even-Hamming weight $n$-bit strings:
\begin{align}
    H^{\otimes n} \ket{\mathrm{GHZ}} = \frac{1}{\sqrt{2^{n-1}}
}\sum_{x \in \mathrm{even} } \ket{x}. \label{eq:GHZinXBasis}
\end{align}
Consider the certification \protoref{verificationGeneral} with level $m = 2$. 
Here for each copy of $\rho$ either one or two qubits are chosen randomly and measured in the Pauli $X$, $Y$, or $Z$-basis at random. 
For the ease of analysis, we slightly modify this procedure and skip the case when only one qubit is chosen randomly.
This means at each round, \emph{exactly} two qubits are randomly selected and measured in a randomized Pauli basis. 

\begin{theorem}
    Consider the GHZ state expressed in the $X$-basis as in Equation \eqref{eq:GHZinXBasis} or more generally, its phase-shifted versions of the form $\frac{1}{\sqrt{2^{n-1}}}\sum_{x \in \mathrm{even} } e^{i \phi(x)}\ket{x}$ for arbitrary phases $\phi(x)\in \mathbb{R}$.
    The simplified \protoref{verificationGeneral} at level $m = 2$ (as described above) successfully verifies this state and has relaxation time $\tau = n/2$.
\end{theorem}
\begin{proof}
   The measurement distribution of this state corresponds to the uniform distribution over even-Hamming weight vertices $x$ of the Boolean hypercube $\{0, 1\}^n$.  
   The simplified certification procedure measures the observable $L$ corresponding to the transition matrix of a lazy random walk on these vertices. 
   That is, starting with vertex $x$,  with probability $1/2$ the walk moves to a vertex $y$ that differs in exactly $2$ bits or otherwise with probability $1/2$, it remains at the same vertex $x$.
   
    Although our protocol is only concerned with the walk on the strings with even Hamming weight, these transition rules describe two disjoint walks on the set of even/odd weight strings.
   The relaxation time of this walk can be bounded and is given by $\tau =n/2$ using standard techniques such as those in  \cite{karlin1993markov}.
   Here, we generate a self-contained proof for completeness. 

    We claim that the transition matrix $P$ of the lazy walk on the vertices with distance $2$ can be written as
   \begin{align}
       P = \frac{1}{2} \iden + \frac{1}{2\binom{n}{2}}\sum_{i<j} X_i \otimes X_j \label{eq:Pexpansion}
   \end{align}
   for Pauli-$X$ matrices $X_i$ and $X_j$ acting on distinct qubits $i$ and $j$.
    This can be verified by noticing that $\bra{x} X_i \otimes X_j \ket{y}$ equals  $1$ when $x$ and $y$ differ only in the $i$ and $j$'th bits, and equals $0$ otherwise.
   
  Using \eqref{eq:Pexpansion}, we can find the spectrum of $P$.
  Each bit string $z\in\{0,1\}^n$ characterizes one of the eigenvalues, which is denoted by $\lambda_z$ and given by
   \begin{align}
       \lambda_{z} &=  \frac{1}{2} + \frac{1}{2\binom{n}{2}} \sum_{z':|z'|=2}(-1)^{z\cdot z'}\nn\\
       &=\frac{1}{2} + \frac{1}{2\binom{n}{2}} \left(\binom{n}{2} - 2|z|(n-|z|)\right).\label{eq:Kravchuk}
    \end{align}
   Using this expression, we see that the top two eigenvalues of $P$ are obtained by setting $z = 0^n$ to get $\lambda_z = 1$ and choosing any $z$ with $|z| = 1$ to get $\lambda_{z} = 1 - \frac{2}{n}$.
   An inspection of \eqref{eq:Kravchuk} shows that all the other eigenvalues for the remaining choices of strings $z$ with $|z| \leq \lfloor n/2 \rfloor$ are less than $1 - \frac{2}{n}$.
   The same set of eigenvalues are found for $|z| \geq \lceil n/2 \rceil$, reflecting the degeneracy of the spectrum due to the two disjoint walks on the sets of even and odd Hamming weight vertices.
   This establishes a relaxation time $= n/2$ for this walk. 
\end{proof}

The idea of changing to the $X$-basis can be more broadly applied to certify a state of the form
\begin{align}
   \ket{\psi} = \alpha_0 \ket{0}^{\otimes n} + \alpha_1 \ket{1}^{\otimes n}\nn
\end{align}
for complex coefficients $\alpha_0$ and $\alpha_1$ such that $|\alpha_0|^2 + |\alpha_1|^2 = 1$. 

In the rotated basis, we have 
\begin{align}
    H^{\otimes n}\ket{\psi} = \frac{\alpha_0 + \alpha_1}{\sqrt{2^{n}}
}\sum_{x \in \mathrm{even}} \ket{x} + \frac{\alpha_0 - \alpha_1}{\sqrt{2^{n}}
}\sum_{x \in \mathrm{odd}} \ket{x}.\label{eq:similarGHZ}
\end{align}
Use the path congestion method introduced in \appref{HaarSpectralGap}, one can loosely bound the relaxation time of this distribution with respect to the Markov chain in level $m=1$ of \protoref{verificationGeneral}.
\begin{theorem}
    \protoref{verificationGeneral} at level $m = 1$ verifies the $n$-qubit state specified in \eqref{eq:similarGHZ} with relaxation time $\mathcal{O}\left(\frac{\alpha^2_{\max}}{\alpha_{\min}}\cdot n^2\right)$ where $\alpha_{\max} = \max\{|\alpha_0 - \alpha_1|^2, |\alpha_0 + \alpha_1|^2\}$ and $\alpha_{min} = \min\{|\alpha_0 - \alpha_1|^2, |\alpha_0 + \alpha_1|^2\}$.
\end{theorem}
\begin{proof}
We use the set of paths defined over $\{0, 1\}^n$ in \lemref{propertiesofCanonicalPaths} with length $|\gamma_{xy}|\leq n + 3$. 
The path congestion is maximized over all edges $e = (e^+, e^-)$ with 
\begin{align}
    \rho(\G) &= \max_{e \in E} \left(\frac{n}{\pi(e^+)} + \frac{n}{\pi(e^-)}\right) \sum_{\g_{xy}\ni e} \pi(x)\pi(y) |\gamma_{xy}| \nn\\
    &\leq \frac{2n\alpha^2_{\max}}{\alpha_{\min}} \max_{e\in E} \sum_{\g_{xy}\ni e} \frac{1}{2^{n-1}}\cdot |\gamma_{xy}| \nn\\
    & \leq \mathcal{O}(n^2) \cdot \frac{\alpha^2_{\max}}{\alpha_{\min}}.
\end{align}
The utility of this bound is when $\alpha_{\max}$ and $\alpha_{\min}$ are some constants away from $0$ or $1$. 
Otherwise, where the state \eqref{eq:similarGHZ} reduces to cases previously discussed. 
\end{proof}
   
\section{Learning quantum states}
\label{Learning by hypothesis selection}\label{sec:learningHypothesisselection}
\subsection{Learning by hypothesis selection}\label{sec:Learning by hypothesis selection}
The certification \protoref{verificationGeneral} can also be applied to \emph{learn} a query model $\Psi$ of an $n$-qubit quantum~$\ket{\psi}$.
Suppose we have $T$ identical copies of a state $\rho$ and a set of models $\{\Psi_1, \dots, \Psi_M\}$ representing quantum states $\{\ket{\psi_1},\dots, \ket{\psi_M}\}$ respectively.
Consider the level $m$ of \protoref{verificationGeneral}.
We assume the following two promises that 
\begin{enumerate}
    \item[(1)] There exists at least one $i \in [M]$ such that $\bra{\psi_i}\rho \ket{\psi_i} \geq 1 - \theta$ for some $0\leq \theta < 1$, and
    \item[(2)] The relaxation time of the random walk defined in Equation \eqref{eq:TransitionMatrix} is bounded by $\tau$ for some $\tau \leq \poly(n)$ for the models $\{\Psi_1,\dots,\Psi_M\}$.
\end{enumerate}
Assuming that at least one of the models satisfying Promise (1) exhibits the local escape property defined in \defref{localEscape},we can apply the procedure in \appref{enforcingMixing} to the set of models $\{ \Psi_1, \ldots, \Psi_M \}$. 
This ensures that the new models are fast mixing while Promise (1) continues to be satisfied for them as well.
This relaxes the second condition to:
\begin{enumerate}
    \item[(2')] The relaxation time of the random walk defined in Equation \eqref{eq:TransitionMatrix} is bounded by $\tau$ for at least one of the models that satisfies Promise (1) and the local escape property in \defref{localEscape}. 
\end{enumerate}

The learning algorithm by hypothesis selection is particularly useful when the parameterized model describing the state $\ket{\psi}$ admits a small covering net. 
Given a set $U$, we say a subset $\mathrm{cover}(U)\subseteq U$ forms an $\varepsilon$-covering net when 
\begin{align}
    \sup_{a \in U}\min_{ b \in \mathrm{cover}(U)} \norm{a - b} \leq \varepsilon.\nn
\end{align}
The $\varepsilon$-covering number is defined by the size of the smallest $\varepsilon$-cover:
\begin{align}
    M( U, \varepsilon, \norm{\cdot}) := \min \{|\mathrm{cover}(U)|: \sup_{a \in U}\min_{ b \in \mathrm{cover}(U)} \norm{a - b} \leq \varepsilon \}.\label{eq:coveringNetDef}
\end{align}

In the upcoming sections, we give two examples of relevant models of quantum states with specified covering numbers. 
This allows us to apply the learning algorithm in this section to learn a (coarse-grained) model of the state and by using the certification \protoref{verificationGeneral}.

\begin{theorem}[Sample-efficient learning]\label{thm:sampleEfficientLearning}    
Assuming Promises (1) and (2)  (or alternatively (2')) stated above, there is a learning algorithm that uses $T = 2^{2m-1}\cdot \frac{1}{\epsilon^2}\cdot \log\left(\frac{2M}{\delta}\right)$ many copies of the state $\rho$ and outputs a query model $\Psi$ for an $n$-qubit state $\ket{\psi}$ such that $\bra{\psi}\rho\ket{\psi} \geq 1 - \tau (\theta + 2\epsilon)$ with probability $1 - \delta$.
Similar to \protoref{verificationGeneral}, the qubits in each copy of $\rho$ are measured in the $Z$-basis except for $m$ randomly chosen qubits, which are measured in a random $X$, $Y$, or $Z$-basis.
\end{theorem}

The learning algorithm in \thmref{sampleEfficientLearning} consists of the following steps:
\begin{enumerate}
    \item We start by performing the Pauli measurements outlined in \protoref{verificationGeneral} on each of the $T$ copies of the state $\rho$.
    \item After this data acquisition phase, if promise (2') is used, we apply the procedure in \appref{enforcingMixing} to enforce fast mixing in models $\{\Psi_1,\dots,\Psi_M\}$. 
    \item Steps 4. to 6. of \protoref{verificationGeneral} are performed on models $\{\Psi_1,\dots,\Psi_M\}$.
    \item Having obtained $M$ empirical overlaps $\hat{\bm{\omega}}_1, \dots, \hat{\bm{\omega}}_M$ associated with the $M$ models, we output the model with the largest value among $\{\hat{\bm{\omega}}_i: i\in [M]\}$. Denote this model by $\Psi$. Let $\ket{\psi}$ be the state expressed by $\Psi$ and  $\hat{\bm{\omega}}$ its measured shadow overlap.
    \item  We report $\bra{\psi}\rho\ket{\psi} \geq 1 -  \tau\cdot(1-\hat{\bm{\omega}} + \epsilon)$.
\end{enumerate}
\begin{proof}[Proof of \thmref{sampleEfficientLearning}]
    
    As before, define $\hat{\bm{\omega}}_i$ to be the estimated overlap for each of the $M$ models. 
   Fix the number of copies to be $T = 2^{2m-1}\cdot \frac{1}{\epsilon^2}\cdot \log\left(\frac{2M}{\delta}\right)$.
    It follows from the concentration bound \eqref{eq:concentrationEomega} and a union bound on $M$ applications of \protoref{verificationGeneral} that we have 
    $|\hat{\bm{\omega}}_i -\E(\bm{\omega}_i)|\leq  \epsilon$ with probability at least $ 1 - \delta$ for each model $\Psi_i$.

    We consider three scenarios for the type of models that the learning algorithm outputs.
   \begin{itemize}
       \item[i.] From the first Promise (1) and the fact that $\E[\bm{\omega}_i] \geq \bra{\psi_i}\rho\ket{\psi_i}$ for any model $\ket{\psi_i}$, we have that there exists an $i \in [M]$ such that $\E[\bm{\omega}_i] \geq 1 - \theta$ for that model. 
       The second Promise (2) (or (2')) ensures that this model $\Psi_i$ has a relaxation time upper bounded by $\tau$.
       Overall, the learning algorithm with probability $\geq 1 - \delta$ identifies at least one model $\Psi_i$ with the empirical overlap $\hat{\bm{\omega}}_i \geq 1 - (\theta + \epsilon)$.

    \item[ii.] Now we move on to the other models examined by the learning algorithm where $\E[\omega] < 1- \theta$.
    Each of these models is either fast mixing  (by satisfying the local escape property) or not. 
    In the former case, the augmentation by the procedure given in \appref{enforcingMixing} does not alter the model.
    The empirical overlap for these models with probability $\geq 1-\delta$ satisfies $\hat{\bm{\omega}} \leq 1-(\theta-\epsilon)$. 

    \item[iii.] Assuming Promise (2') is used, the models that do not exhibit the local escape property are modified through the procedure in \appref{enforcingMixing}.
    These models potentially become very different from the original model when queried.
    In particular, it could be that the expected overlap satisfies $\E[\bm{\omega}] \geq 1 - \theta$ after enforcing the fast mixing condition.
    \end{itemize}
    The learning algorithms obtain the model $\Psi$ that has the largest estimated overlap $\hat{\bm{\omega}}$ among the three cases i., ii., and iii. presented above.
    We see by inspecting these cases that with probability $\geq 1- \delta$, the largest estimated overlap is lower bounded by $\hat{\bm{\omega}} \geq 1 - (\theta + \epsilon)$, which considering the statistical error, implies that $\E[\bm{\omega}]\geq 1 - (\theta + 2\epsilon)$. 
    
    The model that attains this maximum could be from any of the cases i., ii., or iii., but since in each case we have (enforced) fast mixing, following the step 5. specified above, we can claim that the learning algorithm outputs a model with a certified fidelity $\bra{\psi}\rho\ket{\psi}\geq 1 - \tau\cdot(\theta + 2\epsilon)$ with probability $\geq 1 - \delta$.
\end{proof}

\subsection{Learning neural quantum states}\label{sec:learningNeuralStates}
Neural quantum states are a family of parameterized quantum states 
$$\ket{\psi_w} = \sum_{x \in \{0,1\}^n} \psi_w(x) \ket{x},$$
where the complex amplitudes $\psi_w(x)$ are given by the output of 
a neural network whose parameters are denoted by $w$.

The choice of the neural network architecture and how complex numbers $\psi_w(x)$ are represented by the network result in different ansatz of quantum states.
For instance, past works have studied the expressiveness of neural quantum stated based on restricted Boltzmann machines \cite{carleo2017solving}, convolutions neural networks \cite{Carleo2019CNN}, and recurrent neural networks \cite{Hibat2020RNN}. 
In some cases, the neural network includes complex weights and directly outputs the amplitudes. 
In other cases, networks with real weights $w$ are used for the real $\mathrm{Re}(\psi(x))$ and imaginary $\mathrm{Im}(\psi(x))$ parts,
or the magnitude $|\psi(x)|$ and phase $\phi(x)$ parts of $\psi(x) = |\psi(x)| e^{i \phi(x)}$.
In this section, for a clearer exposition, we focus on (real) feedforward neural networks which represents the magnitude and phase components of the amplitudes. Arguments similar to the one used in this section apply to other cases as well.

More formally, a feedforward neural network, consists of various layers each including multiple ``neurons''. The $j$'th neuron in the $i$'th layer is defined to be a function that maps the input $x\in \mathbb{R}^{n_{i-1}}$~to
$$x\mapsto\s(\braket{w_{ij}}{x}+b_{ij}).$$
This map consists of an affine part $\braket{w_{ij}}{x}+b_{ij}$ with some $\ket{w_{ij}} \in \mathbb{R}^{n_{i-1}}$ and $b_{ij}\in \mathbb{R}$,
and a non-linear part involving the \emph{activation function} $\s$ which is often chosen to be rectified linear unit (ReLU) $z\mapsto \max\{0,z\}$ or sigmoid $z\mapsto \frac{1}{1+e^{-z}}$; see \cite{GoodBengCour16} for an introduction to deep learning.

The feedforward neural net, overall, corresponds to the mapping 
\begin{align}
    f(x;w):=\s_L\left(w_L\s_{L-1}\left(\cdots w_2\s_1(w_1 x+b_1)+b_2\cdots \right)+b_L\right),\label{feedforwardNN}
\end{align}
where $w_i:=\sum_{j\in [n_i]} \ketbra{j}{w_{ij}}\in\mathbb{R}^{n_i \times n_{i-1}}$ and $b_i:=\sum_{j\in [n_i]} b_{ij}\cdot\ket{j}\in \mathbb{R}^{n_i}$. Each function $\s_i$ is now a coordination-wise activation function. 
We collectively denote these parameters by $$w=(w_1,b_1,\dots,w_L,b_L).$$
The number of layers $L$ is called the \emph{depth} and the maximum number of neurons across layers $W = \max_{i\in [L]} n_i$ is known as the \emph{width} of the neural network.

Given two feedforward neural networks $f(x; w_{\text{abs}})$ and $g(x; w_{\text{phase}})$, we define a neural quantum state $\sum_{x \in \{0,1\}^n} \psi_w(x) \ket{x}$ with amplitudes $\psi_w(x) = |\psi(x)| e^{i \phi(x)}$ and weights $w = (w_{\text{abs}},w_{\text{phase}})$ by assigning 
\begin{align}
    \log |\psi_w(x)| := f(x;w_{\text{mag}}) \text{\quad and \quad } \phi(x)  := g(x;w_{\text{phase}}).\label{eq:nerualQuantumStates}
\end{align}

Going back to the learning framework in \appref{learningHypothesisselection}, we assume that a quantum state $\rho$ admits an accurate neural net representation $\psi_w(x)$ for some weight parameters $w$.
Our goal is find the parameter $w$ given independent copies of the state $\rho$. 
In practice, this is often achieved by performing stochastic gradient descent or similar optimization subroutines,
where the objective function is e.g., the fidelity, log-likelihood, or KL divergence. 
We now show that the sample complexity of the learning algorithm in \appref{learningHypothesisselection} can be rigorously analyzed.  

This is achieved by defining a covering net over the set of feedforward neural networks such that any neural quantum state $\psi_w(x)$ is close to a point in the covering net.
The following lemma translates the approximation error in coarse-graining the set of neural nets to the corresponding error in the neural quantum state.

\begin{lemma}[Approximation error for pure states, cf. \cite{NNrepresentationTNSharir}]\label{lem:approximationNeuralStates}
    Given two quantum states
    \begin{align}
        \ket{\psi(x)} = \sum_{x\in \{0,1\}^n}|\psi(x)|e^{i \phi(x)} \cdot \ket{x} \text{\quad and \quad} \ket{\psi'(x)} = \sum_{x\in \{0,1\}^n}|\psi'(x)| e^{i \phi'(x)} \cdot \ket{x},
    \end{align}
    suppose that for any $x \in \{0,1\}^n$, it holds that $|\log|\psi(x)| - \log|\psi'(x)||\leq \varepsilon/2$ and $|\phi(x) - \phi(x')|\leq \varepsilon/2$.
    Then, we have 
    \begin{align}
        \frac{|\braket{\psi'}{\psi}|^2}{|\braket{\psi}{\psi}| \cdot |\braket{\psi'}{\psi'}|} \geq 1 - \varepsilon^2.
    \end{align}
 \end{lemma}
 
The covering number of the set of feedforward neural networks has been found before in \cite{bartlett2017spectrally} and depends on the spectral and $(2,1)$-norm of the weight matrices in the network. 
The covering number in general scales exponentially with the depth of the network $L$. 
Given a matrix $A$ with entries $A_{ij}$, its $(2,1)$-norm is defined by the $\sum_{j}(\sum_i |A_{ij}|^2)^{\frac{1}{2}}$

\begin{proposition}[Covering number for neural quantum states, cf. \cite{bartlett2017spectrally}]\label{prop:coveringNumberNN}
     Consider neural quantum states represented, in the sense of Equation \eqref{eq:nerualQuantumStates} by the set of feedforward neural networks with weight matrices that have bounded spectral norm $\norm{\cdot}_2$ and $(2,1)$-norm $\norm{\cdot}_{2,1}$ given by
    \begin{align}
   \mathcal{F} = \{ \s_L\left(w_L\s_{L-1}\left(\cdots w_2\s_1(w_1 x)\cdots \right)\right): \norm{w^\intercal_i}_2\leq s_i, \norm{w_i^\intercal}_{2,1} \leq r_i\},\label{feedforwardNN-set}
\end{align}

We assume a Boolean input data $x \in \{0,1\}^n$ and networks with depth $L$, $n_i$ neurons in layer $i\in[L]$, and width $W =\max_{i \in [L]} n_i$.
The activation functions $(\sigma_i)_{i=1}^L$ have Lipschitz constants $\rho_1,\dots,\rho_L$ and satisfy $\sigma_i(0) = 0$.
The covering number (defined in \eqref{eq:coveringNetDef}) of this set is bounded by
\begin{align}
    \log \left(M(\mathcal{F}, \varepsilon, \norm{\cdot}_F)\right) \leq \frac{n \log(2 W^2)}{\varepsilon^2} \cdot \prod_{j=1}^L \rho^2_j s^2_j \cdot \left(\sum_{i=1}^L \left(\frac{r_i}{s_i}\right)^{2/3}\right)^3.
\end{align}
\end{proposition}

Roughly speaking, if the spectral norm of all weight matrices is bounded by $s$, this bound on the covering number of $\mathcal{F}$ scales with $\log( M(\mathcal{F}, \varepsilon, \norm{\cdot}_F))\leq \widetilde{O}\left(\frac{n L^3 W^3 s^{2L}}{\varepsilon^2}\right)$.
The exponential scaling with respect to the depth $L$ of the network is in contrast to the case of quantum circuits where the size of the covering net scales linearly with the size of the circuit. 

Our learning algorithm in \thmref{sampleEfficientLearning} relies on Promise (1) stated in \appref{Learning by hypothesis selection}. 
That is, we assume that one of the models obtained via casting the covering net in \propref{coveringNumberNN} has high overlap with the lab state~$\rho$. 
If this condition is satisfied for the set of neural quantum states prior to applying the covering net, then we claim that it also holds for the states in the covering net. 
To see this, consider the neural quantum state $\ket{\psi_{w^*}}$ such that $\bra{\psi_{w^*}}\rho \ket{\psi_{w^*}} \geq 1-\theta$. 
Let $\ket{\psi_{w_0}}$ denote the closest state in the covering net to $\ket{\psi_{w^*}}$.
We know from the approximation guarantee in \lemref{approximationNeuralStates}
that $|\braket{\psi_{w^*}}{\psi_{w_0}}|^2 \geq 1-\varepsilon^2$.
The following known relation between the trace distance and fidelity $$1-\bra{\psi}{\rho}\ket{\psi} \leq D_{\mathrm{tr}}(\rho, \ketbra{\psi}{\psi}) \leq \sqrt{1- \bra{\psi}\rho\ket{\psi}}$$ 
implies that $\bra{\psi_{w_0}}\rho \ket{\psi_{w_0}}\geq 1 - \varepsilon -\sqrt{\theta}$.

The learning algorithm in \thmref{sampleEfficientLearning} also requires that Promise (2) or (2') hold for the neural quantum states in the covering net $\{\Psi_1,\dots,\Psi_M\}$. 
While this can be taken as an assumption regarding this coarse-grained set of models, we may also derive Promise (2') assuming that it holds before casting the covering net. 
To this end, suppose that the state $\ket{\psi_{w^*}}$, as considered above, satisfies the local escape property with parameters $(\alpha, c'_u, c_u, c_{\ell})$ and therefore has a bounded relaxation time. 
We know from the definition of the covering net for neural quantum states that the amplitudes $|\psi_{w_0}(x)|$ of the closest sate in the covering net are within $\mathcal{O}(\varepsilon)$ multiplicative error of $|\psi_{w^*}(x)|$.
Hence, the transition probabilities in \eqref{eq:TransitionMatrix} are also changed up to an $(\mathcal{O}(\varepsilon))$ multiplicative error.
One can also directly verify that for a sufficiently small $\varepsilon$, the local escape property in \defref{localEscape} is satisfied up an $\mathcal{O(\varepsilon)}$ multiplicative error of the original constants $(c'_u, c_u, c_{\ell})$.
This means that if the relaxation time of the state $\ket{\psi_{w^*}}$ is bounded by $\tau \leq \mathcal{O}(n^2)$, the same holds for the state $\ket{\psi_{w_0}}$ in the covering net.

The following lemma summerizes these properties of the set of models $\{\Psi_1,\dots,\Psi_M\}$ obtained using \propref{coveringNumberNN}:

\begin{lemma}[Properties of  coarse-grained neural nets]\label{lem:promises for covering net}
Consider a lab state $\rho$ and a set of neural quantum states that satisfy Promises (1) and (2') in \appref{Learning by hypothesis selection}.
These neural states are represented with feedforward neural networks of depth $L$, width $W$, and weight matrices with bounded spectral norm $s$. 
It holds that the models constructed via the covering net in \propref{coveringNumberNN} for a sufficiently small constant $\varepsilon$ also satisfy Promises (1) and (2') with at least one state $\ket{\psi_{w_0}}$ in the covering net exhibiting $\bra{\psi_{w_0}}\rho \ket{\psi_{w_0}} \geq 1- \varepsilon - \sqrt{\theta}$ with a relaxation time $\tau\leq\poly(n)$.
\end{lemma}

An immediate corollary of \lemref{promises for covering net} and \thmref{sampleEfficientLearning} is the following sample complexity for learning neural networks.

\begin{corollary}[Sample complexity of learning neural quantum states] \label{cor:SampleComplexityNN}
Assume Promises (1) and  (2')  in \thmref{sampleEfficientLearning} hold for a family of neural quantum states represented with the set of feedforward neural networks of depth $L$, width $W$, and weight matrices with bounded spectral norm $s$. 
Then, for any sufficiently small constant $\varepsilon$, there is a learning algorithm that uses 
$$T = 2^{2m}\cdot \frac{1}{\epsilon^2}\cdot \left(\widetilde{O}\left(\frac{n L^3 W^3 s^{2L}}{\varepsilon^2}\right) + \log\left(\frac{2}{\delta}\right)\right)$$
many identical copies of the state $\rho$ and outputs a feedforward neural net obtained from the covering net in \propref{coveringNumberNN} with weights $w$ and a relaxation time $\tau \leq \poly(n)$  such that $\bra{\psi_w}\rho\ket{\psi_w} \geq 1 - \tau ( \sqrt{\theta} + \varepsilon + 2\epsilon)$.
\end{corollary}

\subsection{Learning gapped Hamiltonians}\label{sec:LearningNoisyGS}
An illustrative application of the setup of \appref{Learning by hypothesis selection} is to learn a gapped local Hamiltonian $H$ whose ground state has high overlap with some lab state $\rho$.
More concretely, suppose we have a family of $\k$-local Hamiltonians $$H(a) = \sum_{j=1}^{E} h_j(a)$$ consisting of $E$ local terms which are specified by a set of parameters $a \in [-1, 1]^L$ for an integer $L\geq 1$. 
Moreover, the local terms satisfy $\norm{h_j}_{\infty} \leq 1 $ and $\norm{\frac{\partial h_j}{\partial a_k}(a)}_{\infty}\leq c$ for $j\in [E]$, $k \in [L]$, and $a \in [-1, 1]^L$.
These Hamiltonians are gapped and have a unique ground state for any choice of parameters~$a \in [-1, 1]^L$. 
We define $\gamma$ to be a lower bound on the gap of $H(a)$ for ~$a \in [-1, 1]^L$. 

We assume that for one or more instantiations of parameters $a \in [-1,1]^L$, the corresponding ground states have high overlap with the lab state $\rho$.
Our goal is to approximately infer one of these systems by learning their corresponding parameters.

We start with \propref{numberModelsGappedHamiltonians} which shows how to obtain a set of models  $\{\Psi_1, \dots, \Psi_M\}$ for which Promise (1) in \thmref{sampleEfficientLearning} is satisfied.
This entails setting up a discretization of the set of parameters $a \in [-1,1]^L$ by casting a covering net.

Consider a covering net on parameters $a$ given by 
\begin{align}
    \operatorname{Cover}(a): =\{-1 + j \frac{2}{N_c}: j \in \{0,\dots,N_c\}\}^{\times L}.\label{eq:covering set ground states}
\end{align}
It holds that for any parameter $b \in [-1,1]^L$, we can find a point $b' \in \operatorname{Cover}(a)$ in the covering net such that $|b_k - b'_k| \leq \frac{1}{N_c}$ for~$k \in [L]$.

The covering net $\operatorname{Cover}(a)$ naturally induces a covering net on the set of ground state $\ket{\psi(a)}$ of Hamiltonian $H(a)$.
We choose the set of models $\{\Psi_1,\dots,\Psi_M\}$ to be the models representing these coarse-grained ground states $\{\ket{\psi(a)}: a\in \operatorname{Cover}(a)\}$.
The following proposition establishes a general bound on the number of models, $M = N_c^L$, needed for guaranteeing Promise (1) in \thmref{sampleEfficientLearning}. 
Here we assume that level $m = \k$ \protoref{verificationGeneral} is performed for the $\k$-local Hamiltonian~$H(a)$.
 
\begin{proposition}[Covering number of gapped ground states]\label{prop:numberModelsGappedHamiltonians}
Consider the family of gapped Hamiltonians~$H(a) = \sum_{j=1}^{E} h_j(a)$ specified above.
Let the size of the covering net \eqref{eq:covering set ground states} be $N_c = \mathcal{O}\left(\frac{L E}{\gamma \varepsilon}\right)$. 
This induces a covering net over the set of ground states $\ket{\psi(a)}$ of these Hamiltonians such that for any $b \in [-1,1]^L$, there is a parameter $b\in \operatorname{Cover}(a)$ for which $\mathrm{D_{tr}}\left(\ketbra{\psi(b)}{\psi(b)}, \ketbra{\psi(b')}{\psi(b')} \right) \leq \varepsilon$.
\end{proposition}

\begin{proof}
Denote the ground state projector by $\psi(a) := \ketbra{\psi(a)}{\psi(a)}$.
We show that for any $b \in [-1,1]^L$, there exists a parameter $b' \in \operatorname{Cover}(a)$ such that $\mathrm{D_{tr}}(\psi(b), \psi(b'))$ is sufficiently small.
The choice of parameter $b'$ is simply the closest point in $\operatorname{Cover}(a)$ to $b$.
The proof relies on the notion of quasi-adiabatic evolution which gives us the right framework for bounding the change in the ground state due to small variations in the Hamiltonian parameters. 

Define a path $b(s):= s b + (1-s) b'$ for $s\in[0,1]$ from parameter $b$ to $b'$.
It is shown in \cite{Hastings2005Quasiadiabatic1, bachmann2012automorphic} that for any point $s \in [0, 1]$, the ground state projector $\psi(s) = \ketbra{\psi(s)}{\psi(s)}$ satisfies $\frac{\partial  \psi}{\partial s}(s) = i[\operatorname{D}_s(s), \psi(s)]$, where the operators $\operatorname{D}_s(s)$ is defined by
    \begin{align}
        \operatorname{D}_s (s):= \int_{-\infty}^{\infty} W_\gamma(t) e^{i t H(s)} \frac{\partial H}{\partial s}(s) e^{-i t H(s)} dt.
    \end{align}
Here $|W_{\gamma}(t)|$ is continuous and monotone decreasing for $t\geq 0$ with $\sup_{t}|W_{\gamma}(t)| = W_{\gamma}(0) = 1/2$.
It further holds that $\int_{-\infty}^{\infty} |W_{\gamma}(t)| d t \leq K/\gamma$ for a known constant $K$ and the spectral gap $\gamma > 0$.
From this, we have that     
    \begin{align}
        \norm{\psi(1) - \psi(0)}_1 &= \NORM{\int_{0}^1 \frac{d \psi}{d s}(s) \ d s}_1\nn\\
        &= \NORM{\int_{0}^1 [\operatorname{D}_s(s), \psi(s)] \ d s}_1\nn\\
        &\leq 2\max_{s\in [0,1]} \norm{\operatorname{D}_s(s)}_{\infty}.
    \end{align}
For any $s \in [0, 1]$, it holds that 
    \begin{align}
        \norm{\operatorname{D}_s (s)}_{\infty} &\leq \int_{-\infty}^{\infty} |W_\gamma(t)| \cdot \NORM{e^{i t H(s)} \frac{\partial H}{\partial s}(s) \ e^{-i t H(s)}}_{\infty} dt \nn\\
        &\leq  \int_{-\infty}^{\infty} |W_\gamma(t)|\cdot \NORM{  \frac{\partial H}{\partial s}(s) }_{\infty} dt \nn\\
        &\leq  \frac{K}{\gamma}\cdot \max_{s\in[0,1]} \NORM{  \frac{\partial H}{\partial s}(s) }_{\infty} && \text{from\ } \int_{-\infty}^{\infty} |W_\gamma(t)| d t \leq K/\gamma\nn\\
        &\leq  \frac{K}{\gamma}\cdot \sum_{j=1}^{E}\max_{s\in[0,1]}\NORM{\frac{\partial h_j}{\partial s}(s) }_{\infty} \nn\\
        &\leq  \frac{K}{\gamma}\cdot \sum_{j=1}^{E} \sum_{k=1}^L |b_k - b'_k|\cdot \max_{s\in[0,1]}\NORM{\frac{\partial h_j}{\partial a_k}(b(s))}_{\infty} && \text{assumption \ } \NORM{\frac{\partial h_j}{\partial a_k}}\leq c\nn\\
        &\leq cK L\frac{E}{\gamma N_c}.\label{eq:boundCoveringGroundState}
    \end{align}
By defining a constant $c' = cK$ and choosing $N_c = \frac{2c' L E}{\gamma \varepsilon}$ in \eqref{eq:boundCoveringGroundState}, we see that there exists a parameter $b'$ in the covering net such that $\mathrm{D_{tr}}\left(\psi(b), \psi(b') \right) \leq \varepsilon$. 
\end{proof}

Suppose that there exists a set of parameters $b \in [-1, 1]^L$ such that the ground state of $H(b)$, denoted by $\ket{\psi(b)}$, satisfies $\bra{\psi(b)}\rho\ket{\psi(b)}\geq 1- \theta$.
Using \propref{numberModelsGappedHamiltonians} we have that when $N_c = \mathcal{O}\left(\frac{L E}{\gamma \varepsilon}\right)$, there exists $b' \in \operatorname{Cover}(a)$ such that $\bra{\psi(b')}\rho \ket{\psi(b')} \geq 1- \varepsilon-\sqrt{\theta}$.
This can be applied to fulfill Promise~(1) of \thmref{sampleEfficientLearning}.
Moreover, we assume that such Hamiltonians satisfy Promise (2) in \appref{Learning by hypothesis selection}.
Hence, the transition matrices \eqref{eq:TransitionMatrix} associated with the ground states of these Hamiltonians exhibit a relaxation time $\tau \leq \poly(n)$.
This, for instance, holds in the setting of \appref{gappedGroundStates} when Hamiltonians $H(a)$ are further assumed to be stoquastic.

Under these conditions, we can apply \propref{numberModelsGappedHamiltonians} to obtain a discretized set of Hamiltonians such that their ground states satisfy Promises (1) and (2) needed in the learning algorithm in \thmref{sampleEfficientLearning}.
In summary, we get the following corollary of \propref{numberModelsGappedHamiltonians} and \thmref{sampleEfficientLearning}.

\begin{corollary}[Sample complexity of Hamiltonian learning]\label{cor:sampleComplexityGS}  
Consider a lab state $\rho$ and a family of $\k$-local Hamiltonians $H(a) = \sum_{j=1}^{E} h_j(a)$ for $a\in[-1,1]^{L}$ with a spectral gap $\geq \gamma$.
Suppose the ground states of these Hamiltonians fulfill Promises (1) and (2) in \appref{Learning by hypothesis selection} with a relaxation time $\tau \leq \poly(n)$ and at least one ground state $\ket{\psi(b)}$ such that $\bra{\psi(b)}\rho\ket{\psi(b)} \geq 1- \theta$.
Then, there is a learning algorithm that uses 
$$T = 2^{2\k}\cdot \frac{1}{\epsilon^2}\cdot \left(L \log\left(\frac{c' L E }{\gamma \varepsilon}\right) + \log\left(\frac{2}{\delta}\right)\right)$$
copies of the lab state $\rho$ and outputs a Hamiltonian parameter $b'\in \operatorname{Cover}(a)$ as in \propref{numberModelsGappedHamiltonians}, such that its ground state $\ket{\psi(b)}$ satisfies $\bra{\psi(b)}\rho\ket{\psi(b)} \geq 1 - \tau (\sqrt{\theta} + \varepsilon + 2\epsilon)$.
\end{corollary}

\section{Estimating sparse observables}\label{sec:Estimatingsparseobservables}
We have thus far developed the tools needed to certify a learned model $\Psi$ of the quantum state $\ket{\psi}$ which admits efficient query and sample access to amplitudes $\psi(x)$.
Having such a model, one can use Monte Carlo methods to efficiently estimate the expectation of any sparse observable such as the energy of a local Hamiltonian or low-degree polynomials of the reduced density operators such as the R\'enyi entanglement entropy.

We first show this for an observable $O$ which is a $g$-sparse observable on $n$-qubits in the standard basis.
This means for any $x\in \{0, 1\}^n$, there are at most $g$ states $y \in \{0, 1\}^n$ such that $\bra{x}O\ket{y} \neq 0$.
For instance, a $\k$-local Hamiltonian $H = \sum_{i=1}^M H_i$ involving $M$ local terms is $\mathcal{O}(2^{\k}M)$-sparse.

The expectation value of an observable $O$ defined by $\bra{\psi}O\ket{\psi}$ can be expressed as
\begin{align}
    \bra{\psi}O\ket{\psi}&=\sum_{x, y \in \{0,1\}^n} \braket{\psi}{x} \cdot \bra{x}O\ket{y}\cdot \braket{y}{\psi}\nn\\
    &=\sum_{x,y \in \{0,1\}^n} |\braket{x}{\psi}|^2 \cdot \bra{x}O\ket{y}\cdot \frac{\braket{y}{\psi}}{\braket{x}{\psi}}\nn\\
     &=\E_{\bm{x} \sim |\braket{x}{\psi}|^2}  \sum_{y: \bra{\bm{x}}O\ket{y} \neq 0}  \bra{\bm{x}}O\ket{y}\cdot \frac{\braket{y}{\psi}}{\braket{\bm{x}}{\psi}}.\label{eq:monteCarloSparseO}
\end{align}
Given a sampled bit string $\bm{x}$, the expression inside the expectation can be efficiently calculated when $g \leq \poly(n)$ and we have access to the (un-normalized) amplitudes via a query model $\Psi$.
By sampling $\bm{x} \sim |\braket{x}{\psi}|^2$, we can estimate the expression in \eqref{eq:monteCarloSparseO} using empirical averaging.
To determine the number of samples needed to obtain an accurate estimate of $\bra{\psi}O\ket{\psi}$, we bound the variance of  expression \eqref{eq:monteCarloSparseO}.
The random term in \eqref{eq:monteCarloSparseO} equals $ \E_{\bm{x}\sim |\braket{x}{\psi}|^2}\frac{\bra{\bm{x}}O\ket{\psi}}{\braket{\bm{x}}{\psi}}$.  
We first consider the real part of this term, with the imaginary part following a similar argument.
We have
\begin{align}
    \Var_{\bm{x}\sim |\braket{x}{\psi}|^2} \mathrm{Re}\left(\frac{\bra{\bm{x}}O\ket{\psi}}{\braket{\bm{x}}{\psi}}\right) &\leq   \E_{\bm{x}\sim |\braket{x}{\psi}|^2}\mathrm{Re}\left(\frac{\bra{\bm{x}}O\ket{\psi}}{\braket{\bm{x}}{\psi}}\right)^2 \nn\\
    &\leq  \E_{\bm{x}\sim |\braket{x}{\psi}|^2}\left| \frac{\bra{\bm{x}}O\ket{\psi}}{\braket{\bm{x}}{\psi}}\right|^2\nn\\
    &= \E_{\bm{x}\sim |\braket{\bm{x}}{\psi}|^2} \frac{\bra{\psi}O\ketbra{\bm{x}}{\bm{x}}O\ket{\psi}}{|\braket{\psi}{\bm{x}}|^2} \nn\\
    &= \sum_{x\in\{0,1\}^n}\bra{\psi}O\ketbra{x}{x}O\ket{\psi}\nn\\
    &= \bra{\psi}O^2\ket{\psi} \label{eq:varianceMC}
\end{align}

For practically relevant cases such as a $\k$-local Hamiltonians with bounded local terms, we have $\bra{\psi}O^2\ket{\psi}\leq  \norm{O^2}\leq \poly(n)$.
It follows from the  Chebyshev inequality that the number of samples required to estimate $\bra{\psi}O\ket{\psi}$ up to $\epsilon$ additive error with  probability $\geq 1-\delta$ is upper bounded by~$\frac{8\bra{\psi}O^2\ket{\psi}}{\delta \cdot \epsilon^2}$.

As considered in \cite{huang2020predicting}, one could further boost this performance by using the median-of-the-means (MoM) estimator. 
This estimator is robust to outliers and results in an exponential improvement with respect to the error probability $\delta$.
This entails collecting $K$ batches of size $B$ of samples denoted by $\bm{x_1}, \dots, \bm{x_{BK}}$. 
We then find the empirical averages $\bm{o_k}$ of each batch $k\in [K]$ and compute their median $\bm{o_{MoM}}$. 
More formally---again focusing on the real part---we have
\begin{align}
\bm{o_{\mathrm{MoM}}} = \mathbf{Median}\left\{ \bm{o_1},\dots, \bm{o_K} \right\} \quad \textrm{where} \quad  \bm{o_k} = \frac{1}{B}\sum_{i = B(k-1)+1}^{Bk} \mathrm{Re}\left( \frac{\bra{\bm{x_i}}O\ket{\psi}}{\braket{\bm{x_i}}{\psi}}\right)
\label{eq:medianMeans}
\end{align}
Let $B = \frac{34\bra{\psi}O^2\ket{\psi}}{\epsilon^2}$ and $K = 2\log \left( \frac{2}{\delta}\right)$.
Then, for all $\epsilon >0$, we have 
$$\mathrm{Pr} \left[ \left| \bm{o_{\mathrm{MoM}}} -\operatorname{Re}\left(\bra{\psi}O\ket{\psi}\right) \right| \geq \epsilon \right] \leq \delta .$$

These findings are summarized in the following theorem:

\begin{theorem}[Estimating sparse observables]
    For any sparse observable $O$ and state $\ket{\psi}$,
    one can estimate $\bra{\psi}O\ket{\psi}$ up to an error $\epsilon$ with probability $\geq 1-\delta$ 
    using the median-of-the-means estimator and 
    \begin{align}
        T = 272\cdot \frac{\bra{\psi}O^2\ket{\psi}}{\epsilon^2}\cdot \log \left( \frac{4}{\delta}\right)
     \end{align}
    samples $\bm{x} \sim |\braket{x}{\psi}|^2$, as well as, at most $gT$ queries to the entries $O_{xy}$ of observable $O$ and $(g+1)T$ queries to the amplitudes $\psi(x)$ of the state $\ket{\psi}$. 
\end{theorem}

A similar argument can be applied to estimate other non-linear functions such as R\'enyi entanglement entropies. 
In this case, the expected value is given by $\Tr\left( O \ketbra{\psi}{\psi}^{\otimes \ell}\right)$ for some $\ell \geq 2$.
Here for simplicity, we only consider quadratic functions $\ell = 2$. 
We have
\begin{align}
    \Tr\left( O \ketbra{\psi}{\psi}^{\otimes 2}\right) &= \E_{\substack{\bm{x}\sim |\psi(x)|^2 \\\bm{x'} \sim |\psi(x')|^2}}\  \sum_{y,y'} \bra{\bm{x},\bm{x'}} O \ket{y, y'} \cdot \frac{\braket{y}{\psi}}{\braket{\bm{x}}{\psi}}\cdot \frac{\braket{y'}{\psi}}{\braket{\bm{x'}}{\psi}}\label{eq:quadraticObservable}
\end{align}

As before, given samples $\bm{x}$ and $\bm{x'}$, the expectation can be efficiently computed if the observable $O$ is sparse. 

To empirically estimate this quantity, we use the median-of-the-means estimator with $K$ batches of size $B$ as before. 
As an application, consider estimating the \emph{purity} of the reduced density operator $\rho_A = \Tr_B \left(\ketbra{\psi}{\psi}_{AB}\right)$ given by $\Tr\left(\rho^2_A\right)$.
In this case, the observable $$O = \mathrm{SWAP}_{AA'} \otimes \iden_{BB'} \text{\quad with \quad} \mathrm{SWAP}_{AA'} \ \ket{x}_A \ket{x'}_{A'} = \ket{x'}_A \ket{x}_{A'}.$$
We indeed have 
\begin{align}
    \Tr\left( \rho_A^2 \right) &= \Tr \left( \mathrm{SWAP}_{AA'} \cdot \rho_A \otimes \rho_{A'}\right)\nn\\
    &= \bra{\psi}_{AB} \bra{\psi}_{A'B'} \cdot  \mathrm{SWAP}_{AA'} \otimes \iden_{BB'} \cdot \ket{\psi}_{AA'} \ket{\psi}_{BB'}.\nn
\end{align}
Plugging this in Equation \eqref{eq:quadraticObservable} gives us
\begin{align}
     \Tr\left( \rho_A^2 \right) &= \E_{\substack{(\bm{x}_{A},\bm{x}_B)\sim |\psi(x)|^2 \\ (\bm{x'}_{A},\bm{x'}_B) \sim |\psi(x')|^2}}\left( \frac{\braket{\bm{x}'_A \bm{x}_B}{\psi}}{\braket{\bm{x}_A\bm{x}_B}{\psi}}\cdot \frac{\braket{\bm{x}_A \bm{x'}_B}{\psi}}{\braket{\bm{x'}_A \bm{x'}_B}{\psi}}\right)\label{eq:PurityExpression}
\end{align}

A direct calculation similar to Equation \eqref{eq:varianceMC} shows that the variance of this case is bounded by $\bra{\psi}^{\otimes 2} O^2 \ket{\psi}^{\otimes 2} = 1$. 
Hence, the sample complexity of estimating this quantity scales as $\mathcal{O}\left(\frac{1}{\epsilon^2}\cdot \log(\frac{1}{\delta})\right)$.
Following \cite{huang2020predicting}, we may also improve this variance by computing the empirical average in each batch a set of samples $\bm{x_1},\dots, \bm{x_B}$ using the concept of U-statistics.

\section{Details of numerical experiments}

The C++ and Python code for reproducing the numerical experiments are publicly available on Google Drive at \url{https://bit.ly/3U93gvl}.

\subsection{Training and certifying neural quantum states for ML tomography}\label{sec:NumericalTraining}

Here, we review in more detail the numerical experiment in \fig{NQSTraining} where state tomography is performed on a random phase state on $n=120$ qubits by training and certifying a neural quantum state.
The phase state is given by $\ket{\psi} = \frac{1}{\sqrt{2^n}}\sum_{b\in\{0,1\}^n} e^{i \phi(b)} \ket{b}$.
We consider two choices of complex phases $\phi(b)$:\code{type == 0}(Pseudorandom) is when the phases $\phi(b)$ are generated by a (pseudo)random number generator, and \code{type == 1} (Correlated State) corresponds to  creating pairs of indices \code{(i, j)} such that \code{j = (i + 10) \% n}, and assigning a random multiple of  $\pi/2$ to each pair according to the \code{random_ij_phase} vector.
Only the results of \code{type == 0} are presented in \fig{NQSTraining}.
However, the accompanying code includes both cases.

\vspace{0.5em}

\noindent \textbf{Training data:} The training data consist of $50,000$ measurements performed in \protoref{certification} when estimating the shadow overlap. 
In the case of the pseudorandom phase state \code{type == 0}, we perform $n-1$ measurements in the $Z$ basis followed by an $X$ measurement on the remaining qubit. 
For the correlated states in \code{type == 1}, we perform either an $X$ or $Y$ basis measurement on the remaining qubit.
This allows us to access the phase difference between two binary strings $b_0, b_1 \in \{0,1\}$ that differ at one bit. 

In this experiment, the training data is generated by creating a list of binary strings \code{b_list}, along with corresponding indices \code{i_list} indicating the position of the qubit measured in the $X$ or $Y$ basis, basis choices \code{xy_list} ($0$ for $X$ basis, $1$ for $Y$ basis), 
and outcomes \code{o_list}.
For the $X$ basis (\code{xy_list.back() == 0}), 
each entry of \code{o_list} is \code{1} (corresponding to the $\ket{+}$ state) with a probability proportional to \code{(1 + cos(phase0 - phase1)) / 2}, and \code{0} (corresponding to the $\ket{-}$ state) with the remaining probability.
Similarly, for the $Y$ basis (\code{xy_list.back() == 1}), the expected binary outcome is \code{1} (corresponding to the $\ket{i+}$ state) with a probability proportional to \code{(1 + sin(phase0 - phase1)) / 2}, and \code{0} (for the $\ket{i-}$ state) with the remaining~probability.

 \vspace{0.5em}

\noindent \textbf{Training process:} The training process is performed using a neural network model called \code{Model} which is a feedforward neural network with a hidden layer of size \code{h = 4 * n}.
The input to this neural network is a feature vector \code{feat_vec} constructed by concatenating the binary string \code{b0} (with the bit at \code{i_list[r]} set to \code{0}), a one-hot encoding of the index \code{i_list[r]}, and several random phases $\phi(b)$ for both \code{b0} and \code{b1} (the string \code{b1} is obtained by flipping the bit at \code{i_list[r]}) generated for varying choices of seeds. 
The neural network receives a feature vector, \code{feat_vec}, as input. 
This vector is formed by concatenating three components: 
the binary string \code{b0} with its bit at position \code{i_list[r]} set to \code{0}, a one-hot encoded representation of the index \code{i_list[r]}, and a series of random phases $\phi(b)$ for both \code{b0} and the modified string \code{b1}.
The string \code{b1} is derived from \code{b0} by inverting the bit at \code{i_list[r]}. 
These random phases for \code{b0} and \code{b1} are produced using various seed values.
The output of this neural network represents the probabilities \code{p_x} and \code{p_y} of the post-measurement single-qubit state being measured as $0$ or $1$ in the $X$ or $Y$ bases.
These predicted probabilities \code{p_x} and \code{p_y}, stored in \code{prob}, are then used to compute the shadow-based log loss given by 
\begin{center}
\code{log_loss = -outcome * log(prob + EPS) - (1.0 - outcome) * log(1.0 - prob + EPS)} 
\end{center}
for \code{EPS = 1e-10} and  \code{outcome} identifying the corresponding measurement outcome in the training vector \code{o_list}.
The log loss is the negative log-likelihood of the expected outcome \code{outcome} given the probability \code{prob} predicted by the model. 
This gradient of log loss \code{gd} is used in the backpropagation step where the error is backpropagated through the neural network and the weights and biases are updated to minimize the log loss for the given training example.

The training process is performed for a specified number of epochs (\code{num_epoch} = 10), with periodic checks to monitor the validation loss (\code{val_loss}) and updates to the best model (\code{best_predictor}) if a lower validation loss is achieved.

During training, we also report various metrics, including the training log loss (\code{Tlogloss}), validation log loss (\code{Vlogloss}), and the shadow overlap estimated on the training and validation sets (\code{TShadowF} and \code{VShadowF}).
Additionally on a separate test set, the shadow overlap (\code{Shadow_F}) and the fidelity (\code{Fidelity}) are computed by comparing the predicted phases with the true phases.

\vspace{0.5em}

\noindent \textbf{Estimating fidelity of the trained model:} The fidelity is calculated in the following steps: Generate a set of $10,000$ random binary test strings \code{b_test_list}.
For each test string \code{b}, compute the predicted phase using the \code{predict_phase} function.
Compute the true phase for each test string \code{b}.
Calculate the fidelity as the average of \code{exp(i * (predicted_phase - true_phase))} over all test strings, and then take the absolute value squared.
The function \code{predict_phase} finds the phase $e^{i\phi(b)}$ of a given string $b$ using the phase difference predicted by the neural network for a series of adjacent bit strings differing in one bit.
More precisely, given a bit string \code{b}, another string \code{random_init_state} is chosen uniformly at random. 
A path connecting \code{b} to \code{random_init_state} is constructed by flipping indices \code{i} in which \code{b} and \code{random_init_state} differ. 
The function \code{predict_phase_diff} is called to compute the phase difference between two adjacent strings on the path. 
The accumulated phase along the path after each bit flip is stored in~\code{phase}.

\vspace{0.5em}

\noindent \textbf{Estimating shadow overlap of the trained model:} Similar to the subroutine for estimating the fidelity, we generate a set of $10,000$ random binary test strings \code{b_test_list} and corresponding random indices \code{random_i_test_list} indicating the position of the qubit measured in the $X$ or $Y$ basis.  
For each test string \code{b} and random index \code{random_i}, we compute the phase difference between bit string \code{b0} and \code{b1} using the \code{predict_phase_diff} function that applies the trained neural network \code{model}. 
We then calculate the squared magnitude of the difference between the predicted and true phase differences and take the average of the squared magnitude differences over all test strings and indices.
Finally, a linear transformation as in equation \eqref{eq:NormalizedShadowOverlap} is performed to map the average squared magnitude difference to a value between 0 and 1, with higher values indicating better performance.

\vspace{0.5em}

\noindent \textbf{Computing subsystem purity:} We can use the trained model to  estimate the purity $\Tr(\rho^2_A)$ for a subsystem $A$ with $|A|\in\{1,\dots, n\}$ using the expression derived in \eqref{eq:PurityExpression}.
This is achieved by generating $30,000$ pairs of random binary strings \code{b1} and \code{b2} of length \code{n}.
An integer \code{less_than_this} specifies the subsystem size $A$.
For each pair of binary strings \code{b1} and \code{b2}, we create two new binary strings \code{b1_alt} and \code{b2_alt} by swapping the first \code{less_than_this} bits between \code{b1} and \code{b2}.
We then compute the phases \code{phase1}, \code{phase2}, \code{phase1_alt}, and \code{phase2_alt} for these binary strings using the trained \code{predict_phase} function of the neural network model introduced before.
Finally, the purity function is computed by averaging the real part of \code{exp(i * (phase1_alt + phase2_alt - phase1 - phase2))} over~$30,000$ bit strings.

\subsection{Benchmarking noisy quantum devices}\label{sec:NumericalBenchmarking}
In the numerical experiment in \fig{plotsBenchmarking}, we estimate fidelity, shadow overlap, and XEB of two families of noisy quantum states: Haar random and phase states. 
We track these metrics as the strength of the noise is varied. 
We explore both white noise and coherent noise, as detailed below.

\vspace{0.5em}

\noindent \textbf{Haar random states:}
A Haar random state in a \code{d} dimensional Hilbert space is generated.
In the white noise model, the state is subjected to a global depolarizing channel $(1-p) \cdot \ketbra{\psi}{\psi} + p \cdot \frac{\iden}{d}$ with the noise parameter \code{p}. 
In the coherent noise model, the probability amplitudes are randomly changed according to \code{psi[b] = psi[b] + p * (normal(gen) + rmi * normal(gen)) / d}
where \code{normal(gen) + rmi * normal(gen)} is a complex centered Gaussian with variance 1.

In the coherent noise model, the fidelity \code{fid} between the original state \code{psi} and the noisy state \code{psi_noisy} is calculated as the squared overlap between the two states.
We then estimates the shadow overlap and (XEB) using a Monte Carlo approach. 
In each measurement round (for a total of \code{N} rounds), a computational basis state \code{outcome_b} is sampled from the noisy probability distribution \code{prob_noisy}.
The true log-probability \code{true_logp} of this outcome with respect to the original distribution \code{prob} is calculated, and the contribution to the XEB estimator is computed as \code{exp(true_logp) / N}. 
This value is then normalized using 
\code{normalized_XEB = (XEB - 1 / d) / (normalization_XEB - 1 / d)};
where \code{normalization_XEB} corresponds to the sum of squared probabilities of the ideal target distribution \code{prob}.

The shadow overlap (linearly shifted according to \appref{benchmarkingComplexity}) is evaluated by averaging over
\begin{center}
    \code{1.0 * (d - 1) / d * 2 * (local_fid - 0.5) + (1.0 / d)}, 
\end{center}
where the local overlaps \code{local_fid} are calculated as follows:
First, two computational basis states \code{b0} and \code{b1} are randomly drawn to simulate the post-measurement state after measuring the qubits in positions other than \code{rand_i}. 
Then the overlap  \code{local_fid} between the noisy and original post-measurement states are found.

When the white noise model is used, the fidelity is computed as \code{fid = (1 - p) + p / d}.
The local overlaps \code{local_fid}, needed for estimating the shadow overlap, is with probability \code{p} estimated as either 0 or 1 with equal probability, and with probability \code{1 - p} is estimated as 1.
The XEB estimation is done in the same way as before, by sampling from the original distribution \code{prob} and calculating the log-probabilities with respect to this distribution.

\vspace{0.5em}

\noindent \textbf{Phase states with non-uniform amplitudes:} The state in this case is generated by first preparing the product state $$\bigotimes_{i=1}^n \left(\cos(\code{rotation[i]})\cdot\ket{0} + \sin(\code{rotation[i]})\cdot\ket{1} \right),$$
where \code{rotation[i]} are randomly generated with a mean of $\pi/4$ and a standard deviation of $0.01\pi$.
In the next step complex phases $\cos(\code{phase}) + i\sin(\code{phase})$ are added to each probability amplitude where $\code{phase}$ is uniformly chosen from $[0, 2\pi]$.

Two types of noise models are applied to this state: (1) white noise applied as a global depolarizing channel with parameter \code{p} and (2) coherent noise applied via a combination of small Gaussian noise to the phase and magnitude components of probability amplitudes. 
More precisely, we apply a random Gaussian phase shift to each amplitude: \code{psi[b] = psi[b] * (cos(phase) + rmi * sin(phase))}, where $\code{phase} = \frac{\pi}{2}\cdot \code{p} \cdot \mathcal{N}(0,1)$ for the noise parameter \code{p}.
This is followed by adding a random complex term to each amplitude \code{psi[b] = psi[b] + 3 / 4 * p * (normal(gen) + rmi * normal(gen)) / sqrt(d)}, 
where as before \code{d} is the Hilbert space dimension.
The noisy state \code{psi_noisy} is then normalized, and the fidelity, shadow overlap and XEB estimation procedures are similar to the ones explained in the white noise simulation.

\subsection{Optimizing quantum circuits for state preparation}\label{sec:NumericalStatePrep}
The goal of this simulation is to prepare a target state, which is the output of a 1D Instantaneous Quantum Polynomial (IQP) circuit with additional $T$ gates. 
The target state is created by applying a sequence of Hadamard ($H$) gates, a random pattern of $T$ (or  inverse $T$) gates, and controlled-$Z$ ($CZ$) gates between neighboring qubits, followed by another set of Hadamard gates, on the initial state $\ket{0}^{\otimes n}$.
The location of $T$ gates is given by a vector \code{random_T_pattern} initialized with random values $-1$, $0$, or $1$.

A greedy algorithm explores the action space sequentially.
In each iteration, we try all possible actions (\code{act}) and evaluate the fidelity or shadow overlap in the $X$ basis using \code{estimate_fidelity} and  \code{estimate_one_shadow_overlap} functions.
We keep track of the best action that maximizes the fidelity or shadow overlap.
These actions represent the following operations: $(H\otimes H)CZ(H\otimes H)$, $H T H$, or $H T^{-1} H$. 

The fidelity between the prepared state and the target state is estimated by repeatedly generating random bit strings and calculating their complex phases. 
For each bit string, we keep track of \code{how_many_T} which changes as follows: increment by $1$ for a $T$ gate acting on a $\ket{1}$ state, decrease by $1$ (modulo $8$) for an inverse $T$ gate on $\ket{1}$, and increase by $4$ for a $CZ$ gate when both involved qubits are in $\ket{1}$ state.
A similar procedure is performed on the target circuit to generate \code{how_many_T_true}. 
After calculating \code{how_many_T} and \code{how_many_T_true}, the phase difference between the two is calculated as \code{phase_diff = (how_many_T - how_many_T_true + 8) \% 8}.
This phase difference is then used to calculate the fidelity by averaging over $10,000$ repetitions.
If instead of the circuit representation, we use the matrix product representation (MPS) of the target state, we can obtain the complex phase of a given bit string by directly contracting the~MPS.

To calculate the shadow overlap, we first generate a random bit string where all qubits except \code{random_x} are assigned random values ($0$ or $1$).
We then simulate the sequence of actions (\code{seq_action}) on this bit string, but only update the \code{how_many_T} value based on the operations involving the \code{random_x} qubit or its neighbors.
In the query phase of estimating the shadow overlap, we find the post-measurement state of \code{random_X} qubit. 
This can be done by querying the MPS representation of the target state, or in our simple scenario, by directly evaluating $\frac{1}{\sqrt{2}} \ket{0} + \frac{1}{\sqrt{2}}\cdot \exp\left(2\pi i (\code{how_many_T_true}/8)\right)\ket{1}$, where again \code{how_many_T_true} is updated based on the operations involving the \code{random_x} qubit or its neighbors.
The shadow overlap is then computed by averaging over the squared overlap between this state and the prepared single-qubit state (determined by \code{how_many_T}) over $10,000$ repetitions.

\bibliographystyle{alpha}
\bibliography{main}
\end{document}